\documentclass[sigplan,9pt]{acmart}

\acmConference{Online Extended Version}
\acmYear{2018}
\acmISBN{} 
\acmDOI{} 
\startPage{1}

\setcopyright{none}

\usepackage{booktabs}   
\usepackage{subcaption} 
\usepackage{latexsym}
\usepackage{verbatim}
\usepackage{amsthm}
\usepackage{amssymb}
\usepackage{amsmath}
\usepackage{proof}

\newtheorem{theorem}{Theorem}
\newtheorem{proposition}{Proposition}
\newtheorem{definition}{Definition}

\newtheorem{corollary}{Corollary}

\theoremstyle{definition}
\newtheorem{example}{Example}

\newtheorem{digression}{Digression}

\usepackage{amsmath}

\usepackage{mathrsfs}
\usepackage[utf8]{inputenc}
\usepackage[T1]{fontenc}
\usepackage{amsmath}
\usepackage{amsthm}
\usepackage{sansmath}
\usepackage{mathtools}
\DeclareMathAlphabet{\mathpzc}{OT1}{pzc}{m}{it}
\usepackage{rotating}
\usepackage{thm-restate}

\usepackage{amssymb}
\usepackage{bussproofs}
\EnableBpAbbreviations

\usepackage{thrmappendix}
\usepackage{color}
\usepackage{url}
\usepackage{stmaryrd}
\usepackage{xypic}

\newcounter{numberone}

\newenvironment{varenumerate}
{
\begin{list}{\arabic{numberone}.}
{
  \usecounter{numberone}
  \setlength{\itemsep}{0pt}
  \setlength{\topsep}{0pt}
  \setlength{\parsep}{0pt}
  \setlength{\partopsep}{0pt}
  \setlength{\leftmargin}{15pt}
  \setlength{\rightmargin}{0pt}
  \setlength{\itemindent}{0pt}
  \setlength{\labelsep}{5pt}
  \setlength{\labelwidth}{15pt}
}}
{
\end{list}
}

\newcommand{\bnf}{\;::=\;}

\newcommand{\values}{\mathcal{V}}

\newcommand{\valone} {V}
\newcommand{\varone} {x}
\newcommand{\vartwo} {y}
\newcommand{\varthree} {z}
\newcommand{\termone}{e}
\newcommand{\termtwo}{f}
\newcommand{\termthree}{g}

\newcommand{\termfour}{h}

\newcommand{\valtwo} {\mbox{\tt w}}
\newcommand{\valthree}{u}

\newcommand{\envone}{\env}
\newcommand{\envtwo}{\Delta}

\newcommand{\imp}{\vdash}

\newcommand{\ctxone}{\mathcal{C}}

\newcommand{\emptyctx}{\emptyset}

\newcommand{\distone}{\mu}
\newcommand{\disttwo}{\nu}

\newcommand{\signature}{\Sigma}

\newcommand{\abs}[1]{\lambda #1.}

\newcommand{\ocppo}{\ensuremath{\omega\text{-}\mathsf{cppo}}}

\newcommand{\lan}{\langle}
\newcommand{\ran}{\rangle}
\newcommand{\cc}{\cdots}
\newcommand{\hh}{\hdots}

\newcommand{\set}{\mathsf{Set}}
\newcommand{\powerset}{\mathcal{P}}
\newcommand{\distribution}{\mathcal{D}}

\newcommand{\monad}{T}
\newcommand{\unit}{\eta}

\newcommand{\functor}{F}

\newcommand{\defeq}{\triangleq}

\newcommand{\rel}{\mathbf{Rel}}
\newcommand{\howe}[1]{#1^{H}}
\newcommand{\relone}{\mathcal{R}}
\newcommand{\reltwo}{\mathcal{S}}
\newcommand{\vrelone}{\alpha}
\newcommand{\vreltwo}{\beta}
\newcommand{\vrelthree}{\gamma}

\newcommand{\vrelator}{\Gamma}
\newcommand{\vrelatorone}{\vrelator}
\newcommand{\vrelatortwo}{\Delta}
\newcommand{\vidrel}{1}

\newcommand{\true}{\mathsf{true}}
\newcommand{\false}{\mathsf{false}}

\DeclarePairedDelimiter\ket{\lvert}{\rangle}
\newcommand{\dirac}[1]{\ket{#1}}

\newcommand{\cpoleq}{\sqsubseteq}
\newcommand{\lub}{\bigsqcup}

\usepackage{enumitem, hyperref}
\makeatletter
\def\namedlabel#1#2{\begingroup
    #2%
    \def\@currentlabel{#2}%
    \phantomsection\label{#1}\endgroup
}
\makeatother

\newcommand{\bps}{\begin{description}}
\newcommand{\eps}{\end{description}}

\newcommand{\bind}{\mathbin{\scriptstyle{\gg=}}}

\newtheorem{fact}{Fact}

\newcommand{\sem}[1]{| #1 |}

\newcommand{\kleisli}[1]{#1^{*}}
\newcommand{\return}{\mathsf{return}}

\renewcommand{\valone}{v}
\renewcommand{\valtwo}{w}

\renewcommand{\termone}{e}
\renewcommand{\termtwo}{f}
\renewcommand{\termfour}{h}

\newcommand{\catone}{\mathbb{C}}
\newcommand{\TO}[1]{\Rightarrow^{n}}

\newcommand{\open}[1]{#1^{o}}

\newcommand{\substcomp}[3]{#1[#2 := #3]}
\newcommand{\substval}[3]{#1[#2/#3]}

\newcommand{\relator}{\Gamma}

\newcommand{\quantale}{\mathsf{V}}
\newcommand{\tensor}{\otimes}
\newcommand{\qunit}{k}
\newcommand{\join}{\bigvee}
\newcommand{\meet}{\bigwedge}
\newcommand{\qvalone}{a}
\newcommand{\qvaltwo}{b}
\newcommand{\qvalthree}{c}
\newcommand{\qvalfour}{d}
\newcommand{\quantaletwo}{\mathsf{W}}

\newcommand{\mmap}{\multimap}
\newcommand{\two}{\mathsf{2}}

\newcommand{\comp}{\cdot}
\def\tobar{\mathrel{\mkern3mu  \vcenter{\hbox{$\scriptscriptstyle+$}}%
            \mkern-12mu{\to}}}
\newcommand{\torel}{\tobar}
\newcommand{\valfour}{z}

\newcommand{\Vrel}[1]{#1\text{-}\mathbf{Rel}}
\newcommand{\vrel}{\Vrel{\quantale}}

\newcommand{\Vcat}[1]{#1\text{-}\mathsf{Cat}}
\newcommand{\vcat}{\Vcat{\quantale}}
\newcommand{\hausdorff}{H}
\newcommand{\hausdorffsym}{H^s}
\renewcommand{\monad}{T}

\renewcommand{\set}{\mathbf{Set}}
\renewcommand{\powerset}{\mathcal{P}}
\newcommand{\globalstate}{\mathcal{G}}
\renewcommand{\functor}{\monad}
\newcommand{\functorone}{\functor}
\newcommand{\functortwo}{U}

\newcommand{\undeff}{\bot}

\newcommand{\idrel}{1}
\newcommand{\idvrel}{id}
\newcommand{\Monad}{\mathbb{\monad}}

\newcommand{\setone}{X}
\newcommand{\settwo}{Y}
\newcommand{\setthree}{Z}
\newcommand{\setfour}{W}

\newcommand{\typeOne}{\sigma}
\newcommand{\typeTwo}{\tau}

\newcommand{\typeone}{\typeOne}
\newcommand{\typetwo}{\typeTwo}

\newcommand{\typeVar}{t}
\newcommand{\unittype}{\mathsf{unit}}

\newcommand{\sumtype}[2]{\sum\nolimits_{#1}{#2}}
\newcommand{\typevarone}{t}
\newcommand{\recType}[2]{\mu #1.#2}

\newcommand{\op}{\mathbf{op}}

\newcommand{\seq}[2]{\mathbf{let}\ x = #1\ \mathbf{in}\ #2}

\newcommand{\seqy}[2]{\mathbf{let}\ y = #1\ \mathbf{in}\ #2}
\renewcommand{\return}[1]{\mathbf{val}\ #1}

\newcommand{\fold}[1]{\mathbf{fold}\ #1}

\newcommand{\inject}[2]{\lan #1, #2 \ran}
\newcommand{\casesum}[2]{\mathbf{case}\ #1\ \mathbf{of}\ 
			\{\inject{i}{\varone} \to #2\}}

\newcommand{\unfoldin}[2]{\mathbf{case}\ #1\ \mathbf{of}\ 
	\{\fold{\varone} \to #2\}}
\newcommand{\pmfold}[2]{\unfoldin{#1}{#2}}
\newcommand{\casefold}[2]{\unfoldin{#1}{#2}}

\newcommand{\casefoldz}[2]{\mathbf{case}\ #1\ \mathbf{of}\ \{\mathbf{fold}\ \varthree \to #2\}}

\newcommand{\unbang}[2]{\mathbf{case}\ #1\ \mathbf{of}\ 
	\{\bang \varone \to #2\}}
\newcommand{\casebang}[2]{\unbang{#1}{#2}}
\newcommand{\casebangy}[2]{\mathbf{case}\ #1\ \mathbf{of}\ \{ !\vartwo \to #2\}}

\newcommand{\pmbang}[2]{\unbang{#1}{#2}}

\newcommand{\envOne}{\Gamma}
\newcommand{\envTwo}{\Delta}

\renewcommand{\envone}{\envOne}
\renewcommand{\envtwo}{\envTwo}

\newcommand{\substType}[3]{#1[#3/#2]}

\newcommand{\valseq}[2]{#1 \vdash^{\mathsf{v}} #2}
\newcommand{\compseq}[2]{#1 \compimp #2}

\renewcommand{\howe}[1]{#1^{\scriptscriptstyle{H}}}

\newcommand{\graph}{\mathcal{G}}
\newcommand{\strongkleisli}[1]{#1^{*}}
\newcommand{\catunit}{I}
\newcommand{\unitor}{\lambda}

\newcommand{\laxcommutegeq}{\ar @{} [dr] |\geq}
\newcommand{\laxcommute}{\ar @{} [dr] |\leq}

\renewcommand{\powerset}{\mathcal{P}}

\newcommand{\baseone}{s}
\newcommand{\basetwo}{r}
\newcommand{\basethree}{p}
\newcommand{\baseid}{1}
\newcommand{\bang}{{!}}

\newcommand{\toval}[1]{#1^{\scriptscriptstyle \values}}
\newcommand{\toterm}[1]{#1^{\scriptscriptstyle \Lambda}}

\newcommand{\howeimp}{\models}
\newcommand{\ruleHVar}{(\mathsf{H\text{-}var})}
\newcommand{\ruleHReturn}{(\mathsf{H\text{-}val})}
\newcommand{\ruleHSeq}{(\mathsf{H\text{-}let})}
\newcommand{\ruleHApp}{(\mathsf{H\text{-}app})}
\newcommand{\ruleHAbs}{(\mathsf{H\text{-}abs})}

\newcommand{\ruleHInject}{(\mathsf{H\text{-}inj})}
\newcommand{\ruleHCaseSum}{(\mathsf{H\text{-}sum\text{-}cases})}
\newcommand{\ruleHFold}{(\mathsf{H\text{-}fold})}
\newcommand{\ruleHPmFold}{(\mathsf{H\text{-}fold\text{-}cases})}
\newcommand{\ruleHBang}{(\mathsf{H\text{-}bang})}
\newcommand{\ruleHPmBang}{(\mathsf{H\text{-}bang\text{-}cases})}
\newcommand{\ruleHOp}{(\mathsf{H\text{-}op})}

\newcommand{\valimp}{\imp^{\mathsf{v}}}
\newcommand{\compimp}{\imp}

\newcommand{\approxsem}[2]{\sem{#1}_{#2}}
\newcommand{\howeimpcomp}{\models^{\mathsf{c}}}
\newcommand{\howeimpval}{\models^{\mathsf{v}}}
\newcommand{\fun}[1]{\sem{\substcomp{#1}{\varone}{\_}}}

\newcommand{\dual}[1]{#1^{\circ}}

\newcommand{\intervalQuantale}{[0,1]}
\newcommand{\couplings}{\Omega}

\newcommand{\wasserstein}{W}
\newcommand{\wassersteinbot}{W_\mbot}

\newcommand{\acts}{\circ}

\newcommand{\mop}{op}
\newcommand{\qop}{op_{\quantale}}
\newcommand{\changeofbases}{\Pi}
\newcommand{\monadic}[1]{\mathpzc{#1}}

\newcommand{\vsim}{\delta}
\newcommand{\vbisim}{\gamma}
\newcommand{\transitive}[1]{#1^{\scriptscriptstyle T}}
\newcommand{\get}{\mathbf{get}}
\newcommand{\settzero}{\mathbf{set}_{\ell := 0}}
\newcommand{\settone}{\mathbf{set}_{\ell := 1}}
\newcommand{\refine}[1]{\hat{#1}}
\newcommand{\refineimp}{\howeimp}
\newcommand{\refineimpval}{\howeimpval}
\newcommand{\numeral}[1]{\underline{#1}}
\newcommand{\mbot}{\bot}

\newcommand{\qtop}{\rotatebox[origin=c]{180}{$\Bot$}}
\newcommand{\qbot}{\Bot}
\newcommand{\Fuzz}{$\mathsf{Fuzz}$}
\newcommand{\qmmap}{\multimapdot}

\begin{document}

\title[Quantitative Behavioural Reasoning 
for Higher-order Effectful Programs]{Quantitative Behavioural Reasoning 
for Higher-order Effectful Programs: Applicative Distances \\
(Extended Version)}   


\author{Francesco Gavazzo}
\affiliation{
  \institution{Universit\`a di Bologna \& INRIA Sophia Antipolis}            
  \city{Bologna}
  \country{Italy}                    
}

\begin{abstract}
This paper studies quantitative refinements of Abramsky's applicative 
similarity and bisimilarity in the context of a generalisation of \Fuzz, 
a call-by-value $\lambda$-calculus with a linear type system that can 
express program sensitivity, enriched with algebraic operations \emph{\`a la} 
Plotkin and Power.
To do so a general, abstract framework for studying behavioural relations 
taking values over quantales is introduced according to Lawvere's analysis of 
generalised metric spaces. Barr's notion of relator (or lax extension) is 
then extended to quantale-valued relations, adapting and extending results 
from the field of monoidal topology. Abstract notions of 
quantale-valued effectful applicative similarity and bisimilarity 
are then defined and proved to 
be a compatible generalised metric (in the sense of Lawvere) and 
pseudometric, respectively, under mild conditions. 
\end{abstract}



\maketitle

\section{Introduction}

Program preorders and equivalences are fundamental concepts in the 
theory of programming languages since the very birth of the discipline. 
Such notions are usually defined by means of relations 
between program phrases aimed to order or identify programs according to 
their observable \emph{behaviours}, the latter being usually defined 
by means of a primitive notion of observation such as termination to 
a given value. We refer to such 
relations as \emph{behavioural relations}. 
Well-known behavioural relations for higher-order functional languages 
include the \emph{contextual preorder} and \emph{contextual equivalence} 
\cite{Morris/PhDThesis}, \emph{applicative (bi)similarity} \cite{Abramsky/RTFP/1990}, 
and \emph{logical relations} \cite{Reynolds/Logical-relations/1983}.



Instead of asking when two programs $\termone$ and $\termone'$ are 
behaviourally similar or equal, a more informative question 
may be asked, namely how much (behaviourally)
different $\termone$ and $\termone'$ are.
That means that instead of looking at relations relating programs with 
similar or equal behaviours we look at relations assigning programs a numerical  
value representing their \emph{behavioural distance}, i.e. a numerical value 
quantifying the observable differences between their behaviours. 
The question of quantifying observable differences between programs turned 
out to be particularly interesting (and challenging) for effectful higher-order 
languages, where ordinary qualitative
(i.e. boolean-valued) equivalences and preorders are too strong. 
This is witnessed by recent results on behavioural pseudometrics for 
probabilistic $\lambda$-calculi 
\cite{CrubilleDalLago/LICS/2015,CrubilleDalLago/ESOP/2017} as well as  
results on semantics of higher-order languages for differential privacy 
\cite{Pierce/DistanceMakesTypesGrowStronger/2010,GaboardiEtAl/POPL/2017}. 
In the first case one soon realises that programs exhibiting a different 
behaviour only with probability close to zero are fully discriminated 
by ordinary behavioural relations, whereas in the second case relational 
reasoning does not provide any information on how much behavioural differences 
between inputs affect behavioural differences between outputs.

These problems can be naturally addressed by working with quantitative relations 
capturing weakened notions of metric such as \emph{generalised metrics} 
\cite{Lawvere/GeneralizedMetricSpaces/1973} and \emph{pseudometrics}
\cite{steen/CounterexamplesTopology/1995}. It is then natural to ask 
whether and to what extent ordinary behavioural relations can be refined 
into quantitative relations still preserving their nice properties. 
Although easy to formulate, answering such question is far from  
trivial and requires major improvements in the current theory of 
behavioural reasoning about programs.

This paper contributes to answering the above question, and it does so 
by studying the quantitative refinement of Abramsky's 
\emph{applicative similarity} and \emph{bisimilarity} 
\cite{Abramsky/RTFP/1990} for 
higher-order languages enriched with algebraic effects.
Applicative similarity (resp. bisimilarity) is a coinductively defined 
preorder (resp. equivalence) relating programs that exhibit similar 
(resp. equal) extensional behaviours. 
Due to its coinductive nature and to its nice 
properties, applicative (bi)similarity has been studied for a 
variety of calculi, both pure and effectful. 
Notable examples are extensions to nondeterministic 
\cite{Lassen/PhDThesis} and probabilistic 
\cite{DalLagoSangiorgiAlberti/POPL/2014,CrubilleDalLago/ESOP/2014} 
$\lambda$-calculi, and its more recent extension 
\cite{DalLagoGavazzoLevy/LICS/2017} 
to $\lambda$-calculi with algebraic effects \emph{\`a la}  
Plotkin and Power \cite{PlotkinPower/FOSSACS/01}. 
In \cite{DalLagoGavazzoLevy/LICS/2017} an abstract notion of applicative 
similarity is studied for an untyped $\lambda$-calculus enriched with 
a signature of effect-triggering operation symbols. Operation symbols 
are interpreted as algebraic operations with respect to a monad $\monad$ 
encapsulating the kind of effect such operations produce. Examples are 
probabilistic choices with the (sub)distribution monad, and nondeterministic 
choices with the powerset monad. 
The main ingredient used to extend Abramsky's applicative similarity 
is the concept of a \emph{relator} \cite{Barr/LMM/1970,Thijs/PhDThesis/1996}
for a monad $\monad$, i.e.
an abstraction meant to capture the possible ways a relation 
on a set $X$ can be turned into a relation on $\monad X$. 
That allows to define an abstract notion of 
\emph{effectful} applicative similarity parametric in a relator, 
and to prove an abstract precongruence 
theorem stating the resulting notion of applicative similarity is a 
compatible preorder.

The present work originated from the idea of generalising the theory 
developed in \cite{DalLagoGavazzoLevy/LICS/2017} to relations taking values 
over arbitrary quantitative domains (such as the real extended half-line 
$[0,\infty]$ or the unit interval $[0,1]$). 
Such generalisation requires three major improvements in the current 
theory of effectful applicative (bi)similarity:

\begin{enumerate}[wide = 0pt, leftmargin = *]
\item
The first improvement is to move from boolean-valued 
relations to relations taking values on quantitative 
domains such as 
$[0,\infty]$ or $[0,1]$ in such a way that restricting 
these domains to the two element set $\{0,1\}$ (or $\{\false, \true\}$) 
makes the theory collapse to 
the usual theory of applicative (bi)similarity. 
For that we rely on Lawvere's analysis \cite{Lawvere/GeneralizedMetricSpaces/1973} of 
generalised metric spaces and preordered sets as 
enriched categories. 
Accordingly, we replace boolean-valued relations with relations taking 
values over quantales \cite{Rosenthal/Quantales/1990} 
$(\quantale, \leq, \tensor, \qunit)$, i.e. algebraic structures 
(notably complete lattices equipped with a monoid structure) that play the role 
of sets of abstract quantities. Examples of quantales include 
the extended real half-line $([0,\infty], \geq, 0, +)$ ordered by 
the ``greater or equal'' relation $\geq$ and with monoid structure given by 
addition (and its restriction to the unit interval $[0,1]$), and 
the extended real half-line $([0,\infty], \geq, 0, \max)$ with monoid 
structure given by binary maximum (in place of addition), 
as well as any complete Boolean and Heyting algebra. 
This allows to develop an algebra of quantale-valued relations, 
$\quantale$-relations for short, which provides a general framework for studying 
both behavioural relations and behavioural distances (for instance, 
an equivalence $\quantale$-relation instantiates to an ordinary equivalence 
relation on the boolean quantale $(\{\false,\true\}, \leq, \wedge, \true)$, 
and to a pseudometric on the quantale $([0,\infty], \geq, 0, +)$).
\item
The second improvement is the generalisation of the notion of relator to 
quantale-valued relators, i.e. relators acting on relations taking values 
over quantales. Perhaps surprisingly, such 
generalisation is at the heart of the filed of \emph{monoidal topology} 
\cite{Hoffman-Seal-Tholem/monoidal-topology/2014}, a subfield of 
categorical topology aiming to unify ordered, metric, and 
topological spaces in categorical terms. Central to the development of 
monoidal topology is the notion of $\quantale$-relator or $\quantale$-lax 
extension of a monad $\monad$ which, analogously to the notion 
of relator, is a construction lifting $\quantale$-relations on a set $X$ to 
$\quantale$-relations on $\monad X$. 
Notable examples of $\quantale$-relators are obtained from the 
Hausdorff distance (for the powerset monad) 
and from the 
Wasserstein-Kantorovich distance \cite{Villani/optimal-transport/2008} 
(for the distribution monad). 
\item
The third improvement (on which we will expand more in the next paragraph) 
is the development of a \emph{compositional} theory of behavioural 
$\quantale$-relations (and thus of behavioural distances).
As we are going to see, ensuring compositionality in an higher-order 
setting is particularly challenging
due to the ability of higher-order programs to copy their 
input several times, a feature that allows them  to amplify  
distances between their inputs \emph{ad libitum}. 
\end{enumerate}
The result is an abstract theory of behavioural $\quantale$-relations 
that allows to define notions of quantale-valued applicative similarity and 
bisimilarity parametric in a quantale-valued relator. The notions obtained 
generalise the existing notions of real-valued applicative 
(bi)similarity and can be instantiated to concrete calculi to 
provide new notions of applicative (bisimilarity) distance. 
A remarkable 
example is the case of probabilistic $\lambda$-calculi, where to the best of the 
author's knowledge a (non-trivial) applicative distance 
for a universal (i.e. Turing complete) probabilistic $\lambda$-calculus is still 
lacking in the literature (but see Section \ref{section:related-works}). 

The main theorem of this paper states that under suitable conditions on monads and 
quantale-valued relators the abstract notion of quantale-valued applicative 
similarity is a compatible---i.e. compositional---reflexive and transitive 
$\quantale$-relation. 
Under mild conditions such result extends to quantale-valued 
applicative bisimilarity, which is thus proved to be a compatible, 
reflexive, symmetric, and transitive $\quantale$-relation 
(i.e. a compatible pseudometric).

In addition to the concrete results obtained for quantale-valued applicative 
(bi)similarity, the contribution of the present work also relies on
introducing and combining several 
notions and results developed in different fields (such as monoidal topology, 
coalgebra, and programming language theory) 
to build an abstract framework for studying 
quantitative refinements of behavioural relations for higher-order languages 
whose applications go beyond the present study of applicative (bi)similarity.

\paragraph{Compositionality, distance amplification, and linear types}
Once we have understood what is the behavioural distance 
$\vsim(\termone, \termone')$ (which, for the sake of this argument, 
we assume to be a non-negative real number) 
between two programs $\termone$ and $\termone'$, 
it is natural to ask if and how much such distance is modified when $\termone$ and 
$\termone'$ are used inside a bigger program---i.e. a context---$\ctxone[-]$. 
Indeed we would like to reason about the 
distance $\vsim(\ctxone[\termone], \ctxone[\termone'])$ \emph{compositionally}, 
i.e. in terms of the distance $\vsim(\termone, \termone')$. 

Compositionality is at the heart of relational reasoning about program 
behaviours. Informally, compositionality states that observational 
indistinguishability is preserved by language constructors; 
formally, a relation is compositional if it is \emph{compatible} with 
all language constructors, meaning that whenever two programs $\termone$ and 
$\termone'$ are related, then so are the bigger programs $\ctxone[\termone]$ 
and $\ctxone[\termone']$. 

Analogous to the idea that compatible relations are preserved 
by language constructors, we are tempted to define as compatible those 
distances that are not increased by language constructors. 
That is, we would like to say that a behavioural distance $\vsim$ is compatible 
if the distance $\vsim(\ctxone[\termone], \ctxone[\termone'])$ between 
$\ctxone[\termone]$ and $\ctxone[\termone']$ is always bounded by the 
distance $\vsim(\termone, \termone')$, no matter how $\ctxone[-]$ uses 
$\termone$ and $\termone'$. However, we soon realise that such proposal 
cannot work: not only 
how $\ctxone[-]$ uses $\termone$ and $\termone'$ matters, but also 
\emph{how much} it uses them does.
This phenomenon, called \emph{distance amplification} 
\cite{CrubilleDalLago/ESOP/2017}, can be easily observed 
when dealing with probabilistic languages.
Consider the following example for a probabilistic 
untyped $\lambda$-calculus \cite{DalLagoSangiorgiAlberti/POPL/2014}
taken from \cite{CrubilleDalLago/ESOP/2017}. 
Let $I$ be the identity combinator 
and $I \oplus \Omega$ be the program evaluating to $I$ with probability 
$\frac{1}{2}$, and diverging with probability $\frac{1}{2}$. 
Assuming we observe the probability of convergence of a program,
it speaks by itself that we would expect the behavioural distance 
$\vsim(I, I \oplus \Omega)$ between $I$ and $I \oplus \Omega$ to be 
$\frac{1}{2}$. However, it is sufficient to consider a family 
$\{\ctxone_n[-]\}_{n \geq 0}$ of contexts that duplicate their input $n$-times\footnote{ 
For instance 
$\{(\abs{\varone}{\underbrace{(\varone I) \hh (\varone I)}_n})(\lambda y.[-])\}_{n\geq 0}$.} 
to see that any such context amplifies the observable distance between 
$I$ and $I \oplus \Omega$: as $n$ grows, the probability of convergence of
$\ctxone[I \oplus \Omega]$ tends to zero, whereas the one of 
$\ctxone[I]$ remains always equal to one.
During its evaluation, every time the context $\ctxone_n$ evaluates its inputs 
the detected distance between the latter is somehow 
accumulated to the distances previously observed, 
thus exploiting the \emph{linear}---in opposition to classical---nature of 
the act of measuring. Such linearity naturally reflects the monoidal closed 
structure of categories of metric spaces, in opposition with 
the cartesian closed structure characterising `classical' (i.e. 
boolean-valued) observations.    

The above example shows that if we want to reason compositionally 
about behavioural distances, then we have to accept that contexts can amplify 
distances, and thus we should take into account the number of 
times a program accesses its input. 
More concretely, our notion of compatibility 
allows a context $\ctxone[-]$ using its input $\baseone$ times to 
increase the distance $\vsim(\termone, \termone')$ between $\termone$ and $\termone'$, 
but of a factor 
\emph{at most} $\baseone$. That is, the distance 
$\vsim(\ctxone[\termone], \ctxone[\termone'])$ should be bounded by
$\baseone \cdot \vsim(\termone, \termone')$. 
Our main result states that quantale-valued applicative (bi)similarity 
is compatible in this sense. This result allows us to reason about 
behavioural distances compositionally, so that we can e.g. conclude
that the distance between $I$ and $I \oplus \Omega$ 
is indeed $\frac{1}{2}$ (Example \ref{ex:probabilistic-applicative-similarity-distance}).

Reasoning about the number of times programs use (or test) their inputs 
requires a shift from ordinary languages to 
refined languages tracking information about the so-called 
\emph{program sensitivity} 
\cite{Pierce/DistanceMakesTypesGrowStronger/2010,GaboardiEtAl/POPL/2017}. 
The sensitivity of a program is the `law' describing 
how much behavioural differences in outputs are affected by 
behavioural differences in inputs, and thus provides the 
abstraction needed to handle distance amplification.


Our refined language is a generalisation of 
the language \Fuzz\ 
\cite{Pierce/DistanceMakesTypesGrowStronger/2010,GaboardiEtAl/POPL/2017}, 
which we call $\quantale$-\Fuzz. \Fuzz\ is a PCF-like language 
refining standard $\lambda$-calculi by means of a powerful linear type 
system enriched with sensitivity-indexed 
`bang types' that allow to track program sensitivity. 
Despite being parametric with respect to an arbitrary quantale, 
the main difference between $\quantale$-\Fuzz\ and \Fuzz\ 
is that the former is an effectful calculus parametric with 
respect to a signature of (algebraic) operation symbols. 
This allows to consider imperative, nondeterministic, 
and probabilistic versions of \Fuzz, as well as combinations thereof.

\paragraph{Structure of the work} 
After having recalled some necessary mathematical 
preliminaries, we introduce $\quantale$-\Fuzz\ and its monadic operational 
semantics (Section \ref{section:v-fuzz}). We then introduce 
(Section \ref{section:v-relators-and-v-relation-lifting}) 
the machinery of $\quantale$-relators showing how it can be successfully 
instantiated on several examples. 
In Section \ref{section:behavioural-v-relations} we define 
applicative $\vrelator$-similarity, a $\quantale$-relation generalising 
effectful applicative similarity parametric with respect to a$\quantale$-relator 
$\vrelator$, and prove it is a reflexive and 
transitive $\quantale$-relation whose kernel induces an abstract notion of 
applicative similarity. 
Our main theorem states 
that under suitable conditions on the $\quantale$-relator $\vrelator$, 
applicative $\vrelator$-similarity is compatible.
Finally, in Section 
\ref{section:from-applicative-v-similarity-to-applicative-v-bisimilarity}
we define the notion of applicative $\vrelator$-bisimilarity 
and prove that under mild conditions  
such notion is a compatible 
equivalence $\quantale$-relation (viz. a compatible pseudometric). 

\section{Preliminaries}\label{section:preliminaries}

In this section we recall some basic definitions and 
results needed in the rest of the paper. Unfortunately, there is no 
hope to be comprehensive, and thus we assume the reader to be 
familiar with basic domain theory \cite{AbramskyJung/DomainTheory/1994} 
(in particular we assume the notions of $\omega$-complete (pointed) 
partial order, \ocppo\ for 
short, monotone, and continuous functions), basic order 
theory \cite{DaveyPriestley/Book/1990}, and basic category theory 
\cite{MacLane/Book/1971}. 
In particular, for a monoidal category 
$\lan\catone, I, \tensor \ran$ we assume the reader to be familiar with the notion 
of \emph{strong Kleisli triple} 
\cite{MacLane/Book/1971,Kock/StrongMonads/1972} 
$\Monad = \lan \monad, \unit, \strongkleisli{-}\ran$.
We use the notation $\strongkleisli{f}: Z \tensor \monad X \to \monad Y$ for the 
strong Kleisli extension of $f: Z \tensor X \to \monad Y$ (and use the same 
notation for the ordinary Kleisli lifting of $f: X \to \monad Y$, the latter 
being essentially the subcase of $\strongkleisli{-}$ for $Z = I$) and reserve the 
letter $\unit$ to denote the unit of $\Monad$. 
Oftentimes, we refer to a (strong) Kleisli triples as a (strong) monad. 
We denote by $\catone_\Monad$ the Kleisli category of $\Monad$.
Finally, we recall that every monad on $\set$, the category of sets and functions, 
is strong (with respect to the cartesian structure).

We also try to follow the notation used in 
the just mentioned references. As a small difference, we denote 
by $g \comp f$ the composition of $g$ with $f$ rather than by $g \circ f$.

\subsection{Monads and Algebraic Effects}

Following \cite{PlotkinPower/FOSSACS/01} we consider algebraic 
operations as sources of side effects. Syntactically, 
algebraic operations are given via a signature $\signature$ consisting of 
a set of operation symbols (uninterpreted operations) together 
with their arity (i.e. their number of operands). Semantically, 
operation symbols are interpreted as algebraic operations on 
strong monads on $\set$. To any $n$-ary operation symbol $\op \in \signature$ 
and any set $X$ we associate a map $\mop_X: (\monad X)^n \to \monad X$ 
(so that we equip $\monad X$ with a $\signature$-algebra structure) 
such that $\kleisli{f}$ is a parametrised $\signature$-algebra 
(homo)morphis, for any $f:Z \times X \to \monad Y$. Concretely, 
we require 
$\mop_Y(\strongkleisli{f}(z,x_1), \hh, \strongkleisli{f}(z,x_1)) = 
\strongkleisli{f}(z,\mop_X(x_1, \hh, x_n))$
to hold for all $z \in Z, x_i \in \monad Y$.

We also use monads to give operational semantics to $\quantale$-\Fuzz\ 
\cite{DalLagoGavazzoLevy/LICS/2017}. 
Intuitively, a program $\termone$ evaluates to a \emph{monadic value} 
$\monadic{\valone} \in \monad \values$, where $\values$ denotes the set of values. 
For instance, a nondeterministic program evaluates to a \emph{set} of values, 
whereas a probabilistic program evaluates to a \emph{(sub)distribution} of 
values.
Due to the presence of non-terminating programs the evaluation 
of a term is defined as the limit of its ``finite evaluations'', and thus 
we need monads to carry a suitable domain structure. 
Recall that any category $\catone$ is \ocppo-enriched if the hom-set 
$\catone(X,Y)$ carries an 
\ocppo-structure, for all objects $X,Y$, and composition is continuous. 
A (strong) monad $\Monad$ is \ocppo-enriched 
if $\catone_\Monad$ is. In particular, in $\set$ that means that 
we have an \ocppo\ $\lan \monad X, \cpoleq_X, \mbot_X\ran$ for any set $X$. 
In particular, \ocppo-enrichment of $\Monad$ gives the following 
equalities for $g, g_n: X \to \monad Y$ and $f, f_n: Y \to \monad Z$ 
arrows in $\catone$: 
\begin{align*}
  \strongkleisli{(\lub_{n<\omega} f_n)} \comp g            
  & =  \lub_{n<\omega} \strongkleisli{f_n} \comp g,  \\ 
  \strongkleisli{f} \comp (\lub_{n<\omega} g_n)    
  & =  \lub_{n<\omega} (\strongkleisli{f} \comp g_n).
\end{align*}
Since $\quantale$-Fuzz is a call-by-value language, we 
also require the equality $\strongkleisli{f}(z, \mbot_X) = \mbot_Y$, 
for $f: Z \tensor X \to \monad Y$. 

Finally, we say that 
$\Monad$ is $\signature$-\emph{continuous} if satisfies the above conditions and 
operations $\mop_X : (\monad X)^n \to \monad X$ are continuous, meaning 
that for all $\omega$-chains $c_1, \hh, c_n$ in $\monad X$ we have:
$$
\mop_X(\lub c_1, \hh, \lub c_n) = \lub \mop_X(c_1, \hh, c_n).
$$
The reader can consult\cite{PlotkinPower/FOSSACS/01,DalLagoGavazzoLevy/LICS/2017} 
for more details.

\begin{example}\label{ex:monads}
The following are $\signature$-continuous monads:
\begin{enumerate}[wide = 0pt, leftmargin = *]
  \item The partiality monad $(-)_\mbot$ mapping a set $X$ to 
    $X_\mbot \defeq X + \{\mbot_X\}$. We give $X_\mbot$ an 
    \ocppo\ structure via $\cpoleq_X$ defined by 
    $\monadic{x} \cpoleq_X \monadic{y}$ if and only if $\monadic{x} = \mbot_X$ 
    or $\monadic{x} = \monadic{y}$. We equip the function space 
    $X \to Y_\bot$ with the pointwise order induced by $\cpoleq$.
  \item The powerset monad mapping a set to its powerset. 
    The unit maps an element $x$ to  
    $\{x\}$, whereas $\strongkleisli{f}: Z \times \powerset(X) \to \powerset(Y)$ 
    is defined by
    $\kleisli{f}(z,\mathpzc{X}) \defeq \bigcup_{x \in \mathpzc{X}} f(z,x)$, 
    for $f: Z \times X \to \powerset(Y)$, $\mathpzc{X} \subseteq X$, 
    and $z \in Z$. 
    We give $\powerset(X)$ an \ocppo\ structure via 
    subset inclusion $\subseteq$ and order the function space 
    $X \to \powerset(Y)$ with the pointwise order induced 
    by $\subseteq$. Finally,
    we consider the signature $\signature = \{\oplus\}$ consisting 
    of a single binary operation symbol for pure 
    nondeterministic choice and interpret it as set-theoretic union.
  \item The discrete \emph{subdistribution} monad $\distribution_{\leq 1}$ 
    mapping a set $X$ to $\distribution(X_\mbot)$, where $\distribution$ 
    denotes the discrete \emph{full} distribution monad. 
    The unit of $\distribution$ maps an element 
    $x$ to the Dirac distribution $\dirac{x}$ on it, whereas 
    the strong Kleisli extension 
    $\strongkleisli{f}: Z \times \distribution X \to \distribution Y$ 
    of $f: Z \times X \to \distribution Y$ is defined 
    by $\strongkleisli{f}(z,\mu)(y) \defeq \sum_{x \in X} \mu(x) \cdot f(z,x)(y)$.
    On $\distribution(X_\mbot)$, define the order $\cpoleq_X$ by 
    $\mu \cpoleq_X \nu$ if and only if $\forall x \in X.\ \mu(x) \leq \nu(x)$ holds.
    The pair $(\distribution(X_\mbot), \cpoleq_X)$ forms an \ocppo, 
    with bottom element given by the Dirac distribution on $\mbot_X$ 
    (the distribution modelling 
    the always zero subdistribution). The \ocppo\ structure lifts to 
    function spaces pointwisely.
    Finally, consider the signature
    $\signature \defeq \{\oplus_p \mid p \in \mathbb{Q},\ 0 < p < 1\}$ 
    whose interpretation on the subdistribution monad is defined by 
    $(\mu \oplus_p \nu)(x) \defeq p \cdot \mu(x) + (1-p) \cdot \nu(x)$. 
    Restricting to $p \defeq \frac{1}{2}$ we obtain fair probabilistic 
    choice $\oplus$.
  \item The partial global state monad $\globalstate_\bot$ is obtained from 
    the partiality monad and the global state monad; it maps 
    a set $X$ to $(S \times X)_\bot^X$. 
    The global state monad $\globalstate$ maps a set $X$ to 
    $(S \times X)^S$. Since ultimately a location 
    stores a bit we take $S \defeq \{0,1\}^\mathcal{L}$, where 
    $\mathcal{L}$ is a  
    set of (public) location names. 
    We can give an \ocppo\ structure to $\globalstate_\bot X$ by extending the 
    order of point $1$ pointwise. 
    We consider the signature $\signature_{\mathcal{L}} \defeq 
    \{\get, \settzero, \settone \mid \ell \in \mathcal{L}\}$ and 
    interpret operations in $\signature_{\mathcal{L}}$ on 
    $\globalstate$ as follows: 
    \begin{align*}
    \settzero(f)(b)   &  \defeq f(b[\ell := 0]), \\
    \settone(f)(b)    & \defeq  f(b[\ell := 1]), \\
    \get(f,g)(b)      & \defeq 
      \begin{cases} 
        f(b) &\text{if } b = 0, \\
        g(b) &\text{if } b = 1,
      \end{cases}
    \end{align*}
    where for $b \in S$, 
    $b[\ell := x](\ell) \defeq x$ and $b[\ell := x](\ell') \defeq b(\ell')$, 
    for $\ell' \neq \ell$.
\end{enumerate}
\end{example}

\subsection{Relations, Metrics, and Quantales}
We now recall basic notions on quantales \cite{Rosenthal/Quantales/1990} and 
quantale-valued relations ($\quantale$-relations) along the lines of 
\cite{Lawvere/GeneralizedMetricSpaces/1973}. The 
reader is referred to the monograph
\cite{Hoffman-Seal-Tholem/monoidal-topology/2014} for 
an introduction.
 
\begin{definition}\label{def:quantale}
  A (unital) quantale $(\quantale, \leq, \tensor, \qunit)$, 
  $\quantale$ for short,
  consists of a monoid $(\quantale, \tensor, \qunit)$ and a sup-lattice 
  $(\quantale, \leq)$ satisfying the following distributivity laws:
  \begin{align*}
    \qvaltwo \tensor \join_{i\in I} \qvalone_i  
    &= \join_{i \in I} (\qvaltwo \tensor \qvalone_i), 
    & 
    (\join_{i \in I} \qvalone_i) \tensor \qvaltwo
    &= \join_{i \in I} (\qvalone_i \tensor \qvaltwo).
  \end{align*}
  The element $\qunit$ is called unit, whereas  
  $\tensor$ is called multiplication of the quantale.
  Given quantales $\quantale, \quantaletwo$, a 
  \emph{quantale lax morphism} is a \emph{monotone} map 
  $h: \quantale \to \quantaletwo$  
  satisfying the following inequalities:
  \begin{align*}
  \ell 
  &\leq h(\qunit),
  &
  h(\qvalone) \tensor h(\qvaltwo)
  &\leq 
  h(\qvalone \tensor \qvaltwo),  
  \end{align*}
  where $\ell$ is the unit of $\quantaletwo$.
\end{definition}
It is easy to see that 
$\tensor$ is monotone in both arguments.
We denote top and bottom elements of a quantale by 
$\qtop$ and $\qbot$, respectively.  
Moreover, we say that a quantale is commutative if 
its underlying monoid is, and it is non-trivial if 
$\qunit \neq \qbot$.
Finally, we observe that for any $\qvalone \in \quantale$, 
the map $\qvalone \tensor (-): \quantale \to \quantale$ 
has a right adjoint $\qvalone \qmmap (-): \quantale \to \quantale$
which is uniquely determined by:
$$
\qvalone \tensor \qvaltwo \leq \qvalthree \iff
\qvaltwo \leq \qvalone \qmmap \qvalthree.
$$
From now on we tacitly assume quantales to be commutative 
and non-trivial.

\begin{example}\label{ex:quantales}
The following are examples of quantales:
\begin{enumerate}[wide = 0pt, leftmargin = *]
  \item The \emph{boolean quantale}
  $(\two, \leq, \wedge, \true)$ where $\two = \{\true, \false\}$
  and $\false \leq \true$. 
  \item The extended real half-line $([0, \infty], \geq, +,0)$ 
    ordered by 
    the ``greater or equal'' relation $\geq$ and 
    extended\footnote{We extend 
    ordinary as follows: 
    $x + \infty \defeq \infty \defeq \infty + x$.} 
    addition as monoid multiplication. 
    We refer to such quantale as the \emph{Lawvere quantale}.
    Note that in the Lawvere quantale the bottom element is $\infty$, 
    the top element is $0$, whereas infimum and supremum are defined 
    as $\sup$ and $\inf$, respectively. Notice also that $\qmmap$ is 
    truncated subtraction.  
  \item Replacing addition with maximum in the Lawevere 
    quantale we obtain the \emph{ultrametric Lawvere quantale}
    $([0,\infty], \geq, \max, 0)$, which  
    has been used to study generalised ultrametric spaces 
    \cite{Rutten/ultrametricSpaces/1996} (note that 
    in the ultrametric Lawvere quantale monoid multiplication and binary 
    meet coincide).
  \item Restricting the Lawvere quantale to the unit interval 
    we obtain the \emph{unit interval quantale} $([0,1], \geq, +, 0)$, 
    where $+$ stands for truncated addition.
  \item A left continuous \emph{triangular norm} ($t$-norm for short) 
    is a binary operator 
    $*: [0,1] \times [0,1] \to [0,1]$ that induces a quantale 
    structure over the complete lattice $([0,1], \leq)$ in 
    such a way that the quantale is commutative.
    Examples $t$-norms are:
    \begin{enumerate}
      \item The \emph{product $t$-norm}: $x *_p y \defeq x \cdot y$.
      \item The \emph{\L{}ukasiewicz $t$-norm}: $x *_l y \defeq \max\{x + y - 1, 0\}$.
      \item The \emph{G\"{o}del $t$-norm}: $x *_g y \defeq \min\{x,y\}$.
    \end{enumerate}
\end{enumerate}
\end{example}

In all quantales of Example \ref{ex:quantales} the unit $\qunit$ 
coincide the top element (i.e. $\qunit = \qtop$). 
Quantales with such property are called \emph{integral quantales}, and 
are particularly well-behaved. 
For instance, in an 
integral quantale $\qvalone \tensor \qvaltwo$ is a lower bound 
of $\qvalone$ and $\qvaltwo$ (and thus 
$\qvalone \tensor \bot = \bot$, for any $\qvalone \in \quantale$). 
From now on we tacitly assume quantales to be integral.

\paragraph{$\quantale$-relations}
The notion of $\quantale$-relation, for a quantale $\quantale$, 
provides an abstraction of the notion relation that subsumes 
both the qualitative---boolean valued---and the 
quantitative---real valued---notion of 
relation, as well as the associated notions of equivalence and 
(pseudo)metric. Moreover, sets and $\quantale$-relations 
form a category which, thanks to the quantale structure of $\quantale$, 
behaves essentially like $\rel$, the category of sets and relations. 
That allows to develop an algebra of $\quantale$-relations on the same
line of the usual algebra of relations.

Formally, for a quantale $\quantale$, a $\quantale$-relation 
$\vrelone: X \torel Y$ between sets 
$\setone$ and $\settwo$
is a function $\vrelone: \setone \times \settwo \to \quantale$.
For any set $\setone$ we can define the identity 
$\quantale$-relation $\idvrel_{\setone} : \setone \torel \setone$ 
mapping diagonal elements $(x,x)$ to $\qunit$, and all other 
elements to $\qbot$. Moreover, for $\quantale$-relations 
$\vrelone: \setone \torel \settwo$ and 
$\vreltwo: \settwo \torel \setthree$, 
we can define the composition 
$\vreltwo \comp \vrelone: \setone \torel \setthree$ 
by the so-called `matrix multiplication formula':
$$
(\vreltwo \comp \vrelone)(x,z) \defeq  
  \join_{y \in Y} \vrelone(x,y) \tensor \vreltwo(y,z).
$$
Composition of $\quantale$-relations is associative, 
and $\idvrel$ is the unit of composition. 
As a consequence, we have that sets and 
$\quantale$-relations form a category, called $\vrel$. 
$\vrel$ is a monoidal category with unit given by the one-element 
set and tensor product given by cartesian product of sets with 
$\vrelone \tensor \vreltwo : X \times Y 
\torel X' \times Y'$
defined pointwise, for $\vrelone: X \torel X'$ and $\vreltwo: Y \to Y'$. 
Moreover, for all sets $\setone, \settwo$, the hom-set 
$\vrel(\setone, \settwo)$ inherits a complete lattice structure from 
$\quantale$ according to the pointwise order. Actually, 
the whole quantale structure of $\quantale$ is inherited, in the 
sense that $\vrel$ is a \emph{quantaloid} 
\cite{Hoffman-Seal-Tholem/monoidal-topology/2014}. In particular, 
for all $\quantale$-relations 
$\vrelone: \setone \torel \settwo$, 
$\vreltwo_i : \settwo \torel \setthree$ ($i \in I)$, 
and $\vrelthree: \setthree \torel \setfour$ we have 
the following distributivity laws:
  \begin{align*}
  \vrelthree \comp (\join_{i \in I} \vreltwo_i) 
  &= \join_{i \in I} (\vrelthree \comp \vreltwo_i),
  &
  (\join_{i \in I} \vreltwo_i) \comp \vrelone 
  &= \join_{i \in I} (\vreltwo_i \comp \vrelone).
  \end{align*}

There is a bijection 
$\dual{-}: \vrel(\setone, \settwo) \to \vrel(\settwo, \setone)$ that 
maps each $\quantale$-relation $\vrelone$ to its dual $\dual{\vrelone}$ defined by
$\dual{\vrelone}(y,x) \defeq \vrelone(x,y)$. 
It is straightforward to see that 
$\dual{-}$ is monotone 
(i.e. $\vrelone \leq \vreltwo$ implies $\dual{\vrelone} \leq \dual{\vreltwo}$), 
idempotent (i.e. $\dual{(\dual{\vrelone})} = \vrelone$), and preserves the identity 
relation (i.e. $\dual{\idvrel} = \idvrel$). Moreover, since $\quantale$ 
is commutative we also have the equality 
$\dual{(\vreltwo \comp \vrelone)} = \dual{\vrelone} \comp \dual{\vreltwo}$. 

Finally, we define the graph functor $\graph$ from $\set$ to $\vrel$ acting 
as the identity on sets and mapping each function $f$ to its graph (so that 
$\graph(f)(x,y)$ is equal to $\qunit$ if $y = f(x)$, and $\qbot$ otherwise). 
It is easy to see that since $\quantale$ is non-trivial $\graph$ is 
faithful. In light of this observation we will use the notation 
$f : \setone \to \settwo$ in place of $\graph(f) : \setone \torel \settwo$ 
in $\vrel$.

A direct application of the definition of composition gives the equality:
$$
(\dual{g} \comp \vrelone \comp f)(x,w) = \vrelone(f(x),g(w))
$$
for $f: \setone \to \settwo$, $\vrelone: \settwo \torel \setthree$, and 
$g: \setfour \to \setthree$. Moreover, it is useful to keep in mind the 
following adjunction rules \cite{Hoffman-Seal-Tholem/monoidal-topology/2014}
(for $\vrelone, \vreltwo, \vrelthree$ 
$\quantale$-relations, and  $f,g$ functions with appropriate source and 
target):
\begin{align*}
g \comp \vrelone \leq \vreltwo      
  &\iff \vrelone \leq \dual{g} \comp \vreltwo, \\
\vreltwo \comp \dual{f} \leq \vrelthree  
  &\iff \vreltwo \leq \vrelthree \comp f.
\end{align*}
The above inequalities turned out to be useful in making 
pointfree calculations with $\quantale$-relations. 
In particular, we can use \emph{lax} 
commutative diagrams of the form
\begin{center}
\(
\vcenter{
\xymatrix{
\laxcommute
\setone     
\ar[r]^{f} 
\ar[d]_{\vrelone}|@{|}  &  
\setthree   
\ar[d]^{\vreltwo}|-*=0@{|}
\\
\settwo   
\ar[r]_{g}   &  
\setfour  } }
\)
 \end{center}
as diagrammatic representation for the inequation
$g \comp \vrelone \leq \vreltwo \comp f$. By adjunction rules, 
the latter is equivalent to $\vrelone \leq \dual{g} \comp \vreltwo \comp f$,
which pointwisely gives the following 
generalised non-expansiveness condition\footnote{Taking $f = g$ generalised 
non-expansiveness expresses  
monotonicity of $f$ in the boolean quantale, and non-expansiveness of 
$f$ in the Lawvere quantale and its variants (recall that when we 
instantiate $\quantale$ as e.g. the Lawvere quantale we have to 
\emph{reverse} inequalities).}:
$
\forall (x,y) \in \setone \times \settwo.\ 
\vrelone(x,y) \leq \vreltwo(f(x), g(y)).
$

Among $\quantale$-relations we are interested in those 
generalising equivalences and pseudometrics. 

\begin{definition}
  A $\quantale$-relation $\vrelone: \setone \torel \setone$ 
  is \emph{reflexive} if $\idvrel_{\setone} \leq \vrelone$, 
  \emph{transitive} if $\vrelone \comp \vrelone \leq \vrelone$, and 
  \emph{symmetric} if $\vrelone \leq \dual{\vrelone}$. 
\end{definition}

Pointwisely, reflexivity, transitivity, and symmetry give the 
following inequalities:
$$
\qunit \leq \vrelone(x,x),  \quad
\vrelone(x,y) \tensor \vrelone(y,z) \leq \vrelone(x,z), \quad
\vrelone(x,y) \leq \vrelone(y,x),
$$
for all $x,y,z \in X$.
We call a reflexive and transitive $\quantale$-relation a 
$\quantale$\emph{-preorder} or \emph{generalised metric} 
\cite{Lawvere/GeneralizedMetricSpaces/1973,BonsangueBreguelRutten/GeneralisedMetricSpaces/1998}, and a reflexive, symmetric, and 
transitive $\quantale$-relation a $\quantale$\emph{-equivalence} or \emph{pseudometric}.

\begin{example}
\begin{enumerate}[wide = 0pt, leftmargin = *]
  \item We see that $\two$-$\rel$ is the ordinary category $\rel$ of 
    sets and relations. Moreover, instantiating reflexivity and transitivity 
    on the boolean quantale, we recover the usual notion of preorder. If 
    we additionally require symmetry, then we obtain the usual notion 
    of equivalence relation.
  \item On the Lawvere quantale transitivity gives:
    $$
    \inf_y \vrelone(x,y) + \vrelone(y,z) \geq \vrelone(x,z),
    $$
    which means $\vrelone(x,z) \leq \vrelone(x,y) + \vrelone(y,z)$, for any 
    $y \in X$. That is, in the Lawvere quantale transitivity gives exactly 
    the triangle inequality. Similarly, reflexivity gives 
    $0 \geq \vrelone(x,x)$, i.e. $\vrelone(x,x) = 0$. 
    If additionally $\vrelone$ is symmetric, then we recover the usual notion 
    of pseudometric \cite{steen/CounterexamplesTopology/1995}. 
  \item Analogously to point $2$, if we consider the ultrametric 
    Lawvere quantale, we recover the ultrametric variants of the above notions.  
\end{enumerate}
\end{example}

\begin{digression}[$\quantale$-categories]\label{digression:v-categories}
Lawvere introduced generalised metric spaces in his seminal paper 
\cite{Lawvere/GeneralizedMetricSpaces/1973} as pairs $(X, \vrelone)$ 
consisting of a set $X$ and a generalised metric $\vrelone: X \torel X$ 
over the Lawvere quantale. 
Generalising from the Lawvere quantale to an arbitrary quantale $\quantale$ 
we obtain the so-called $\quantale$-categories 
\cite{Hoffman-Seal-Tholem/monoidal-topology/2014}. 
In fact, a $\quantale$-category $(X, \vrelone)$ is nothing but a category 
enriched over $\quantale$ regarded as a 
bicomplete monoidal category. The notion of $\quantale$-enriched 
functor precisely instantiates as non-expansive map between 
$\quantale$-categories, so that one can consider the category 
$\vcat$ of $\quantale$-categories and $\quantale$-functors. The 
category $\vcat$ has a rich structure. In particular, it is monoidal 
closed category. Given $\quantale$-categories $(X, \vrelone), (Y, \vreltwo)$, 
their exponential $(Y^X, [\vrelone, \vreltwo])$ 
is defined by 
$$
[\vrelone, \vreltwo](f,g) \defeq \meet_{x \in X} \vreltwo(f(x),g(x))
$$ 
(cf. with the usual, real-valued, sup-metric on function spaces), 
whereas their tensor product $(X \times Y, \vrelone \tensor \vreltwo)$ is defined 
pointwise.

Although in this work we will not work with $\quantale$-categories 
(we will essentially work in $\vrel$), it is sometimes useful to think in 
terms of $\quantale$-categories for `semantical intuitions'.
\end{digression}

\paragraph{Operations} 
For a signature $\signature$, we need to specify how operations in $\signature$ 
interact with $\quantale$-relations (e.g. how they modify distances), 
and thus how they interact with quantales. 

\begin{definition}\label{def:signature-quantale}
  Let $\signature$ be a signature. A $\signature$-quantale is a 
  quantale $\quantale$ 
  equipped with \emph{monotone} operations
  $\qop: \quantale^n \to \quantale$,
  for each $n$-ary operation $\op \in \signature$, 
  satisfying the following inequalities:
  \begin{align*}
    \qunit 
    &\leq \qop(\qunit, \hh, \qunit), \\
    \qop(\qvalone_1, \hh, \qvalone_n) \tensor 
    \qop(\qvaltwo_1, \hh, \qvaltwo_n) 
    &\leq \qop(\qvalone_1 \tensor \qvaltwo_1, 
    \hh, \qvalone_n \tensor \qvaltwo_n).
  \end{align*}
\end{definition}

\begin{example}\label{ex:quantale-operations}
  Both in the Lawvere quantale and in the unit 
  interval quantale we can interpret 
  operations $\oplus_p$ from Example \ref{ex:monads} 
  as probabilistic choices: $x \oplus_p y \defeq p \cdot x + (1-p) \cdot y$. 
  In general, for a quantale $\quantale$ we can 
  interpret $\qop(\qvalone_1, \hh, \qvalone_n)$ 
  both as $\qvalone_1 \tensor \hh \tensor \qvalone_n$ and 
  $\qvalone_1 \wedge \hh \wedge \qvalone_n$. 
\end{example}


\paragraph{Change of Base Functors}
We model sensitivity of a program 
as a function 
giving the `law' describing how distances between inputs 
are modified by the program. The notion of \emph{change of base functor} 
provides a mathematical abstraction to model the concept of sensitivity 
with respect to an arbitrary quantale. 

\begin{definition}\label{def:change-of-base-functor}
A change of base functor \cite{Hoffman-Seal-Tholem/monoidal-topology/2014}, 
CBF for short,
between quantales $\quantale, \quantaletwo$ 
is a lax quantale morphism $h :\quantale \to \quantaletwo$
(see Definition \ref{def:quantale}). If $\quantale = \quantaletwo$ 
we speak of change of base \emph{endofunctors} (CBEs, for short), 
and denote them by $\baseone, \basetwo \hh$. Clearly, every CBE $\baseone$ 
is also a CBF.
\end{definition}
The action $h \acts \vrelone$ of a CBF $h: \quantale \to \quantaletwo$ 
on a $\quantale$-relation $\vrelone: X \torel Y$ is defined by 
$h \acts \vrelone(x,y) \defeq h(\vrelone(x,y))$
(to improve readability we omit brackets).
Note that since $\quantale$ is integral, CBFs preserve 
the unit.

\begin{example}\label{ex:change-of-base-functor}
\begin{enumerate}[wide = 0pt, leftmargin = *]
\item Extended\footnote{We extend real-valued multiplication by:
  $0 \cdot \infty \defeq 0 \defeq \infty \cdot 0$, 
  $\infty \cdot x  \defeq \infty \defeq x \cdot \infty$.
 } real-valued multiplication $c \cdot -$, for $c \in [0,\infty]$, 
 is a CBE on the Lawvere quantale. Functions $c \cdot -$ act as CBEs 
 also on the unit interval quantale (where multiplication is meant to 
  be truncated).  
\item Both in the Lawvere quantale and in the unit interval 
  quantale, polynomials $P$ such that $P(0) = 0$ are
  CBEs. 
\item Define CBEs
  $n, \infty : \quantale \to \quantale$, for $n < \omega$ by 
  $0(\qvalone) \defeq \qunit$, $(n+1)(\qvalone) \defeq 
  \qvalone \tensor n(\qvalone)$, and $\infty(\qvalone) \defeq \qbot$. 
  Note that $\baseid$ acts as the identity function. 
\end{enumerate}
\end{example}

Finally, we observe that the action of CBFs on a $\quantale$-relation
obeys the following laws:
\begin{align*}
  (h \comp h')(\vrelone)  
  &= h \acts (h' \acts \vrelone), 
  \\
  (h \acts \vrelone) \comp (h \acts \vreltwo) 
  &\leq h \acts (\vrelone \comp \vreltwo).
\end{align*}

\begin{digression}\label{digression:change-of-base-functor}
We saw that $\quantale$-categories generalise the notions of 
metric space and ordered set, and that the notion of 
$\quantale$-functor generalises the notions of monotone and 
non-expansive function. However, when dealing with metric spaces 
besides non-expansive functions, a prominent role is played by 
\emph{Lipshitz continuous} functions. Given metric spaces 
$(X, d_X)$ and $(Y, d_Y)$, a function $f: X \to Y$ is called 
\emph{$c$-continuous}, for $c \in \mathbb{R}_{\geq 0}$ if the 
inequation $c \cdot d_X(x,x') \geq d_Y(f(x),f(x'))$ holds, for 
all $x,x' \in X$. Example \ref{ex:change-of-base-functor} shows 
that multiplication $c \cdot -$ by a real number $c$ is a 
change of base endofunctor on the Lawvere quantale, meaning that 
using CBEs we can generalise the notion of Lipshitz-continuity 
to $\quantale$-categories. In fact, easy calculations show that 
for any $\quantale$-category $(X, \vrelone)$ and any CBE $\baseone$ 
on $\quantale$, $(X, \baseone \acts \vrelone)$ is a $\quantale$-category. 
In particular, we can define $\baseone$-continuous functions from 
$(X, \vrelone)$ to $(Y, \vreltwo)$ as $\quantale$-functors from 
$(X, \baseone \acts \vrelone)$ to $(Y, \vreltwo)$. That is, 
we say that a function $f: X \to Y$ is $\baseone$-continuous if 
$\baseone \acts \vrelone(x,x') \leq \vreltwo(f(x),f(x'))$ holds, 
for all $x,x' \in X$.
\end{digression}

We conclude this section with the following result on 
the algebra of CBEs.

\begin{lemma}\label{lemma:algebra-change-of-base-functors}
Let $\quantale$ be a $\signature$-quantale. CBEs
are closed under the following operations (where $\op \in \signature$):
\begin{align*}
(\baseone \tensor \basetwo)(\qvalone) 
&\defeq \baseone(\qvalone) \tensor \basetwo(\qvalone), \\
(\basetwo \comp \baseone)(\qvalone)
&\defeq \basetwo(\baseone(\qvalone)), \\
(\baseone \wedge \basetwo)(\qvalone) &= \baseone(\qvalone) \wedge \baseone(\qvaltwo), \\
\qop(\baseone_1, \hh, \baseone_n)(\qvalone)
&\defeq \qop(\baseone_1(\qvalone), \hh, \baseone_n(\qvalone)).
\end{align*}
\end{lemma}

\section{The $\quantale$-fuzz Language }\label{section:v-fuzz}

As already observed in the introduction, when dealing with 
behavioural $\quantale$-relations a crucial parameter in amplification phenomena 
is program sensitivity. To deal with such parameter we 
introduce $\quantale$-fuzz, a higher-order effectful language 
generalising \Fuzz\ \cite{GaboardiEtAl/POPL/2017}. As \Fuzz, 
$\quantale$-\Fuzz\ is characterised by a powerful type system 
inspired by \emph{bounded linear logic} \cite{DBLP:journals/tcs/GirardSS92}
giving syntactic information on program sensitivity. 

\paragraph{Syntax} 
$\quantale$-fuzz is a \emph{fine-grained call-by-value} 
\cite{Levy/InfComp/2003} \emph{linear} $\lambda$-calculus with  
finite sum and recursive types. 
In particular, we make a formal distinction between values 
and computations (which we simply refer to as terms), and use syntactic 
primitives to returning values ($\return$) and sequentially compose 
computations (via a $\mathbf{let}$-$\mathbf{in}$ constructor). 
The syntax of $\quantale$-\Fuzz\ is parametrised over a 
signature $\signature$ of operation symbols, a $\signature$-quantale $\quantale$, 
and a family $\Pi$ of CBEs. From now on we 
assume $\signature$, $\quantale$, and $\Pi$ to be fixed. Moreover, 
we assume $\Pi$ to contain at least CBEs $n, \infty$ 
in Example \ref{ex:change-of-base-functor} and 
to be closed under operations in Lemma \ref{lemma:algebra-change-of-base-functors}.
Types, values, and terms of $\quantale$-\Fuzz\ are defined in Figure 
\ref{fig:types-terms-values-v-fuzz}, where $\typevarone$ denotes a type variable, 
$I$ is a \emph{finite} set (whose elements are denoted by 
$\hat \imath, \hat \jmath, \hh$), and $\baseone$ is in $\Pi$.

\begin{figure}[htbp]
{\hrulefill}
\begin{align*}
\typeone 
&\bnf \typevarone 
  \mid \sumtype{i \in I}{\typeone_i}
  \mid \typeone \mmap \typeone 
  \mid \recType{\typevarone}{\typeone}
  \mid \bang_{\baseone} \typeone.
\\
\valone 
&\bnf \varone
  \mid \abs{\varone}{\termone} 
  \mid \inject{\hat \imath}{\valone}
  \mid \fold{\valone}     
  \mid \bang \valone.  
\\   
\termone  
&\bnf \return{\valone} 
  \mid \valone \valone
  \mid \casesum{\valone}{\termone_i}
  \mid \seq{\termone}{\termone}
  \\
& \text{ }\mid \casebang{\valone}{\termone}
  \mid \casefold{\valone}{\termone} 
  \mid \op(\termone, \hh, \termone).
\end{align*}{\hrulefill}
\caption{Types, values, and terms of $\quantale$-\Fuzz.}
\label{fig:types-terms-values-v-fuzz}
\end{figure}

Free and bound variables in terms and values are defined as usual.
We work with equivalence classes of terms modulo renaming and tacitly 
assume conventions on bindings. 
Moreover, we denote by $\substval{\valtwo}{\valone}{\varone}$ and 
$\substcomp{\termone}{\varone}{\valone}$ the \emph{value} and 
\emph{term} obtained by capture-avoiding substitution 
of the \emph{value} $\valone$ for $\varone$ in $\valtwo$ and $\termone$, 
respectively (see \cite{DalLagoGavazzoLevy/LICS/2017} for details). 

Similar conventions hold for types. In particular, we denote 
by $\typeone[\typetwo/\typevarone]$ the result of capture-avoiding substitution 
of type $\typetwo$ for the type variable $\typevarone$ in $\typeone$. Finally, 
we write \textbf{0} for the empty sum type, \textbf{1} for 
$\textbf{0} \mmap \textbf{0}$, and $\mathsf{nat}$ for $\mu t.\textbf{1} + t$. 
We denote the numeral $n$ by $\numeral{n}$. 

$\quantale$-\Fuzz\ type system is essentially based on 
judgments of the form
$
\varone_1 :_{\baseone_1} \typeone_1, \hh, 
\varone_n :_{\baseone_n} \typeone_n \imp \termone: \typeone,
$
where $\baseone_1, \hh, \baseone_n$ are CBEs.
The informal meaning of such judgment is that on input $\varone_i$ ($i \leq n$), 
the term $\termone$ has sensitivity $\baseone_i$. That is, 
$\termone$ amplifies the (behavioural) distance between two input
values $\valone_i, \valtwo_i$ 
of \emph{at most} a factor $\baseone_i$; symbolically, 
$\baseone_i \acts \vrelone(\valone_i, \valtwo_i) \leq 
\vrelone(\substcomp{\termone}{\varone_i}{\valone_i}, 
\substcomp{\termone}{\varone_i}{\valtwo_i})$

An \emph{environment} $\envone$ is a sequence 
$\varone_1 :_{\baseone_1} \typeone_1, \hh, \varone_n :_{\baseone_n} \typeone_n$ 
of distinct identifiers with associated closed types and CBEs 
(we denote the empty environment by $\emptyctx$). 
We can lift operations on CBEs in 
Lemma \ref{lemma:algebra-change-of-base-functors} to environments as follows:
\begin{align*}
\basetwo \comp \envone
&= \varone_1 :_{\basetwo \comp \baseone_1} \typeone_1, \hh, \varone_n :_{\basetwo \comp \baseone_n} \typeone_n,
\\
\envone
\tensor 
\envtwo
&= \varone_1 :_{\baseone_1 \tensor \basetwo_1} \typeone_1, \hh, 
\varone_n :_{\baseone_n \tensor \basetwo_n} \typeone_n,
\\
\qop(\envone^{1}, \hh, \envone^{m})
&= \varone_1 :_{\qop(\baseone^{1}_1, \hh, \baseone^{m}_1)} \typeone_1, \hh, 
\varone_n :_{\qop(\baseone^{1}_n, \hh, \baseone^{m}_n)} \typeone_n, 
\end{align*}
for 
$\envone = \varone_1 :_{\baseone_1} \typeone_1, \hh, 
\varone_n :_{\baseone_n} \typeone_n$,
$\envtwo = \varone_1 :_{\basetwo_1} \typeone_1, \hh, 
\varone_n :_{\basetwo_n} \typeone_n$, and
$\envone^{i} = \varone_1 :_{\baseone^{i}_1} \typeone_1, \hh, 
\varone_n :_{\baseone^{i}_n} \typeone_n$. 
Note that the above operations are defined for environments having 
the same structure (i.e. differing only on CBEs). 
This is not a real restriction since we can always 
add the missing identifiers 
$\vartwo :_{\qunit} \typeone$, where $\qunit$ is the constant function 
returning the unit of the quantale 
(but see \cite{Pierce/DistanceMakesTypesGrowStronger/2010}).

The type system for $\quantale$-\Fuzz\ is defined in Figure 
\ref{fig:typing-system}. The system is based on two kinds of judgment 
(exploiting the fine-grained style of the calculus): 
judgments of the form $\envone \valimp \valone: \typeone$ for values 
and judgments of the form $\envone \compimp \termone: \typeone$ for terms. 
We denote by $\values_{\typeone}$ and $\Lambda_{\typeone}$ 
for the set of closed values and terms of type $\typeone$, 
respectively. Sometimes we also use the notation 
$\Lambda_{\envone \imp \typeone}$ for the set 
$\{\termone \in \Lambda \mid \envone \imp \termone: \typeone\}$
(and similarity for values).

\begin{figure}
{\hrulefill}
\vspace{0.2cm}

\(
\infer
  {\valseq{\envone, \varone :_{\baseone} \typeOne}{\varone : \typeOne}}
  {\baseone \leq \baseid}
\quad
\infer
  {\qop(\envOne_1, \hh, \envOne_n) \imp 
  \op(\termone_1, \hh, \termone_n): 
  \typeone}
  {\compseq{\envOne_1}{\termone_1: \typeOne}
  &
  \cc
  &
  \compseq{\envOne_n}{\termone_n: \typeOne}
  }
\)

\vspace{0.2cm}

\(
\infer
  {\valseq{\envOne}{\abs{\varone}\termone: \typeOne \mmap \typeTwo}}
  {\envOne, \varone :_{\baseid} \typeOne \compimp \termone: \typeTwo}
\quad
\infer
  {\compseq{\envOne \tensor \envTwo}{\valone \valtwo: \typeTwo}}
  {
  \valseq{\envOne}{\valone: \typeOne \mmap \typeTwo}
  &
  \valseq{\envTwo}{\valtwo: \typeOne}
  }
\)

\vspace{0.2cm}

\(
\infer
  {\envone \valimp \inject{\hat \imath}{\valone}: \sumtype{i \in I}{\typeone_i}}
  {\envone \valimp \valone: \typeone_{\hat \imath}}
\quad
\infer
  {\baseone \comp \envone \tensor \envtwo \imp 
    \casesum{\valone}{\termone_i}: \typetwo}
  {
  \envone \valimp \valone: \sumtype{i \in I}{\typeone_i}
  &
  \envtwo, \varone :_\baseone \typeone_i \imp \termone_i: \typetwo
  & 
  (\forall i \in I)
  }
\)

\vspace{0.2cm}

\(
\infer
  {\compseq{\envOne}{\return{\valone}: \typeOne}}
  {\valseq{\envOne}{\valone: \typeOne}}
\quad
\infer
  {\compseq{(\baseone \wedge \baseid) \comp \envOne \tensor \envTwo}
  {\seq{\termone}{\termtwo}: \typetwo}}
  {\compseq{\envOne}{\termone: \typeOne}
  &
  \compseq{\envTwo, \varone :_{\baseone} \typeOne}{\termtwo: \typetwo}
  }
\)

\vspace{0.2cm}

\(
\infer
  {\valseq{\baseone \comp \envOne}{\bang \valone: \bang_{\baseone} \typeOne}}
  {\valseq{\envOne}{\valone: \typeone}}
\quad
\infer
  {\compseq{\baseone \comp \envOne \tensor \envTwo}
  {\unbang{\valone}{\termone}: \typeTwo}}
  {
  \valseq{\envOne}{\valone:\bang_{\basetwo} \typeOne}
  &
  \compseq{\envTwo, \varone :_{\baseone \comp \basetwo} \typeone}
  {\termone: \typeTwo}
  }
\)

\vspace{0.2cm}

\(
\infer
  {\valseq{\envOne}{\fold{\valone}: \recType{\typeVar}{\typeOne}}}
  {
  \valseq{\envOne}{\valone: 
  \substType{\typeOne}{\typeVar}{\recType{\typeVar}{\typeOne}}}
  }
\quad
\infer
  {\compseq{\baseone \comp \envOne \tensor \envTwo}
  {\pmfold{\valone}{\termone}: \typeTwo}}
  {
  \valseq{\envOne}{\valone: \recType{\typeVar}{\typeOne}}
  &
  \compseq{\envTwo, \varone :_{\baseone} \substType{\typeOne}{\typeVar}
  {\recType{\typeVar}{\typeOne}}}
  {\termone: \typeTwo}
  }
\)

{\hrulefill}
\caption{Typing rules.}
\label{fig:typing-system}
\end{figure}

\begin{example}\label{ex:instances-of-v-fuzz}
\begin{enumerate}[wide = 0pt, leftmargin = *]
  \item Instantiating $\quantale$-\Fuzz\ with $\signature \defeq \emptyset$, 
    the Lawvere quantale, and CBEs 
    $\Pi = \{c \cdot - \mid c \in [0,\infty]\}$ we obtain 
    the original \Fuzz\ 
    \cite{Pierce/DistanceMakesTypesGrowStronger/2010}
    (provided we add a basic type for real numbers). We can also add 
    nondeterminism via a binary nondeterminism choice 
    operation $\oplus$. 
  \item We define the language $P$-\Fuzz\ as the instantiation of $\quantale$-\Fuzz\ 
    with a fair probabilistic choice operation $\oplus$, 
    the unit interval quantale $([0,1], \geq, +, 0)$, and CBEs 
    $\changeofbases = \{c \cdot - \mid c \in [0,\infty]\}$  
    (as usual we are actually referring to truncated multiplication).
    We interpret $\oplus$ in $[0,1]$ as in Example 
    \ref{ex:quantale-operations}. 
  \item We can add global states to 
    $P$-\Fuzz\ enriching $P$-\Fuzz's signature with operations 
    in $\signature_{\mathcal{L}}$ from Example \ref{ex:monads}.
  \end{enumerate}
\end{example}

Typing rules for $\quantale$-\Fuzz\ are similar to those of 
\Fuzz\ (e.g. in the variable rule we require $\baseone \leq 1$, meaning 
that the open value $\varone$ can access  $\varone$ at least once) with 
the exception of the rule for sequencing where we 
apply sensitivity $\baseone \wedge \baseid$ to the environment 
$\envone$ even if the sensitivity of $\varone$ in $\termtwo$ is $\baseone$. 
Consider the following instance of the sequencing rule on the Lawvere quantale:
\[
\infer[]
{x:_{\max(0,1) \cdot 1} \typeone \compimp \seqy{\termone}{\termtwo}: \typetwo}
{x:_\baseid \typeone \compimp \termone: \typeone
&
y:_0 \typeone \compimp \termtwo: \typetwo}
\]
where $\termtwo$ is a closed term of type $\typetwo$ and thus we can 
assume it to have sensitivity $0$ on all variables.
According to our informal intuition, $\termone$ has sensitivity $1$ 
on input $\varone$, meaning that $(i)$ $\termone$ can possibly detect 
(behavioural) differences between input values $\valone, \valtwo$, 
and $(ii)$ $\termone$ cannot amplify their behavioural distance 
of a factor bigger than $1$. Formally, 
point $(ii)$ states that we have the inequality 
$\vrelone(\valone, \valtwo) \geq 
\vrelone(\substcomp{\termone}{\varone}{\valone},
\substcomp{\termone}{\varone}{\valtwo})$, 
where $\vrelone$ denotes a suitable behavioural $[0,1]$-relation. 
On the contrary, $\termtwo$ is closed term and thus has sensitivity $0$ on 
any input, 
meaning that it cannot detect any observable difference between 
input values. 
In particular, for all values $\valone, \valtwo$ we have
$\vrelone(\substcomp{\termtwo}{\vartwo}{\valone},
\substcomp{\termtwo}{\vartwo}{\valtwo})
= \vrelone(\termtwo, \termtwo) = 0$ (provided that $\vrelone$ is reflexive). 
Replacing $\max(0,1)$ with $0$ in 
the above rule (i.e. $\baseone \wedge \baseid$ with $\baseone$ in the general case) 
would allow to infer the judgment 
$\varone :_0 \typeone \imp \seqy{\termone}{\termtwo}: \typetwo$, 
and thus to conclude
$
\vrelone(\seqy{\substcomp{\termone}{\varone}{\valone}}
              {\termtwo},
          \seqy{\substcomp{\termone}{\varone}{\valtwo}}
              {\termtwo})
= 0.
$
The latter equality is unsound as evaluating 
$\seqy{\substcomp{\termone}{\varone}{\valone}}
{\termtwo}$ 
(resp. $\seqy{\substcomp{\termone}{\varone}{\valtwo}}
{\termtwo}$) requires to 
\emph{first} evaluate $\substcomp{\termone}{\varone}{\valone}$
(resp. $\substcomp{\termone}{\varone}{\valtwo}$) thus making 
observable differences between $\valone$ and $\valtwo$ 
detectable (see also Section \ref{section:behavioural-v-relations} 
for a formal explanation).

\begin{example}\label{ex:v-fuzz-terms}
For every type $\typeone$ we have 
the term $I \defeq \return{(\abs{\varone}{\return \varone})}$ 
of type $\typeone \mmap \typeone$ as well as the purely 
divergent divergent term 
$\Omega \defeq \omega \bang (\fold \omega)$ of type $\typeone$, where 
$\omega \in 
\Lambda_{\bang_\infty(\recType{\typevarone}{\bang_\infty \typevarone \mmap \typeone}) 
\mmap \typeone}$ is defined by:
$
  \omega \defeq \abs{\varone}{\casebangy{\varone}
        {\casefoldz{\vartwo}{\varthree \bang (\fold \varthree)}}}.
$
\end{example} 

Before moving to the operational semantics of $\quantale$-\Fuzz, 
we remark that the syntactic distinction between terms and values gives 
the following equalities.

\begin{lemma}
The following equalities hold:
\begin{align*}
  \values_{\typeone \mmap \typetwo} &= 
  \{\abs{\varone}{\termone} \mid \varone :_{\baseid} 
  \typeone \compimp \termone: \typetwo\}, \\
  \values_{\sumtype{i \in I}{\typeone_i}} &= 
  \bigcup_{\hat \imath \in I} \{\inject{\hat \imath}{\valone} 
  \mid \valone \in \values_{\typeone_{\hat \imath}}\}, 
  \\
  \values_{\bang_{\baseone} \typeone} &= 
  \{\bang \valone \mid \valone \in \values_\typeone\}.
\end{align*}
\end{lemma}

\paragraph{Operational Semantics} 
We give $\quantale$-\Fuzz\ monadic operational (notably evaluation)
semantics in the style of \cite{DalLagoGavazzoLevy/LICS/2017}. 
Let $\Monad = \lan \monad, \unit, \kleisli{-}\ran$ be a 
$\signature$-continuous monad. 
Operational semantics is defined by means of an evaluation 
function $\sem{-}^\typeone$ indexed over closed types,
associating to any term in $\Lambda_\typeone$ 
a monadic value in $\monad \values_\typeone$. The evaluation 
function $\sem{-}^\typeone$ is itself defined by means of 
the family of functions $\{\sem{-}^\typeone_n\}_{n < \omega}$
defined in Figure \ref{fig:approximation-semantics}. 
Indeed $\sem{-}_n^\typeone$ is a function from $\values_\typeone$ to 
$\monad \values_\typeone$. 

\begin{figure}
{\hrulefill}
\begin{align*}
  \sem{\termone}^\typeone_0  
  &\defeq \mbot_{\values_\typeone}    \\
  \sem{\return{\valone}}_{n+1}^\typeone 
  &\defeq \unit_{\values_\typeone}(\valone) \\
  \sem{(\abs{\varone}{\termone})\valone}_{n+1}^\typeone
  &\defeq \sem{\substcomp{\termone}{\varone}{\valone}}_n^\typeone \\
  \sem{\casesum{\inject{\hat \imath}{\valone}}{\termone_i}}_{n+1}^\typeone
  &\defeq \sem{\substcomp{\termone_{\hat \imath}}{\varone}{\valone}}_n^\typeone \\
  \sem{\pmfold{(\fold{\valone})}{\termone}}_{n+1}^\typeone
  &\defeq \sem{\substcomp{\termone}{\varone}{\valone}}_{n}^\typeone \\
  \sem{\pmbang{\bang \valone}{\termone}}_{n+1}^\typeone
  &\defeq \sem{\substcomp{\termone}{\varone}{\valone}}_{n}^\typeone \\
  \sem{\seq{\termone}{\termtwo}}_{n+1}^\typeone 
  &\defeq \kleisli{(\sem{\substcomp{\termtwo}{\varone}{-}}_n^{\typetwo, \typeone})}
    \sem{\termone}_n^\typetwo \\
  \sem{\op(\termone_1, \hh, \termone_k)}_{n+1}^\typeone 
  &\defeq \mop_{\values_\typeone}(\sem{\termone_1}_n^\typeone, \hh, 
  \sem{\termone_k}_n^\typeone)
\end{align*}
{\hrulefill}
\caption{Approximation evaluation semantics.}
\label{fig:approximation-semantics}
\end{figure}

Let us expand on the definition of $\sem{\seq{\termone}{\termtwo}}_{n+1}^\typeone$. 
Since $\seq{\termone}{\termtwo} \in \Lambda_\typeone$, 
there must be derivable judgments  
$\emptyctx \imp \termone: \typetwo$ and 
$\varone:_\baseone \typetwo \imp \termtwo: \typeone$. As a consequence,  
for any $\valone \in \values_\typetwo$, we have
$\sem{\substcomp{\termtwo}{\varone}{\valone}}_n^\typeone \in \monad\values_{\typeone}$.
This induces a function 
$\sem{\substcomp{\termtwo}{\varone}{-}}_n^{\typetwo,\typeone}$ 
from $\values_\typetwo$ to $\monad \values_\typeone$ 
whose Kleisli extension can be applied to 
$\sem{\termone}_n^\typetwo \in \monad \values_\typetwo$.

Finally, it is easy to see that  
$(\sem{\termone}_n)_{n < \omega}$ forms an $\omega$-chain 
in $\monad \values_\typeone$ (see Appendix \ref{appendix:proofs-v-fuzz} 
for a proof of the following result).

\begin{restatable}{lemma}{evaluationSemanticsOmegaChain}
  For any $\termone \in \Lambda_\typeone$, 
  we have
  $
  \sem{\termone}_n^\typeone \cpoleq_{\values_\typeone} \sem{\termone}_{n+1}^\typeone
  $,
  for any $n \geq 0$.
\end{restatable}

As a consequence, we can define  
$\sem{-}^\typeone: \Lambda_\typeone \to \monad \values_\typeone$
by 
$$
\sem{\termone}^\typeone \defeq \lub\nolimits_{n<\omega} \sem{\termone}_n^\typeone.
$$
In order to improve readability we oftentimes omit type superscripts 
in $\sem{\termone}^\typeone$.
We also notice that because $\mop$ is continuous and 
$\Monad$ is \ocppo-enriched, $\sem{-}^\typeone$ is itself continuous.

\begin{proposition}\label{prop:continuity-evaluation-function}
  The following equations hold:
    \begin{align*}
    \sem{\return{\valone}} 
    &= \unit(\valone), \\
    \sem{(\abs{\varone}{\termone})\valone}
    &= \sem{\substcomp{\termone}{\varone}{\valone}}, \\
    \sem{\casesum{\inject{\hat \imath}{\valone}}{\termone_i}}
    &= \sem{\substcomp{\termone_{\hat \imath}}{\varone}{\valone}}, \\
    \sem{\pmfold{(\fold{\valone})}{\termone}}
    &= \sem{\substcomp{\termone}{\varone}{\valone}}, \\
    \sem{\pmbang{\bang \valone}{\termone}}
    &= \sem{\substcomp{\termone}{\varone}{\valone}}, \\
    \sem{\seq{\termone}{\termtwo}}
    &= \kleisli{\sem{\substcomp{\termtwo}{\varone}{-}}}(\sem{\termone}), \\
    \sem{\op(\termone_1, \hh, \termone_k)} 
    &= \mop_{\values_\typeone}(\sem{\termone_1}, \hh, \sem{\termone_k}).
  \end{align*}
\end{proposition}

\section{$\quantale$-relators and $\quantale$-relation Lifting}
\label{section:v-relators-and-v-relation-lifting}

In \cite{DalLagoGavazzoLevy/LICS/2017} the abstract theory of 
relators \cite{Barr/LMM/1970,Thijs/PhDThesis/1996} has been used to define
notions of applicative (bi)similarity for an untyped 
$\lambda$-calculus enriched with algebraic operations. 
Intuitively, a relator $\relator$ for a set endofunctor $\functor$
is an abstraction meant to capture 
the possible ways a relation on a set $X$ can be turned (or lifted) 
into a relation on $\monad X$. 
Relators allow to abstractly express the idea that 
bisimilar programs, when executed, exhibit the same  
observable behaviour (i.e. they produce the same effects) 
and evaluate to bisimilar values. 
In particular, whenever two 
programs $\termone$ and $\termone'$ are related by a 
(bi)simulation $\relone$, then the results  
$\sem{\termone}$ and $\sem{\termone'}$  of their evaluation 
must be related by $\relator \relone$. 
The latter relation ranging over 
monadic values, it takes into account the visible effects of executing 
$\termone$ and $\termone'$, such effects being encapsulated 
via $\monad$. 

The notion of $\quantale$-relator \cite{Hoffman-Seal-Tholem/monoidal-topology/2014} 
is somehow the `quantitative' generalisation
of the concept of a relator. Analogously to ordinary relators, 
$\quantale$-relators for a set endofunctor $\functor$ are abstractions
meant to capture the possible ways a $\quantale$-relation on a set $X$ 
can be (nicely) turned into a $\quantale$-relation on $\monad X$, and thus 
provide ways to lift a behavioural distance between programs to a 
(behavioural) distance between monadic values.
On a formal level, we say that a $\quantale$-relator extends $\functor$ from 
$\set$ to $\vrel$, laxly\footnote{Relators are also known as \emph{lax extensions} 
\cite{Hoffman-Seal-Tholem/monoidal-topology/2014,Hoffman/Cottage-industry/2015}.}. 

\begin{definition}\label{def:v-relator}
For a set endofucunctor $\monad$ a \emph{$\quantale$-relator} 
for $\monad$ is a mapping
$
(\vrelone: \setone \torel \settwo) \mapsto 
(\vrelator \vrelone: \monad \setone \torel \monad \settwo)
$
satisfying conditions \eqref{vrel-1}-\eqref{vrel-4}. 
We say that $\vrelator$ is conversive if it additionally satisfies 
condition \eqref{vrel-5}.
\begin{align*}
\idrel_{\monad \setone} 
&\leq  \vrelator(\vidrel_{\setone}),  
\tag{$\quantale$-rel 1} \label{vrel-1}    
\\
\vrelator\vreltwo \comp \vrelator\vrelone   
& \leq \vrelator(\vreltwo \comp \vrelone),  
\tag{$\quantale$-rel 2} \label{vrel-2}    
\\
\monad f \leq \vrelator f   
& \text{,}\phantom{\leq} \dual{(\monad f)} \leq \vrelator\dual{f},
\tag{$\quantale$-rel 3} \label{vrel-3}    
\\
\vrelone \leq \vreltwo                 
& \implies \vrelator\vrelone \leq \vrelator\vreltwo,
\tag{$\quantale$-rel 4} \label{vrel-4}
\\
\vrelator(\dual{\vrelone})
&=  \dual{(\vrelator \vrelone)}.   
\tag{$\quantale$-rel 5} \label{vrel-5}
\end{align*}
\end{definition}

Conditions (\ref{vrel-1}), (\ref{vrel-2}), and (\ref{vrel-4}) 
are rather standard. Condition (\ref{vrel-3}), which actually consists 
of two conditions, states that $\quantale$-relators behave in 
the expected way on functions. 
It is immediate to see that when instantiated with $\quantale = \two$, 
the above definition gives the usual notion of relator, with some 
minor differences. In \cite{DalLagoGavazzoLevy/LICS/2017} 
and \cite{Levy/FOSSACS/2011} 
a kernel preservation condition is required in place of 
(\ref{vrel-3}). Such condition is also known as stability in 
\cite{Huge-Jacobs/Simulations-in-coalgebra/2004}. 
Stability requires the equality 
$$
\vrelator(\dual{g} \comp \vrelone \comp f) 
= 
\dual{(\monad g)} \comp \vrelator \vrelone \comp \monad f
$$
to hold. It is easy to see that a $\quantale$-relator always 
satisfies stability. Notice also that stability gives the following implication:
$$
\vrelone \leq \dual{g} \comp \vreltwo \comp f 
\implies 
\vrelator \vrelone \leq \dual{(\monad g)} \comp \vrelator \vreltwo \comp \monad f, 
$$
which can be diagrammatically expressed as:
\begin{center}
\(
\vcenter{
\xymatrix{
\laxcommute
\setone     
\ar[r]^{f} 
\ar[d]_{\vrelone}|@{|}  &  
\setthree   
\ar[d]^{\vreltwo}|-*=0@{|}
\\
\settwo   
\ar[r]_{g}   &  
\setfour  } }
\)
$\implies$
\(
\vcenter{
\xymatrix{
\laxcommute
\monad\setone     
\ar[r]^{\monad f} 
\ar[d]_{\vrelator\vrelone}|@{|}  &  
\monad\setthree   
\ar[d]^{\vrelator\vreltwo}|-*=0@{|}
\\
\monad\settwo   
\ar[r]_{\monad g}   &  
\monad \setfour  } }
\).
\end{center}

Finally, we observe that any $\quantale$-relator $\vrelator$ for $\functor$ induces an endomap $\functor_\vrelator$ on $\vrel$ 
that acts as $\functor$ on sets and as $\vrelator$ as $\quantale$-relation. 
It is easy to check that conditions in Definition \ref{def:v-relator} 
makes $\monad_\vrelator$ a \emph{lax endofunctor}.

Before giving examples of $\quantale$-relators it is useful to 
observe that the collection $\quantale$-relators is closed under specific 
operations.

\begin{restatable}{proposition}{algebraVrelators}
\label{prop:algebra-of-v-relators}
Let $\functorone, \functortwo$ be set endofunctors. Then:
\begin{enumerate}[wide = 0pt, leftmargin = *]
  \item
  If $\vrelatorone$ and $\vrelatortwo$ are $\quantale$-relators for 
  $\functorone$ and $\functortwo$, respectively, then 
  $\vrelatortwo \comp \vrelatorone$ defined by 
  $
  (\vrelatortwo \comp \vrelatorone)\vrelone \defeq 
  \vrelatortwo \vrelatorone\vrelone$
  is a $\quantale$-relator for $\functortwo \functorone$.
  \item 
  If $\{\vrelator\}_{i \in I}$ is a family of $\quantale$-relators 
  for $\functorone$, then $\meet_{i \in I}\vrelator_i$ defined by 
  $
  (\meet_{i \in I} \vrelator_i)\vrelone 
  \defeq
  \meet_{i \in I} \vrelator_i\vrelone 
  $
  is a $\quantale$-relator for $\functorone$.
  \item 
  If $\vrelator$ is a $\quantale$-relator for $\functorone$, 
  then $\dual{\vrelator}$ defined by $\dual{\vrelator}\vrelone \defeq 
  \dual{(\vrelator\dual{\vrelone})}$ is a 
  $\quantale$-relator for $\monad$. 
  \item 
  For any $\quantale$-relator $\vrelator$, 
  $\vrelator \wedge \dual{\vrelator}$ is the greatest conversive 
  $\quantale$-relator smaller than $\vrelator$.
\end{enumerate}
\end{restatable}

\begin{proof}
See Appendix \ref{appendix:proofs-v-relators-and-v-relation-lifting}.
\end{proof}

\begin{example}\label{ex:v-relators}
Let us consider the monads in Example \ref{ex:monads} regarded as 
functors.
\begin{enumerate}[wide = 0pt, leftmargin = *]
  \item For the partiality functor $(-)_\mbot$ 
    define the $\quantale$-relator $(-)_\mbot$ by:
    $$
     \vrelone_\mbot(x,y) \defeq \vrelone(x,y), 
     \quad \vrelone_\mbot(\mbot_X,\mathpzc{y}) \defeq \qunit,
     \quad \vrelone_\mbot(x, \mbot_Y) = \qbot,
     $$
     where $x \in X, y \in Y,
     \mathpzc{y} \in Y_\mbot$, and $\vrelone: X \torel Y$.
    The $\quantale$-relation $\vrelone_\mbot$ generalises
    the usual notion of \emph{simulation} for partial computations. 
    Similarly, $\vrelone_{\mbot\mbot} \defeq 
    \vrelone_\mbot \wedge \dual{((\dual{\vrelone})_\mbot)}$
    generalises the usual notion of \emph{bisimulation} for partial computation. 
  \item For the powerset functor $\powerset$ 
    define the $\quantale$-relator $\hausdorff$ 
    (called Hausdorff lifting) and 
    its conversive counterpart 
    $\hausdorffsym \defeq \hausdorff \wedge \dual{\hausdorff}$ by
    $
    \hausdorff \vrelone(\mathpzc{X},\mathpzc{Y}) 
    \defeq
    \meet_{x \in \mathpzc{X}} \join_{y \in \mathpzc{Y}} \vrelone(x,y).$
    If we instantiate $\quantale$ as the Lawvere 
    quantale, then $\hausdorffsym$ gives the usual 
    Hausdorff lifting of distances on a set $X$ to distances on $\powerset X$,
    whereas for $\quantale = \two$ we recover the usual notion of
    \emph{(bi)simulation} for unlabelled transition systems.
\item For the full distribution functor $\distribution$
  we define a $[0,1]$-relator (with respect to 
  the unit interval quantale) using the so-called 
  \emph{Wasserstein-Kantorovich lifting} \cite{Villani/optimal-transport/2008}.
  For $\mu \in \distribution(X), \nu \in \distribution(Y)$, 
  the set $\Omega(\mu, \nu)$ of \emph{couplings} of $\mu$ and $\nu$
  is the set of joint distributions $\omega \in \distribution(X \times Y)$ 
  such that $\mu = \sum_{y\in Y} \omega(-, y)$ and 
  $\nu = \sum_{x \in X} \omega(x,-)$. For a $[0,1]$-relation 
  $\vrelone: X \torel Y$ define:
  $$
  \wasserstein \vrelone(\mu, \nu) \defeq
  \inf\nolimits_{\omega \in \Omega(\distone, \disttwo)} 
  \sum\nolimits_{x,y} \vrelone(x,y) \cdot \omega(x,y).
  $$ 
  $\wasserstein \vrelone(\mu, \nu)$ attains its infimum and has 
  a dual characterisation. 
  \begin{restatable}{proposition}{dualityWassersteinLifting}
  \label{prop:duality-wasserstein-lifting}
  Let $\distone \in \distribution(X), \disttwo \in \distribution(\settwo)$ 
  be countable distributions and 
  $\vrelone: \setone \torel \settwo$ be a $\intervalQuantale$-relation.
  Then:
  \begin{align*}
  \wasserstein \vrelone(\mu,\nu) 
  &= 
  \min \{ \sum\nolimits_{x,y} \vrelone(x,y) \cdot \omega(x,y)  
  \mid\omega \in \Omega(\distone, \disttwo)\} \\
  &=
  \max \{\sum\nolimits_{x} a_x \cdot \mu(x) + 
  \sum\nolimits_y b_y \cdot \nu(y)  
  \\
  &\phantom{=} \mid a_x + b_y \leq \vrelone(x,y), 
    a_x, b_y \text{ bounded}\},
  \end{align*}
  where $a_x,b_y$ bounded means that  
  there exist $\bar a,\bar b \in \mathbb{R}$ such 
  that $\forall x.\ a_x \leq \bar a$, and 
  $\forall y.\ b_y \leq \bar b$.
  \end{restatable}
  The above proposition (see Appendix 
  \ref{appendix:proofs-behavioural-v-relations} for a proof) 
  is a direct consequence of the  
  Duality Theorem for \emph{countable} transportation problems 
  \cite{Kortanek/InfiniteTransportationProblems/1995} 
  (Theorem 2.1 and 2.2).
  Using Proposition \ref{prop:duality-wasserstein-lifting} we can show that
  $\wasserstein$ indeed defines a $[0,1]$-relator 
  (but see Digression \ref{digression:building-v-relators}).  
  Finally, we can compose the Wasserstein lifting $\wasserstein$ with the 
  $\quantale$-relator $(-)_\mbot$ 
  of point 1 obtaining the (non-conversive) 
  $[0,1]$-relator $\wassersteinbot$ for the countable 
  subdistribution functor $\distribution_{\leq 1}$.   
  \end{enumerate}
\end{example}

\begin{digression}[Building $\quantale$-relators] 
\label{digression:building-v-relators}
Most of the $\quantale$-relators in Example \ref{ex:v-relators} can be 
obtained using a general abstract construction refining the so-called 
\emph{Barr extension} of a functor \cite{Kurz/Tutorial-relation-lifting/2016}. 
Recall that any relation $\relone: X \torel Y$ 
(i.e. a $\two$-relation $\relone: X \times Y \to \two$) can be 
equivalently presented as a subset of $X \times Y$ via its graph 
$G_{\relone}$. This allows to express 
$\relone$ as $\pi_2 \comp \dual{\pi_1}$ (in $\rel$), where 
$\pi_1: G_\relone \to X$, $\pi_2: G_\relone \to Y$ are the 
usual projection functions.

\begin{definition}\label{def:barr-extension}
Let $\monad$ be an endofunctor on $\set$ and $\relone: X \torel Y$ 
be a a relation. The \emph{Barr extension} $\overline{\monad}$ of 
$\monad$ to $\rel$ is defined by:
$$
\overline{\monad} \relone \defeq \monad \pi_2 \comp \dual{(\monad \pi_1)},
$$
where $\relone = \pi_2 \comp \dual{\pi_1}$.
Pointwise, $\overline{\monad}$ is defined by:
$$
\mathpzc{x}\ \overline{\monad} \relone\ \mathpzc{y} \iff 
\exists \mathpzc{w} \in \monad G_\relone.\ ( 
\monad \pi_1 (\mathpzc{w}) = \mathpzc{x},\ 
\monad \pi_2 (\mathpzc{w}) = \mathpzc{y}),
$$
where $\mathpzc{x} \in \monad X$ and $\mathpzc{y} = \monad Y$
\end{definition}
In general, $\overline{\monad}$ is not a $\two$-relator, but it is 
so if $\monad$ preserves weak pullback diagrams 
\cite{Kurz/Tutorial-relation-lifting/2016}
(or, equivalently, if $\functor$ satisfies the Beck-Chevalley condition 
\cite{Hoffman-Seal-Tholem/monoidal-topology/2014}).  
Such condition is satisfied by all functors we have considered so far 
in our examples. 

Definition \ref{def:barr-extension} crucially relies on the 
double nature of a relation, which can be viewed both as an 
arrow in $\rel$ and as an object in $\set$. This is no 
longer the case for a $\quantale$-relation, and thus it is not 
clear how to define the Barr extension of a functor $\functor$ from 
$\set$ to $\vrel$.
However, the Barr extension of $\functor$ can be characterised in an 
alternative way if we assume $\functor$ to preserves weak pullback diagrams 
(although the reader can see 
\cite{Manes/Taut-monads,Hoffman-topological-theories-as-closed-objects} 
for more general conditions). Let 
$\xi: \functor \two \to \two$ be the map defined by 
$\xi(\mathpzc{x}) = \true$ 
if and only if $\mathpzc{x} \in \functor \{\true\}$, where $\functor \{\true\}$ 
is the image of the map $\functor \iota$ for the inclusion 
$\iota: \{\true\} \to \two$. That is, $\xi(\mathpzc{x}) = \true$ if and 
only if there exists an element $\mathpzc{y} \in \functor\{\true\}$ such 
that $\functor \iota (\mathpzc{y}) = \mathpzc{x}$. Note that this makes sense 
since $\functor$ preserves monomorphisms (recall that we can describe monomorphism 
as weak pullbacks) and thus 
$\functor \iota: \functor \{\true\} \to \functor \two$ is a monomorphism.
We can now characterise
$\overline{\functor} \relone$ without mentioning the graph of 
$\relone$:
$$
\overline{\functor}\relone(\mathpzc{x},\mathpzc{y}) = \true
\iff 
\exists \mathpzc{w} \in \functor (X \times Y).\ 
\begin{cases}
\functor \pi_1 (\mathpzc{w}) &= \mathpzc{x}, \\
\functor \pi_2 (\mathpzc{w}) &= \mathpzc{y}, \\
\xi \comp \functor \relone(\mathpzc{w}) &= \true.
\end{cases}
$$
Since the existential quantification is nothing but the joint of the 
boolean quantale $\two$, the above characterisation
of $\overline{\functor}$ can be turned into a definition of an extension 
of $\functor$ to $\vrel$ parametric with respect to a map 
$\xi: \functor \quantale \to \quantale$. 

\begin{definition}\label{def:v-barr-extension}
For a set endofunctor $\functor$ and a map 
$\xi: \functor \quantale \to \quantale$ define the 
\emph{$\quantale$-Barr extension} $\overline{\functor}_\xi$ 
of $\functor$ to $\vrel$ with respect to $\xi$ as follows:
$$
\overline{\functor}_\xi \vrelone(\mathpzc{x}, \mathpzc{y}) 
\defeq \join_{\mathpzc{w} \in \Omega(\mathpzc{x}, \mathpzc{y})} 
\xi \comp \functor \vrelone(\mathpzc{w}),
$$
for $\mathpzc{x} \in \functor X, \mathpzc{y} \in \functor Y$, where 
the set $\Omega(\mathpzc{x}, \mathpzc{y})$ of generalised couplings of 
$\mathpzc{x}, \mathpzc{y}$ is defined by:
$$
\Omega(\mathpzc{x}, \mathpzc{y}) \defeq 
\{\mathpzc{w} \in \functor(X \times Y) \mid
  \functor \pi_1 (\mathpzc{w}) = \mathpzc{x},\ 
  \functor \pi_2 (\mathpzc{w}) = \mathpzc{y}\}.
$$
\end{definition}

\begin{example}\label{ex:stucture-maps}
\begin{enumerate}[wide = 0pt, leftmargin = *]
  \item Taking 
    $\xi: \powerset \quantale \to \quantale$ defined by 
    $\xi(\mathpzc{X}) \defeq \meet \mathpzc{X}$ we 
    recover the Hausdorff lifting $\hausdorffsym$. 
  \item Taking expectation function 
    $\xi: \distribution [0,1] \to [0,1]$ defined by 
    $\xi(\mu) \defeq \sum_x x \cdot \mu(x)$ we 
    recover Wasserstein lifting $\wasserstein$.
\end{enumerate}
\end{example} 
Using the map $\xi: \functor \quantale \to \quantale$ we can define an extension 
of $\functor$ to $\vrel$. However, such extension is in general not a 
$\quantale$-relators. 
Nonetheless, under mild conditions on $\xi$ and assuming $\functor$ to 
preserve weak pullback, it is possible to show that 
$\overline{\functor}_\xi$ is indeed a $\quantale$-relator. The following 
proposition has been proved in \cite{Clementino-Tholen/From-lax-monad-extensions-to-topological-theoreis,Hoffman-topological-theories-as-closed-objects} 
(a similar result for real-valued pseudometric spaces has been proved in 
\cite{Bonchi/Behavioral-metrics-via-functor-lifting/FSTTCS/2014,Bonchi/Towards-trace-metrics-via-functor-fifting/CALCO/2015}, where
an additional extension still parametric over $\xi$ is also studied).

\begin{proposition}\label{prop:v-barr-extensions-are-v-relators}
  Let $\functor$ be functor preserving weak pullbacks and 
  $\xi: \functor \quantale \to \quantale$ 
  be a map such that:
  \begin{enumerate}[wide = 0pt, leftmargin = *]
    \item $\xi$ respect quantale multiplication:
      $$
        \xymatrix{
        \laxcommute
          \functor(\quantale \times \quantale)    
            \ar[r]^{\functor \tensor} 
            \ar[d]_{\lan \xi \comp \functor \pi_1, \xi \comp \functor \pi_2\ran}  &  
          \functor \quantale   
            \ar[d]^{\xi\ .}
          \\
          \quantale \times \quantale   
            \ar[r]_{\tensor}   &  
          \quantale  
          } 
      $$
    \item $\xi$ respects the unit of the quantale:
      $$
        \xymatrix{
        \laxcommute
          \functor 1    
            \ar[r]^{\functor \qunit} 
            \ar[d]_{!}  &  
          \functor \quantale   
            \ar[d]^{\xi\ .}
          \\
          1    
            \ar[r]_{\qunit}   &  
          \quantale  
          } 
      $$
    \item $\xi$ respects the order of the quantale. That is, the map 
      $\varphi \mapsto \xi \comp \functor \varphi$, for $\varphi: X \to \quantale$, 
      is monotone.
  \end{enumerate}
  Then $\overline{\functor}_\xi$ is a conversive $\quantale$-relator.
\end{proposition}

It is straightforward to check that the expectation function in 
Example \ref{ex:stucture-maps} satisfies the above three conditions. 
By Proposition \ref{prop:v-barr-extensions-are-v-relators} 
it follows that the Wasserstein lifting gives indeed a $[0,1]$-relator, 
and thus so does its composition with the $[0,1]$-relator $(-)_\undeff$.

The extension $\overline{\functor}_{\xi}$ gives a somehow canonical  
\emph{conversive} $\quantale$-relator and thus provides a way to build 
canonical (applicative) $\quantale$-\emph{bi}simulations. 
However, $\overline{\functor}_{\xi}$ being intrinsically conversive 
it is not a good candidate to build $\quantale$-simulations.
For most of the examples considered we can get around the problem
considering $(\overline{\functor}_\xi)_\bot$ (as we do with e.g. $\wasserstein_\bot$). 
Nonetheless, it is desirable 
to have a general notion of extension characterising 
notions of $\quantale$-simulations. That has been done for ordinary relations 
in e.g. \cite{Huge-Jacobs/Simulations-in-coalgebra/2004,Levy/FOSSACS/2011} 
for functors $\functor$ inducing 
a suitable order $\leq_{X}$ on $\monad X$ and considering the 
relator $\overline{\functor}_{\leq} \defeq 
\leq_{-} \comp \overline{\functor} \comp \leq_{-}$. 
Proving that $\overline{\functor}_{\leq}$ gives indeed a relator requires 
$\monad$ to satisfy specific conditions. For instance,
in \cite{Levy/FOSSACS/2011} it is proved that if $\functor$ 
satisfies a suitable form of weak-pullback preservation (which takes into 
account the order induced by $\functor$), then 
$\overline{\functor}_{\leq}$ is indeed a relator.
This suggests to consider functors $\functor$ inducing a suitable 
$\quantale$-relation $\vrelone_X$ on $\monad X$ and thus to 
study if, and under which conditions, 
$\vrelone_{-} \comp \overline{\functor}_\xi \comp \vrelone_{-}$ 
is a $\quantale$-relator. This proposal has not been investigated 
in the context of the present work but it definitely constitutes a topic for 
future research.
\end{digression}

\paragraph{$\quantale$-relators for Strong Monads}

In previous paragraph we saw that a $\quantale$-relator 
extends a functor from $\set$ to $\vrel$ laxly.
Since we model effects through strong monads it seems more natural to 
require $\quantale$-relators to
extend strong monads from $\set$ to $\vrel$ laxly. 

The reason behind such requirement can be intuitively understood as follows. 
Recall that by Proposition \ref{prop:continuity-evaluation-function} 
we have (for readability we omit types) 
$\sem{\seq{\termone}{\termtwo}} = 
\kleisli{\sem{\substcomp{\termtwo}{\varone}{-}}} \sem{\termone}$. 
This operation can be described using the so called bind function 
$$
\bind: (X \to \monad Y) \times \monad X \to \monad Y,
$$ 
so that we have 
$\sem{\seq{\termone}{\termtwo}} = 
\sem{\substcomp{\termtwo}{\varone}{-}} \bind \sem{\termone}$. 
Now, let $f,g: X \to Y$ be functions, $\vrelone: X \torel X, \vreltwo: Y \torel Y$ 
be $\quantale$ relations, 
and $\vrelator$ be a $\quantale$-relator for $\monad$.
Considering the compound $\quantale$ relation 
$[\vrelone, \vrelator \vreltwo] \tensor \vrelator \vrelone$ 
(see Digression \ref{digression:v-categories}) and ignoring 
issues about sensitivity, it is then natural to require 
$\bind$ to be non-expansive. That is, we require the inequality 
$$
[\vrelone, \vrelator \vreltwo](f,g) \tensor 
\vrelator \vrelone(\mathpzc{x},\mathpzc{y}) \leq 
\vrelator \vreltwo(f \bind \mathpzc{x},g \bind \mathpzc{y})
$$
i.e.
$$
\meet_{x \in X}\vrelator\vreltwo(f(x),g(x)) \tensor 
\vrelator \vrelone(\mathpzc{x},\mathpzc{y}) \leq 
\vrelator \vreltwo(f \bind \mathpzc{x},g \bind \mathpzc{y}).
$$
Informally, we are requiring the behavioural distance 
between sequential compositions of programs to be
bounded by the behavioural distances between their components 
(this is of course a too strong requirement, but at this point 
it should be clear to the reader that it is sufficient to 
require $\bind$ to be Lipshitz continuous rather 
than non-expansive). 
Since $\bind$ is nothing but the strong Kleisli extension 
$\strongkleisli{\mathsf{apply}}$ of the application function 
$\mathsf{apply}: (X \to \monad Y) \times X \to \monad Y$ 
defined by $\mathsf{apply}(f,x) \defeq f(x)$, what we need to do 
is indeed to extend strong monads from $\set$ to $\vrel$ (laxly).

\begin{definition}\label{def:strong-v-relator}
  Let $\Monad = \lan \monad, \unit, \strongkleisli{-} \ran$ 
  be a strong monad on $\set$, and $\vrelator$ be a $\quantale$-relator 
  for $\monad$ (regarded as a functor). 
  We say that $\vrelator$ is an $L$-continuous\footnote{
  Instantiating $\quantale$ as the Lawvere quantale, 
  we see that condition \eqref{s-Strong-Lax-Bind} is requiring 
  Lipshitz continuity of multiplication and strength of $\Monad$.
  } 
  $\quantale$-relator for $\Monad$ if it satisfies 
  the following conditions for any CBE $\baseone \leq 1$.
  \begin{align}
  \vrelone \leq \dual{\unit_Y} \comp \vrelator \vrelone \comp \unit_X,
  \tag{Lax unit} \label{Lax-Unit} 
  \\
  \vrelthree \tensor (\baseone \acts \vrelone) 
  \leq \dual{g} \comp \vrelator \vreltwo \comp f 
  &\implies 
  \vrelthree \tensor (\baseone \acts  \vrelator\vrelone) 
  \leq \dual{(\strongkleisli{g})} \comp \vrelator \vreltwo \comp \strongkleisli{f},
  \tag{$L$-Strong lax bind} \label{s-Strong-Lax-Bind}
  \end{align}
\end{definition}
The condition $\baseone \leq 1$ reflects the presence of 
$\baseone \wedge \baseid$ in the typing rule for sequencing.
Also notice that by taking $\baseone \defeq 1$, conditions \eqref{Lax-Unit} 
and \eqref{s-Strong-Lax-Bind}
are equivalent to requiring unit, multiplication, and strength of $\Monad$ 
to be non-expansive. 
 
\begin{example}\label{ex:wasserstein-lifintg-satisfies-strong-lax-bind}
It is easy to check that $\quantale$-relators for the 
partiality and the powerset monads satisfy conditions in  
Definition \ref{def:strong-v-relator}. 
Using Proposition \ref{prop:duality-wasserstein-lifting} 
it is possible to show that also the Wasserstein 
lifting(s) $\wasserstein$ and $\wassersteinbot$ do, 
although this is less trivial (see Appendix 
\ref{appendix:proofs-behavioural-v-relations}).
\end{example}

Finally, if $\Monad$ is $\signature$-continuous we require 
$\quantale$-relators for $\Monad$ to be compatible with the 
$\signature$-continuous structure.

\begin{definition}
  Let $\Monad$ be a $\signature$-continuous monad, 
  $\quantale$ be a $\signature$-quantale, and  
  $\vrelator$ be a $\quantale$-relator for $\Monad$. 
  We say that $\vrelator$ is \emph{$\signature$-compatible} 
  and \emph{inductive} if the following 
  inequalities hold:
  \begin{align*}
   \qop(\vrelator\vrelone(\mathpzc{u}_1, \mathpzc{y}_1), \hh 
  \vrelator\vrelone(\mathpzc{u}_n, \mathpzc{y}_n)) 
  &\leq
  \vrelator\vrelone(\mop_{X}
  (\mathpzc{u}_1, \hh, \mathpzc{u}_n),
  \mop_{Y}(\mathpzc{y}_1, \hh, \mathpzc{y}_n)),
  \\
  \qunit 
  &\leq \vrelator \vrelone(\mbot_{X}, \mathpzc{y}), 
  \\
  \meet\nolimits_n \vrelator\vrelone(\mathpzc{x}_n, \mathpzc{y})
  &\leq \vrelator\vrelone(\lub\nolimits_n \mathpzc{x}_n, \mathpzc{y}).
  \end{align*}
  for any 
  $\omega$-chain $(\mathpzc{x}_n)_{n < \omega}$ and 
  elements $\mathpzc{u}_1, \hh, \mathpzc{u}_n$ in $\monad X$, 
  elements $\mathpzc{y},\mathpzc{y}_1, \hh, \mathpzc{y}_n \in \monad Y$,
  $n$-ary operation symbol $\op \in \signature$, and 
  $\quantale$-relation $\vrelone: X \torel Y$.
\end{definition}    
In particular, if $\vrelator$ is inductive and  
$\qvalone \leq \vrelator\vrelone(\mathpzc{x}_n, \mathpzc{y}))$ holds 
for any $n < \omega$, then
$
\qvalone \leq \vrelator\vrelone(\lub_{n< \omega} \mathpzc{x}_n, \mathpzc{y}).
$

\begin{example}\label{ex:inductive-v-relator}
Easy calculations show that $(-)_\mbot$ and $\hausdorff$
are inductive and $\signature$-compatible.
Using results from \cite{Villani/optimal-transport/2008} 
and \cite{Wasserstein-metric-and-subordination} (Lemma 5.2) 
it is possible to show that $\wassersteinbot$ is 
inductive, the relevant inequality being
$$
\wassersteinbot\vrelone(\sup_n \mu_n, \nu)
\leq \sup_n \wassersteinbot\vrelone(\mu_n, \nu).
$$
Proving $\signature$-compatibility of $\wasserstein$ and $\wassersteinbot$ 
amounts to prove 
$$
\vrelator\vrelone(\mu_1 \oplus_p \nu_1, \mu_2 \oplus_p \nu_2) \leq
\vrelator \vrelone(\mu_1,\mu_2) \oplus_p
\vrelator \vrelone(\nu_1, \nu_2),
$$ 
which is straightforward. 
\end{example}

\paragraph{From $\quantale$-relators to $\two$-relators}
\label{paragraph:from-v-relators-to-two-relators} 
Before applying the abstract theory of $\quantale$-relators to 
$\quantale$-\Fuzz\ we show how a $\quantale$-relator induces a 
canonical $\two$-relator (this will be useful in the next section). 
Consider the maps:
\begin{center}
\begin{tabular}{cc}
$\varphi  : \quantale \to \two$ &   
$\psi  : \two \to \quantale$  
\\
$\qunit \mapsto \true,\ \qvalone \mapsto \false$ &
$\true \to \qunit,\ \false \to \qbot$
\end{tabular}
\end{center}
We immediately see that $\varphi$ and $\psi$ are
CBFs and that $\varphi$ is the right adjoint of $\psi$.
We associate to every $\quantale$-relation $\vrelone$  
its kernel $\two$-relation $\varphi \acts \vrelone$ and
to any $\two$-relation $\relone$ the
$\quantale$-relation $\psi \acts \relone$. 
Similarly, we can associate to each 
$\quantale$-relator $\vrelator$ the $\two$-relator 
$\vrelatortwo_\vrelator\relone \defeq \varphi \acts \vrelator(\psi \acts \relone).$ 
Moreover, since
$\varphi$ is the right adjoint of $\psi$ we have the inequalities:
\begin{align*}
\psi \acts \vrelatortwo_\vrelator \relone   &\leq \vrelator(\psi \acts \relone) \\
\vrelatortwo_\vrelator(\varphi \acts \vrelone) & \leq \varphi \acts \vrelator\vrelone.
\end{align*}

Finally, we say that $\vrelator$ is compatible with $\varphi$ if 
$\vrelatortwo_\vrelator(\varphi \acts \vrelone) = \varphi \acts \vrelator\vrelone$
holds for any $\vrelone: X \torel Y$.

\begin{example}
\label{ex:simulation-relators-from-v-simulation-relators}
\begin{enumerate}[wide = 0pt, leftmargin = *]
  \item For the $\quantale$-relator $(-)_\bot$ and $\relone: X \torel Y$ 
    we have 
    $\vrelatortwo_\bot\relone(\mathpzc{x}, \mathpzc{y}) = \true$ if and 
    only if $\mathpzc{x} \in X, \mathpzc{y} \in Y$ and 
    $\relone(\mathpzc{x},\mathpzc{y}) = \true$, or $\mathpzc{x} = \bot$. 
    That is, $\vrelatortwo_\bot$ gives the usual simulation relator 
    for `effect-free' $\lambda$-calculi. 
    An easy calculation shows that 
    $\vrelatortwo_\bot(\varphi \acts \vrelone) = \varphi \acts \vrelone_\bot$. 
    Replacing $(-)_\bot$ with $(-)_{\bot\bot}$ we recover the bisimulation 
    relator for `effect-free' $\lambda$-calculi.
  \item For the $\quantale$-relator $\hausdorff$ and $\relone: X \torel Y$ 
    we have: 
    $$\vrelatortwo_\hausdorff\relone(\mathpzc{X}, \mathpzc{Y}) = \true 
    \iff \forall x \in \mathpzc{X}.\ \exists y \in \mathpzc{Y}.\ 
    \relone(x,y) = \true.
    $$ 
    Therefore, $\vrelatortwo_\hausdorff$ gives the usual notion of simulation 
    for nondeterministic systems. 
    Proving compatibility with $\varphi$, i.e.
    $\vrelatortwo_\hausdorff(\varphi \acts \vrelone) = \varphi \acts \hausdorff\vrelone$, 
    is straightforward. 
    A similar argument holds for $\hausdorffsym$.
  \item Consider the Wasserstein lifting $\wasserstein$ and observe that 
    we have
    $
    \vrelatortwo_{\wasserstein}\relone(\mu, \nu) = \true$ if and only if
    the following holds:
    $$
    \exists \omega \in \couplings(\mu,\nu).\ 
    \forall x,y.\ \omega(x,y) > 0 \implies \relone(x,y) = \true.
    $$
    We have thus recovered the usual notion of probabilistic relation lifting 
    via couplings \cite{Kurz/Tutorial-relation-lifting/2016}. 
    Moreover, if $\varphi \acts \wasserstein\vrelone(\mu,\nu) = \true$, 
    then $\wasserstein\vrelone(\mu,\nu) = 0$, meaning that there exists a coupling 
    $\omega \in \couplings(\mu, \nu)$ such that 
    $\sum_{x,y} \omega(x,y) \cdot \vrelone(x,y) = 0$. In particular,  
    if $\omega(x,y) > 0$, then $\vrelone(x,y) = 0$ i.e. 
    $(\varphi \acts \vrelone)(x,y) = \true$. That is, $\wasserstein$ is 
    compatible with $\varphi$. From point $1$ it follows that 
    $\wassersteinbot$ is compatible with 
    $\varphi$ as well.
\end{enumerate}
\end{example}

We conclude this section with the following auxiliary lemma 
(whose proof is given in Appendix \ref{appendix:proofs-behavioural-v-relations}), 
which will be useful to prove that the kernel of applicative distances are 
suitable applicative (bi)simulations. 

\begin{restatable}{lemma}{kernelLemma}
\label{lemma:kernel-lemma}
Let $\vrelator$ be $\quantale$-relator compatible with 
$\varphi$. Then the following hold:
  \begin{align*}
  \vcenter{
     \xymatrix{
     \laxcommute
      X   
      \ar[r]^-{f} 
      \ar[d]_{\vrelone}|-*=0@{|}   
      &  
      \monad \setthree        
      \ar[d]^{\vrelator \vreltwo}|-*=0@{|}
      \\
      \settwo   
      \ar[r]_-{g}   
      &  
      \monad \setfour 
      } }
  &\implies 
      \vcenter{
      \xymatrix{
      \laxcommute
      X   
      \ar[r]^-{f} 
      \ar[d]_{\varphi \acts \vrelone}|-*=0@{|}   
      &  
      \monad \setthree        
      \ar[d]^{\vrelatortwo_{\vrelator} (\varphi \acts \vreltwo)}|-*=0@{|}
      \\
      \settwo   
      \ar[r]_-{g}   
      &  
      \monad \setfour 
      } },
      \\
        \vcenter{
     \xymatrix{
     \laxcommute
      X   
      \ar[r]^-{f} 
      \ar[d]_{\relone}|-*=0@{|}   
      &  
      \monad \setthree        
      \ar[d]^{\vrelatortwo_{\vrelator} \reltwo}|-*=0@{|}
      \\
      \settwo   
      \ar[r]_-{g}   
      &  
      \monad \setfour 
      } }
  &\implies 
      \vcenter{
      \xymatrix{
      \laxcommute
      X   
      \ar[r]^-{f} 
      \ar[d]_{\psi \acts \relone}|-*=0@{|}   
      &  
      \monad \setthree        
      \ar[d]^{\vrelator (\psi \acts \reltwo)}|-*=0@{|}
      \\
      \settwo   
      \ar[r]_-{g}   
      &  
      \monad \setfour 
      } }.
  \end{align*}
\end{restatable}

\section{Behavioural $\quantale$-relations}
\label{section:behavioural-v-relations}
 
In this section we extend the relational theory developed in 
e.g. \cite{Lassen/PhDThesis,Gordon/FOSSACS/01} for higher-order 
functional languages to $\quantale$-relations for $\quantale$-\Fuzz. 
Following \cite{Pitts/ATBC/2011} we refer to such relations as 
\emph{$\lambda$-term $\quantale$-relations}. 
Among such $\quantale$-relations we define 
applicative $\vrelator$-similarity,
the generalisation of Abramsky's applicative similarity  to 
both algebraic effects and $\quantale$-relations, and 
prove that under suitable conditions it is compatible 
generalised metric. We postpone the study of applicative 
$\vrelator$-bisimilarity to Section 
\ref{section:from-applicative-v-similarity-to-applicative-v-bisimilarity}.
As usual we assume a signature $\signature$, 
a $\signature$-quantale $\quantale$,
a collection of CBEs $\Pi$ (according to 
Section \ref{section:v-fuzz}), and a $\signature$-continuous 
(strong) monad $\Monad$ to be fixed. We also assume 
$\quantale$-relators to satisfy all requirements given in Section 
\ref{section:v-relators-and-v-relation-lifting}. 

\begin{definition}
  A \emph{closed} $\lambda$-term $\quantale$-relation  
  $\vrelone = (\toterm{\vrelone}, \toval{\vrelone})$ 
  associates to each closed type $\typeone$, binary $\quantale$-relations
  $\toval{\vrelone}_\typeone, \toterm{\vrelone}_\typeone$ on closed values 
  and terms inhabiting it, respectively. 
\end{definition}

Since the syntactic shape of expressions 
determines whether we are dealing with terms or 
values, oftentimes we will write 
$\vrelone_\typeone(\termone, \termtwo)$ (resp. 
$\vrelone_\typeone(\valone, \valtwo)$) in place of 
$\toterm{\vrelone}_\typeone(\termone, \termtwo)$ (resp. 
$\toval{\vrelone}_\typeone(\valone, \valtwo)$). 

In order to be able to work with open 
terms we introduce the notion of 
\emph{open $\lambda$-term $\quantale$-relation}.

\begin{definition}
  An \emph{open} $\lambda$-term $\quantale$-relation 
  $\vrelone$
  associates to each (term) sequent 
  $\envone \compimp \typeone$ a $\quantale$-relation 
  $\envone \compimp \vrelone(-, -): \typeone$ on 
  terms inhabiting it, and to each value sequent 
  $\envone \valimp \typeone$ a $\quantale$-relation 
  $\envone \valimp \vrelone(-, -): \typeone$ on 
  values inhabiting it. We require open $\lambda$-term 
  $\quantale$-relations to be closed under weakening, 
  i.e. for any environment $\envtwo$ we require:
  \begin{align*}
  (\envone \compimp \vrelone(\termone, \termtwo): \typeone) 
  &\leq 
  (\envone \tensor \envtwo \compimp \vrelone(\termone, \termtwo): \typeone),
  \\
  (\envone \valimp \vrelone(\valone, \valtwo): \typeone)
  &\leq 
  (\envone \tensor \envtwo \valimp \vrelone(\valone, \valtwo): \typeone).
  \end{align*}
\end{definition}
As for closed $\lambda$-term $\quantale$-relations, we 
will often write $\envone \imp \vrelone(\valone, \valtwo): \typeone$ in 
place of $\envone \valimp \vrelone(\valone, \valtwo): \typeone$ and 
simply refer to open $\lambda$-term $\quantale$-relations as 
$\lambda$-term $\quantale$-relations (whenever relevant we will 
explicitly mention whether we are dealing with open or closed 
$\lambda$-term $\quantale$-relations).

\begin{example}
Both the discrete and the indiscrete $\quantale$-relations are 
open $\lambda$-term $\quantale$-relations. 
The discrete $\lambda$-term $\quantale$-relation is defined by:
$$
\envone \imp \mathsf{disc}(\termone, \termone): \typeone 
\defeq \qunit, 
\qquad 
\envone \imp \mathsf{disc}(\termone, \termtwo): \typeone 
\defeq \bot,
$$ 
(and similarly for values), 
whereas the indiscrete $\lambda$-term $\quantale$-relation is 
defined by 
$$\envone \imp \mathsf{indisc}(\termone, \termtwo): \typeone \defeq \qunit$$
(and similarly for values).
\end{example}

We notice that the collection 
of open $\lambda$-term $\quantale$-relations carries
a complete lattice structure (with respect to the pointwise order), 
meaning that we can define $\lambda$-term $\quantale$-relation both 
inductively and coinductively.

We can always extend a closed $\lambda$-term $\quantale$-relation 
$\vrelone = (\toterm{\vrelone}, \toval{\vrelone})$ to an open 
one. 

\begin{definition}
Let $\envone \defeq \varone_1 :_{\baseone_1} \typeone_1, \hh, 
\varone_n :_{\baseone_n} \typeone_n$ be an environment. For values 
$\vec{\valone} \defeq \valone_1, \hh, \valone_n$ we write 
$\vec{\valone} : \envone$ 
if for any $i \leq n$, 
$\emptyctx \valimp \valone_i : \typeone_i$ holds.
Given a closed $\lambda$-term $\quantale$-relation 
$\vrelone = (\toterm{\vrelone}, \toval{\vrelone})$ 
we define its open extension 
$\open{\vrelone}$ as follows\footnote{The superscript is the 
letter `o' (for open), and should not be confused with 
$\circ$ which we use for the map $\dual{-}$ sending a 
$\quantale$-relation to its dual.}:
\begin{align*}
\envone \imp \open{\vrelone}(\termone, \termtwo): \typetwo 
&\defeq
 \bigwedge\nolimits_{\vec{\valone}: \envone} \toterm{\vrelone}_\typetwo
(\substcomp{\termone}{\vec{\varone}}{\vec{\valone}},
\substcomp{\termtwo}{\vec{\varone}}{\vec{\valone}}) 
\\
\envone
\valimp \open{\vrelone}(\valone, \valtwo): \typetwo
& \defeq
\bigwedge\nolimits_{\bar{\valthree}:\envone} \toval{\vrelone}_\typetwo
(\substval{\valone}{\vec{\valthree}}{\vec{\varone}},
\substval{\valtwo}{\vec{\valthree}}{\vec{\varone}}).
\end{align*}
\end{definition}

We now define \emph{applicative $\vrelator$-similarity}.

\begin{definition}\label{def:v-applicative-simulation}
Let $\vrelator$ be a $\quantale$-relator and 
$\vrelone = (\toterm{\vrelone}, \toval{\vrelone})$ be a closed 
$\lambda$-term $\quantale$-relation. 
Define the closed $\lambda$-term $\quantale$-relation 
$[\vrelone] = (\toterm{[\vrelone]}, \toval{[\vrelone]})$ 
as follows:
    \begin{align*}
        \toterm{[\vrelone]}_{\typeone}(\termone, \termtwo) 
        &\defeq 
        \vrelator\toval{\vrelone}_{\typeone} 
        (\sem{\termone}, \sem{\termtwo}), 
        \\
        \toval{[\vrelone]}_{\typeone \mmap \typetwo}
        (\valone, \valtwo)
        &\defeq 
        \meet\nolimits_{\valthree \in \values_\typeone} 
        \toterm{\vrelone}_{\typetwo}
        (\valone \valthree, \valtwo \valthree),
        \\
        \toval{[\vrelone]}_{\sumtype{i \in I}{\typeone_i}}
        (\inject{\hat \imath}{\valone}, \inject{\hat \imath}{\valtwo}) 
        &\defeq  \toval{\vrelone}_{\typeone_{\hat \imath}}(\valone, \valtwo),
        \\
        \toval{[\vrelone]}_{\sumtype{i \in I}{\typeone_i}}
        (\inject{\hat \imath}{\valone}, \inject{\hat \jmath}{\valtwo}) 
        &\defeq  \qbot,
        \\
        [\vrelone]_{\recType{\typeVar}{\typeone}}
        (\fold{\valone}, \fold{\valtwo})
        &\defeq  
        \vrelone_{\substType{\typeone}{\typeVar}{\recType{\typeVar}{\typeone}}}
        (\valone, \valtwo),
        \\
        [\vrelone]_{\bang_{\baseone} \typeone}
        (\bang \valone, \bang \valtwo)
        &\defeq  
        (\baseone \acts \vrelone_{\typeone})(\valone, \valtwo).
    \end{align*}
(notice that the definition of $\toval{[\vrelone]}$ 
is by case analysis on $\emptyctx \valimp \valone, \valtwo: \typeone$).
A $\lambda$-term $\quantale$-relation $\vrelone$ is an 
applicative $\vrelator$-simulation if $\vrelone \leq [\vrelone]$.
\end{definition}
The clause for $\typeone \mmap \typetwo$ generalises the usual applicative 
clause, whereas the clause for $\bang_\baseone \typeone$ `scale' 
$\toval{\vrelone}_\typeone$ by $\baseone$.
It is easy to see that the above definition induces a map 
$\vrelone \mapsto [\vrelone]$ on the complete lattice of 
closed $\lambda$-term $\quantale$-relations. 
Moreover, such map is monotone since both $\vrelator$ and 
CBEs are.  

\begin{definition}\label{def:v-applicative-similarity}
  Define applicative $\vrelator$-similarity $\vsim$ as the 
  greatest fixed point of $\vrelone \mapsto [\vrelone]$. That is, 
  $\vsim$ is the greatest (closed) $\lambda$-term $\quantale$-relation 
  satisfying the equation $\vrelone = [\vrelone]$ 
  (such greatest solution exists by the Knaster-Tarski Theorem).
\end{definition} 

Applicative $\vrelator$-similarity comes with an associated 
coinduction principle: 
for any closed $\lambda$-term $\quantale$-relation $\vrelone$,
if $\vrelone \leq [\vrelone]$, then $\vrelone \leq \vsim$.

\begin{example}\label{ex:probabilistic-applicative-similarity-distance}
Instantiating Definition \ref{def:v-applicative-similarity} with the 
Wasserstein lifting $\wassersteinbot$  
we obtain the quantitative analogue of 
\emph{probabilistic applicative similarity} 
\cite{DalLagoSangiorgiAlberti/POPL/2014} for $P$-\Fuzz.
In particular, for two terms $\termone, \termtwo \in \Lambda_\typeone$, 
$\vsim(\termone, \termtwo)$ is (for readability we omit subscripts):
\begin{equation*}
\begin{split}
& \min_{\omega  \in \couplings{(\sem{\termone}, \sem{\termtwo})}}
 \sum_{\valone, \valtwo \in \values}
\omega(\valone,\valtwo) \cdot \toval{\vsim}(\valone,\valtwo) 
+ \sum_{\valone \in \values} 
\omega(\valone, \mbot) \cdot \toval{\vsim}_{\mbot}(\valone, \mbot)
\\
&+ 
\sum_{\valtwo \in \values} 
\omega(\mbot, \valtwo) \cdot \toval{\vsim}_{\mbot}(\mbot,\valtwo) 
+ 
\omega(\mbot, \mbot) \cdot \toval{\vsim}_{\mbot}(\mbot,\mbot).
\end{split}
\end{equation*}

The above formula can be simplified observing that we have 
$\toval{\vsim}_{\mbot}(\mbot,\mbot) = 0$, 
$\toval{\vsim}_{\mbot}(\valone, \mbot) = 1$, and 
$\toval{\vsim}_{\mbot}(\mbot,\valtwo) = 0$ by very 
definition of $\vsim_\mbot$.
We immediately notice  
that $\vsim$ is adequate in the following
sense: for all terms $\termone, \termtwo \in \Lambda_\typeone$ 
we have the inequality
$$
\sum \sem{\termone} - \sum \sem{\termtwo}
\leq 
\toterm{\vsim}(\termone, \termtwo),
$$
where $\sum \sem{\termone}$ is the 
probability of convergence of $\termone$, i.e.
$\sum_{\valone \in \values} \sem{\termone}(\valone)$, 
and subtraction is actually truncated subtraction.

Let us now consider terms $I,\Omega \in \Lambda_{\typeone \mmap \typeone}$ 
of Example \ref{ex:v-fuzz-terms}. 
We claim that
$\toterm{\vsim}(I,I \oplus \Omega) = \frac{1}{2}$. By adequacy 
we immediately see that 
$\frac{1}{2} \leq \toterm{\vsim}(I,I \oplus \Omega)$. 
We prove $\toterm{\vsim}(I,I \oplus \Omega) \leq \frac{1}{2}$.
Let $\valone \defeq \abs{\varone}{\return{\varone}}$ and consider 
the coupling $\omega$ defined by:
$$\omega(\valone,\valone) 
= \frac{1}{2}, 
\quad \omega(\valone, \mbot) = 
\frac{1}{2}$$
and zero for the rest.
Indeed $\omega$ is a coupling of $\sem{I}$ and $\sem{I \oplus \Omega}$. 
Moreover, by very definition of $\vsim$ and $\wassersteinbot$ 
we have:
$$
\toterm{\vsim}(I, I \oplus \Omega) \leq
\omega(\valone,\valone)
 \cdot \toval{\vsim}(\valone,\valone) + 
\omega(\valone, \mbot).
$$
The right hand side of the above inequality gives exactly $\frac{1}{2}$, 
provided that 
$\toval{\vsim}(\valone,\valone) = 0$. 
This indeed holds in full generality.
\end{example}

\begin{restatable}{proposition}{applicativeVSimilarityReflexivityTransitivity}
\label{prop:applicative-v-similarity-is-reflexive-and-transitive}
  Applicative $\vrelator$-similarity $\vsim$ is a reflexive and 
  transitive $\lambda$-term $\quantale$-relation.
\end{restatable}

\begin{proof}[Proof sketch.]
The proof is by coinduction, showing that both the identity 
$\lambda$-term $\quantale$-relation and $\vsim \comp \vsim$ 
are applicative $\vrelator$-simulations. 
A formal proof is given in Appendix \ref{appendix:proofs-behavioural-v-relations}.
\end{proof}

In light of Example \ref{ex:simulation-relators-from-v-simulation-relators} 
we can look at the kernel 
of $\vsim$ and recover well-known notions of (relational) applicative similarity 
(properly generalised to $\quantale$-\Fuzz).

\begin{restatable}{proposition}{kernelApplicativeSim}
\label{prop:kernel-of-app-v-sim-is-app-sim}
Define applicative $\vrelatortwo_{\vrelator}$-similarity $\preceq$
by instantiating Definition \ref{def:v-applicative-simulation} with the 
$\two$-relator $\vrelatortwo_{\vrelator}$ and replacing the clause for 
types of the form $\bang_\baseone \typeone$ as follows:
$\bang \valone\ \relone_{\bang_\baseone \typeone}\ \bang \valtwo $
implies
$(\varphi \comp \baseone \comp \psi) \acts \relone_\typeone(\valone,\valtwo)$.
Then the kernel $\varphi \acts \vsim$ of $\vsim$ coincide with 
$\preceq$. 
\end{restatable}

\begin{proof}[Proof sketch]
By coinduction (and using Lemma \ref{lemma:kernel-lemma}) one shows 
that $\varphi \acts \vsim$ is an 
applicative $\vrelatortwo_\vrelator$-simulation and that 
$\psi \acts \preceq$ is an applicative $\vrelator$-simulation. 
A detailed proof is given 
in Appendix \ref{appendix:proofs-behavioural-v-relations}.
\end{proof}

Note that if $\relone_\typeone(\valone,\valtwo)$ holds, then so does 
$(\varphi \comp \baseone \comp \psi) \acts \relone_\typeone(\valone,\valtwo)$, 
but the vice-versa does not necessarily hold. For instance, 
taking $\baseone \defeq 0$ we see that 
$$
(\varphi \comp 0 \comp \psi) \acts \relone_\typeone(\valone, \valtwo) = 
\varphi(0(\psi(\false))) = \varphi(0 \cdot \infty) = \varphi(0) = \true,
$$
which essentially means we identify distinguishable values if they are 
not used. 
Nonetheless, the reader should notice that 
the encoding of a `standard' $\lambda$-calculus $\Lambda$ in $\quantale$-\Fuzz\ can 
be obtained via the usual encoding of $\Lambda$ in its linear refinement 
\cite{Wadler/CbN-CbV-linear-lambda-calculus/1999} which 
corresponds to the fragment of $\quantale$-\Fuzz\ based on CBEs 
$\baseid$ and $\infty$, thus avoiding the above undesired 
result. 

Finally, we introduce the notion of \emph{compatibility} 
which captures a form of Lipshitz-continuity with respect to $\quantale$-\Fuzz\ 
constructors. It is useful to follow \cite{Lassen/PhDThesis} and 
define compatibility via the notion of \emph{compatible refinement}.

\begin{definition}\label{def:compatible-refinement}
  The \emph{compatible refinement} $\refine{\vrelone}$ of 
  an open $\lambda$-term $\quantale$-relation $\vrelone$ is defined by:
  \begin{align*}
  (\envone \imp \refine{\vrelone}(\termone, \termtwo): \typeone) 
  &\defeq \join \{\qvalone \mid \envone \howeimp \qvalone 
  \leq \refine{\vrelone}(\termone, \termtwo): \typeone\},
  \\
  (\envone \valimp \refine{\vrelone}(\valone, \valtwo): \typeone) 
  &\defeq \join \{\qvalone \mid \envone \howeimpval \qvalone \leq 
  \refine{\vrelone}(\valone, \valtwo): \typeone\},
  \end{align*}
  where judgments  
  $\envone \howeimp \qvalone \leq \refine{\vrelone}(\termone, \termtwo): \typeone$ 
  and 
  $\envone \howeimpval \qvalone \leq \refine{\vrelone}(\valone, \valtwo): \typeone$
  are inductively defined for $\qvalone \in \quantale$, 
  $\envone \compimp \termone, \termtwo : \typeone$, and 
  $\envone \valimp \valone, \valtwo : \typeone$  by 
  rules in Figure \ref{fig:compatible-refinement}.
  We say that $\vrelone$ is compatible if $\refine{\vrelone} \leq \vrelone$.
\end{definition}

\begin{figure*}[htbp]
\hrule 

\vspace{0.2cm}

\(
\infer
  {\envone, \varone :_\baseone \typeone \refineimp 
  \qunit \leq \refine{\vrelone}(\varone, \varone): 
  \typeone
  }{}
\qquad
\infer
  {\qop(\envone_1, \hh, \envone_n) \refineimp 
  \qop(\qvalone_1, \hh, \qvalone_n) \leq 
  \refine{\vrelone}(\op(\termone_1, \hh, \termone_n), 
  \op(\termone_1, \hh, \termone_n)): \typeone
  }
  {
  \qvalone_1 \leq \envone_1 \imp \vrelone(\termone_1, \termtwo_1) : \typeone
  &
  \cc
  &
  \qvalone_n \leq \envone_n \imp \vrelone(\termone_n, \termtwo_n) : \typeone
  }
\)

\vspace{0.2cm}

\(
\infer
  {\envone \refineimpval \qvalone \leq \refine{\vrelone}
  (\abs{\varone}{\termone}, \abs{\varone}{\termtwo}): 
  \typeone \mmap \typetwo
  }
  {
  \qvalone \leq \envone, \varone :_{\baseid} \typeone \imp 
  \vrelone(\termone, \termtwo): \typetwo
  }
\qquad
\infer
  {\envone \tensor \envtwo \refineimp \qvalone \tensor \qvaltwo  
  \leq \refine{\vrelone}(\valone \valtwo, \valone' \valtwo') : \typetwo
  }
  {
  \qvalone \leq \envone \valimp \vrelone(\valone, \valone'): 
  \typeone \mmap \typetwo
  &
  \qvaltwo \leq \envtwo \valimp \vrelone(\valtwo, \valtwo') : \typeone
  }
\)

\vspace{0.2cm}

\(
\infer
  {\envone \refineimpval \qvalone \leq 
  \refine{\vrelone}(\inject{\hat \imath}{\valone}, \inject{\hat \imath}{\valtwo}): 
  \sumtype{i \in I}{\typeone_i}
  }
  {
  \qvalone \leq \envone \valimp 
  \vrelone(\valone, \valtwo): 
  \typeone_{\hat \imath}
  }
\qquad
\infer
  {\baseone \comp \envone \tensor \envtwo 
  \refineimp \baseone(\qvalone) \tensor \qvaltwo_{\hat \imath} 
  \leq \refine{\vrelone}
  (\casesum{\inject{\hat \imath}{\valone}}{\termone_i}, 
  \casesum{\inject{\hat \imath}{\valtwo}}{\termtwo_i}): 
  \typetwo
  }
  {
  \qvalone \leq \envone \valimp \vrelone
  (\inject{\hat \imath}{\valone},\inject{\hat \imath}{\valtwo}): 
  \sumtype{i \in I}{\typeone_i}
  &
  \qvaltwo_i \leq \envtwo, \varone:_{\baseone_i} \typeone_i \imp 
  \leq \vrelone(\termone_i, \termtwo_i): \typetwo
  &
  (\forall i \in I)
  }
\)

\vspace{0.2cm}
\(
\infer
  {\envone \refineimp \qvalone \leq 
  \refine{\vrelone}(\return{\valone}, \return{\valtwo}): \typeone
  }
  {
  \qvalone \leq \envone \valimp \vrelone(\valone, \valtwo): \typeone
  }
\qquad
\infer
  {(\baseone \wedge \baseid) \comp \envone \tensor \envtwo \refineimp 
  (\baseone \wedge \baseid)(\qvalone) \tensor \qvaltwo  
  \leq \refine{\vrelone}(\seq{\termone}{\termtwo}, \seq{\termone'}{\termtwo'}) 
  : \typetwo
  }
  {
  \qvalone \leq \envone \imp \vrelone
  (\termone, \termone'): \typeone
  &
  \qvaltwo \leq \envtwo, \varone :_{\baseone} \typeone \imp  
  \vrelone(\termtwo', \termtwo') : \typetwo
  }
\)

\vspace{0.2cm}

\(
\infer
  {\baseone \comp \envone \refineimpval \baseone(\qvalone) 
  \leq \vrelone
  (\bang \valone, \bang \valtwo): \bang_{\baseone} \typeone
  }
  {
  \qvalone \leq \envone \refineimp 
  \vrelone(\valone, \valtwo): \typeone
  }
\qquad
\infer
  {\baseone \comp \envone \tensor \envtwo \refineimp 
  \baseone(\qvalone) \tensor \qvaltwo \leq 
  \refine{\vrelone}
  (\pmbang{\valone}{\termone}, \pmbang{\valtwo}{\termtwo}): \typetwo
  }
  {
  \qvalone \leq \envone \valimp \vrelone(\valone, \valtwo): 
  \bang_{\basetwo} \typeone
  &
  \qvaltwo \leq \envtwo, \varone :_{\baseone \comp \basetwo} \typeone \imp 
  \vrelone(\termone, \termtwo): \typetwo
  }
\)

\vspace{0.2cm}

\(
\infer
  {\envone \refineimpval \qvalone \leq \refine{\vrelone}
  (\fold{\valone}, \fold{\valtwo}) : \recType{\typeVar}{\typeone}
  }
  {
  \qvalone \envone \valimp \vrelone(\valone, \valtwo): 
  \substType{\typeone}{\typeVar}{\recType{\typeVar}{\typeone}}
  }
\qquad
\infer
  {\baseone \comp \envone \tensor \envtwo 
  \refineimp \baseone(\qvalone) \tensor 
  \qvaltwo \leq \refine{\vrelone}
  (\pmfold{\valone}{\termone}, \pmfold{\valtwo}{\termtwo}):\typetwo
  }
  {
  \qvalone \leq \envone \valimp \vrelone(\valone, \valtwo): 
  \recType{\typeVar}{\typeone}
  &
  \qvaltwo \leq \envtwo, \varone :_{\baseone} \substType{\typeone}
  {\typeVar}{\recType{\typeVar}{\typeone}}
  \imp \qvaltwo \leq \vrelone(\termone, \termtwo): \typetwo
  }
\)

\vspace{0.2cm}

\hrule
\caption{Compatible refinement.}
\label{fig:compatible-refinement}
\end{figure*}

It is easy to see that if $\vrelone$ is compatible, then it satisfies 
inequalities in Figure \ref{fig:compatibility-clauses}. Actually, 
$\vrelone$ is compatible precisely if it satisfies the 
inequalities in Figure \ref{fig:compatibility-clauses}.

\begin{figure*}
\hrule \text{ }\\
\begin{align*}
\qunit &\leq (\envone \valimp \vrelone(\varone, \varone): \typeone)
\\
  \envone, \varone :_{\baseid} \typeone \compimp 
  \vrelone(\termone, \termtwo): \typetwo
& \leq 
  \envone \valimp \vrelone(\abs{\varone}{\termone}, 
    \abs{\varone}{\termtwo}): \typeone \mmap \typetwo
\\
(\envone \valimp \vrelone(\valone, \valone') :
  \typeone \mmap \typetwo) \tensor 
  (\envtwo \valimp \vrelone(\valtwo, \valtwo'): \typeone)
& \leq
  (\envone \tensor \envtwo \compimp \vrelone
    (\valone \valtwo, \valone' \valtwo'): \typetwo)
\\
\envone \valimp \vrelone(\valone, \valtwo): \typeone_{\hat \imath} 
&\leq 
  \envone \valimp \vrelone(\inject{\hat \imath}{\valone}, 
  \inject{\hat \imath}{\valtwo}): \sumtype{i \in I} \typeone_i
\\
\baseone \acts (\envone \valimp 
\vrelone(\inject{\hat \imath}{\valone}, \inject{\hat \imath}{\valtwo}):
 \sumtype{i \in I} \typeone_i) \tensor 
(\envtwo, \varone :_{\baseone} \typeone \compimp 
\vrelone(\termone_{\hat \imath}, \termtwo_{\hat \imath}): \typetwo) 
& \leq
\baseone \comp \envone \tensor \envtwo \compimp
\vrelone(\casesum{\inject{\hat \imath}{\valone}}{\termone_i}, 
        \casesum{\inject{\hat \imath}{\valtwo}}{\termtwo_i}): \typetwo
\\
\envone \valimp \vrelone(\valone, \valtwo): \typeone
& \leq 
  \envone \compimp \vrelone(\return{\valone}, \return{\valtwo}): 
  \typeone
\\
  (\baseone \wedge \baseid) \acts 
  (\envone \compimp \vrelone(\termone, \termone'): \typeone) 
  \tensor 
  (\envtwo, \varone :_{\baseone} \typeone \compimp \vrelone
  (\termtwo, \termtwo'): \typetwo )
 &\leq 
  (\baseone \wedge \baseid) \comp \envone \tensor \envtwo \compimp 
  \vrelone(\seq{\termone}{\termtwo}, \seq{\termone'}{\termtwo'}): \typetwo
\\
  \baseone \acts (\envone \valimp \vrelone(\valone, \valtwo): \typeone) 
 & \leq 
  \baseone \comp \envone \valimp \vrelone(\bang{\valone}, \bang{\valtwo}):
  \bang_{\baseone} \typeone
  \\
  \baseone \acts (\envone \valimp \vrelone(\valone, \valtwo):
  \bang_{\basetwo} \typeone) \tensor 
  (\envtwo, \varone :_{\baseone \comp \basetwo} \typeone \compimp 
  \vrelone(\termone, \termtwo): \typetwo) 
&\leq
  \baseone \comp \envone \tensor \envtwo \compimp
  \vrelone(\pmbang{\valone}{\termone}, \pmbang{\valtwo}{\termtwo}): \typetwo
\\
\envone \valimp \vrelone(\valone, \valtwo): 
  \substType{\typeone}{\typeVar}{\recType{\typeVar}{\typeone}}
&\leq 
  \envone \valimp \vrelone(\fold{\valone}, \fold{\valtwo}): 
  \recType{\typeVar}{\typeone}
  \\
\baseone \acts (\envone \valimp \vrelone(\valone, \valtwo):
  \recType{\typeVar}{\typeone}) 
  \tensor 
  (\envtwo, \varone :_{\baseone} 
  \substType{\typeone}{\typeVar}{\recType{\typeVar}{\typeone}} 
  \compimp \vrelone(\termone, \termtwo): \typetwo) 
&\leq
  \baseone \comp \envone \tensor \envtwo \compimp
  \vrelone(\casefold{\valone}{\termone}, \casefold{\valtwo}{\termtwo}): \typetwo
\\
\qop(\envone_1 \compimp \vrelone(\termone_1, \termtwo_1): \typeone 
,\hh, 
  \envone_n \compimp \vrelone(\termone_n, \termtwo_n): \typeone)
&\leq 
  \qop(\envone_1, \hh, \envone_n) \compimp
  \vrelone(\op(\termone_1, \hh, \termone_n), 
  \op(\termtwo_1, \hh, \termtwo_n)): \typeone
\end{align*}
\hrule
\caption{Compatibility clauses.}
\label{fig:compatibility-clauses}
\end{figure*}

Notice that in the clause for sequential composition the presence of 
$\baseone \wedge \baseid$, instead of $\baseone$, ensures that 
for terms like $\termone \defeq \seq{I}{\numeral{0}}$ and 
$\termone' \defeq \seq{\Omega}{\numeral{0}}$, 
the distance $\vrelone(\termone, \termone')$ is determined 
\emph{before} sequencing (which captures the 
idea that although $\numeral{0}$ will not `use' any input, $I$ and 
$\Omega$ will be still evaluated, thus producing observable differences 
between $\termone$ and $\termone'$). In fact, if we replace 
$\baseone \wedge \baseid$ with $\baseone$, then by taking 
$\baseone \defeq 0$ compatibility would imply 
$\vrelone(\termone, \termone') = \qunit$, which is clearly unsound.

In order to make applicative $\vrelator$-similarity a useful tool, 
we need it to allow compositional reasoning about programs. Formally, 
that amount to prove that applicative $\vrelator$-similarity is 
compatible.

\section{Howe's Method}
\label{section:howe-method}

To prove compatibility of applicative $\vrelator$-similarity 
we design a generalisition of the so-called 
Howe's method \cite{Howe/IC/1996} combining and extending 
ideas from 
\cite{CrubilleDalLago/LICS/2015} and \cite{DalLagoGavazzoLevy/LICS/2017}. 
We start by defining the notion of \emph{Howe's extension}, 
a construction extending a $\lambda$-term $\quantale$-open relation 
to a compatible and substitutive $\lambda$-term $\quantale$-relation.

\begin{definition}[Howe's extension (1)]\label{def:howe-extension-one}
  The \emph{Howe's extension} $\howe{\vrelone}$ of 
  an open $\lambda$-term $\quantale$-relation
  $\vrelone$ 
  is defined as the least solution to the equation 
  $\vreltwo = \vrelone \comp \refine{\vreltwo}$.
\end{definition}

It is easy to see that compatible refinement $\refine{-}$ is monotone, 
and thus so is the map $\Phi_{\vrelone}$ defined by 
$\Phi_{\vrelone}(\vreltwo) \defeq \vrelone \comp \refine{\vreltwo}$. 
As a consequence, we can define $\howe{\vrelone}$ as the least fixed point 
of $\Phi_{\vrelone}$. Since open extension $\open{-}$ is monotone as well, 
we can define the Howe's extension of a closed $\lambda$-term $\quantale$-relation 
$\vrelone$ as $\howe{(\open{\vrelone})}$.

It is also useful to spell out the above definition.

\begin{definition}[Howe's extension (2)]\label{def:howe-extension-two}
  The \emph{Howe's extension} $\howe{\vrelone}$ of 
  an open $\lambda$-term $\quantale$-relation $\vrelone$ is defined by:
  \begin{align*}
  (\envone \imp \howe{\vrelone}(\termone, \termtwo): \typeone) 
  &\defeq \join \{\qvalone \mid \envone \howeimp \qvalone 
  \leq \howe{\vrelone}(\termone, \termtwo): \typeone\},
  \\
  (\envone \valimp \howe{\vrelone}(\valone, \valtwo): \typeone) 
  &\defeq \join \{\qvalone \mid \envone \howeimpval \qvalone \leq 
  \howe{\vrelone}(\valone, \valtwo): \typeone\},
  \end{align*}
  where judgments  
  $\envone \howeimp \qvalone \leq \howe{\vrelone}(\termone, \termtwo): \typeone$ 
  and 
  $\envone \howeimpval \qvalone \leq \howe{\vrelone}(\valone, \valtwo): \typeone$
  are inductively defined for $\qvalone \in \quantale$, 
  $\envone \compimp \termone, \termtwo : \typeone$, and 
  $\envone \valimp \valone, \valtwo : \typeone$  by 
  rules in Figure \ref{fig:howe-extension}.
\end{definition}

\begin{figure*}[htbp]
\hrule 
\vspace{0.2cm}

\(
\infer[\ruleHVar]
  {\envone, \varone :_{\baseone} \typeone \howeimp 
  \qvalone \leq \howe{\vrelone}(\varone, \valtwo): 
  \typeone
  }
  {
  \qvalone \leq \envone, \varone :_{\baseone} \typeone 
  \valimp \vrelone(\varone, \valtwo): \typeone
  }
\)

\vspace{0.2cm}

\(
\infer[\ruleHAbs]
  {\envone \howeimp \qvalone \tensor 
  \qvalthree \leq \howe{\vrelone}
  (\abs{\varone}{\termone}, \termtwo): 
  \typeone \mmap \typetwo
  }
  {
  \envone, \varone :_{\baseid} \typeone \howeimp 
  \qvalone \leq \howe{\vrelone}(\termone, \termthree): 
  \typetwo
  &
  \qvalthree \leq \envone \imp \vrelone
  (\abs{\varone}{\termthree}, \termtwo): 
  \typeone \mmap \typetwo
  }
\)

\vspace{0.2cm}

\(
\infer[\ruleHApp]
  {\envone \tensor \envtwo \howeimp \qvalone 
  \tensor \qvaltwo \tensor \qvalthree 
  \leq \howe{\vrelone}(\valone \valtwo, \termtwo) : \typetwo
  }
  {
  \envone \howeimp \qvalone \leq 
  \howe{\vrelone}(\valone, \valone'): 
  \typeone \mmap \typetwo
  &
  \envtwo \howeimp \qvaltwo \leq 
  \howe{\vrelone}(\valtwo, \valtwo') : \typeone
  &
  \qvalthree \leq \envone \tensor \envtwo \imp
  \vrelone(\valone' \valtwo', \termtwo):\typetwo
  }
\)

\vspace{0.2cm}

\(
\infer[\ruleHInject]
  {\envone \howeimpval \qvalone \tensor \qvaltwo \leq 
  \vrelone(\inject{\hat \imath}{\valone}, \valthree): 
  \sumtype{i \in I}{\typeone_i}
  }
  {
  \envone \howeimpval \qvalone \leq 
  \howe{\vrelone}(\valone, \valtwo): 
  \typeone_{\hat \imath}
  &
  \qvaltwo \leq \envone \valimp 
  \vrelone(\inject{\hat \imath}{\valtwo}, \valthree): 
  \sumtype{i \in I}{\typeone_i}
  }
\)

\vspace{0.2cm}

\(
\infer[\ruleHCaseSum]
  {\baseone \comp \envone \tensor \envtwo 
  \howeimp \baseone(\qvalone) \tensor \qvaltwo_{\hat \imath} 
  \tensor \qvalthree \leq \howe{\vrelone}
  (\casesum{\inject{\hat \imath}{\valone}}{\termone_i}, \termthree): 
  \typetwo
  }
  {
  \envone \howeimpval \qvalone \leq \howe{\vrelone}
  (\inject{\hat \imath}{\valone},\inject{\hat \imath}{\valtwo}): 
  \sumtype{i \in I}{\typeone_i}
  &
 \forall i \in I.\  \envtwo, \varone:_\baseone \typeone_{i} \howeimp 
  \qvaltwo_{i} \leq 
  \howe{\vrelone}(\termone_{i}, \termtwo_{i}): \typetwo
  &
  \qvalthree \leq \baseone \comp \envone \tensor \envtwo \compimp 
  \vrelone(\casesum{\inject{\hat \imath}{\valtwo}}{\termtwo_i}, \termthree): 
  \typetwo
  }
\)

\vspace{0.2cm}

\(
\infer[\ruleHReturn]
  {\envone \howeimp \qvalone \tensor \qvalthree \leq 
  \howe{\vrelone}(\return{\valone}, \termtwo): \typeone
  }
  {
  \envone \howeimp \qvalone \leq \howe{\vrelone}
  (\valone, \valtwo): \typeone
  &
  \qvalthree \leq \envone \imp 
  \vrelone(\return{\valtwo}, \termtwo): \typeone
  }
\)

\vspace{0.2cm}

\(
\infer[\ruleHSeq]
  {(\baseone \wedge \baseid) \comp \envone \tensor \envtwo \howeimp 
  (\baseone \wedge \baseid)(\qvalone) \tensor \qvaltwo \tensor \qvalthree 
  \leq \howe{\vrelone}(\seq{\termone}{\termone'}, \termtwo) : \typetwo
  }
  {
  \envone \howeimp \qvalone \leq \howe{\vrelone}
  (\termone, \termthree): \typeone
  &
  \envtwo, \varone :_{\baseone} \typeone \howeimp \qvaltwo \leq 
  \howe{\vrelone}(\termone', \termthree') : \typetwo
  &
  \qvalthree \leq (\baseone \wedge \baseid) \comp \envone \tensor \envtwo \imp 
  \vrelone(\seq{\termthree}{\termthree'}, \termtwo): \typetwo
  }
\)

\vspace{0.2cm}

\(
\infer[\ruleHBang]
  {\baseone \comp \envone \howeimp \baseone(\qvalone) 
  \tensor \qvalthree \leq \howe{\vrelone}
  (\bang \valone, \valfour): \bang_{\baseone} \typeone
  }
  {
  \envone \howeimp \qvalone \leq 
  \howe{\vrelone}(\valone, \valtwo): \typeone
  &
  \qvalthree \leq \baseone \comp \envone \imp \vrelone
  (\bang \valtwo, \valfour): \bang_{\baseone} \typeone
  }
\)

\vspace{0.2cm}

\(
\infer[\ruleHPmBang]
  {\baseone \comp \envone \tensor \envtwo \howeimp 
  \baseone(\qvalone) \tensor \qvaltwo \tensor \qvalthree \leq 
  \howe{\vrelone}(\pmbang{\valone}{\termone}, \termtwo): \typetwo
  }
  {
  \envone \howeimp \qvalone \leq \howe{\vrelone}(\valone, \valtwo): 
  \bang_{\basetwo} \typeone
  &
  \envtwo, \varone :_{\baseone \comp \basetwo} \typeone \howeimp 
  \qvaltwo \leq \howe{\vrelone}(\termone, \termthree): \typetwo
  &
  \qvalthree \leq \baseone \comp \envone \tensor \envtwo 
  \imp \vrelone(\pmbang{\valtwo}{\termthree}, \termtwo): \typetwo
  }
\)

\vspace{0.2cm}

\(
\infer[\ruleHPmFold]
  {\baseone \comp \envone \tensor \envtwo 
  \howeimp \baseone(\qvalone) \tensor 
  \qvaltwo \tensor \qvalthree \leq \howe{\vrelone}
  (\pmfold{\valone}{\termone}, \termtwo):\typetwo
  }
  {
  \envone \howeimp \qvalone \leq 
  \howe{\vrelone}(\valone, \valtwo): 
  \recType{\typeVar}{\typeone}
  &
  \envtwo, \varone :_{\baseone} \substType{\typeone}
  {\typeVar}{\recType{\typeVar}{\typeone}}
  \howeimp \qvaltwo \leq \howe{\vrelone}
  (\termone, \termthree): \typetwo
  &
  \qvalthree \leq \baseone \comp \envone \tensor \envtwo \imp 
  \vrelone(\pmfold{\valtwo}{\termthree}, \termtwo): \typetwo
  }
\)

\vspace{0.2cm}

\(
\infer[\ruleHFold]
  {\envone \howeimp \qvalone \tensor \qvalthree \leq \howe{\vrelone}
  (\fold{\valone}, \valfour) : \recType{\typeVar}{\typeone}
  }
  {
  \envone \howeimp \qvalone \leq \howe{\vrelone}(\valone, \valtwo): 
  \substType{\typeone}{\typeVar}{\recType{\typeVar}{\typeone}}
  &
  \qvalthree \leq \envone \imp \vrelone(\fold{\valtwo}, \valfour):
  \recType{\typeVar}{\typeone}
  }
\)

\vspace{0.2cm}

\(
\infer[\ruleHOp]
  {\qop(\envone_1, \hh, \envone_n) \howeimp 
  \qop(\qvalone_1, \hh, \qvalone_n) \tensor 
  \qvalthree \leq \howe{\vrelone}(\op(\termone_1, \hh, \termone_n), 
  \termtwo): \typeone
  }
  {
  \forall i \leq n.\ \envone_i \howeimp \qvalone_i \leq 
  \howe{\vrelone}(\termone_i, \termthree_i) : \typeone
  &
  \qvalthree \leq \qop(\envone_1. \hh, \envone_n) 
  \imp \vrelone(\op(\termthree_1, \hh, \termthree_n), \termtwo): 
  \typeone
  }
\)

\vspace{0.2cm}

\hrule
\caption{Howe's extension.}
\label{fig:howe-extension}
\end{figure*}

The next lemma (whose proof is given in Appendix \ref{appendix:proofs-howe-method}) 
is useful for proving properties of Howe's extension. 
It states that $\howe{\vrelone}$ attains 
its value via the rules in Figure \ref{fig:howe-extension}.

\begin{restatable}{lemma}{howeOptimalValue}
\label{lemma:howe-optimal-value}
The following hold:
\begin{varenumerate}
\item Given well-typed values 
  $\envone \valimp \valone, \valtwo: \typeone$, 
  let 
  $$
  A \defeq \{\qvalone \mid \envone \howeimpval 
  \qvalone \leq \howe{\vrelone}(\valone, \valtwo): \typeone\}
  $$ 
  be non-empty. Then
  $\envone \howeimpval \join A \leq \howe{\vrelone}(\valone, \valtwo)$
  is derivable.
\item Given well-typed terms 
  $\envone \compimp \termone, \termtwo: \typeone$, 
  let 
  $$A \defeq \{\qvalone \mid \envone \howeimpcomp 
  \qvalone \leq \howe{\vrelone}(\termone, \termtwo): 
  \typeone\}
  $$ 
  be non-empty. 
  Then $\envone \howeimpcomp \join A \leq 
  \howe{\vrelone}(\termone, \termtwo)$ is derivable.
\end{varenumerate} 
\end{restatable} 

It is easy to see that Definition \ref{def:howe-extension} and 
\ref{def:howe-extension-two} gives the same $\lambda$-term 
$\quantale$-relation. In particular, 
for an open $\lambda$-term $\quantale$-relation $\vrelone$, 
$\howe{\vrelone}$ is the least compatible open $\lambda$-term $\quantale$-relation
satisfying the inequality $\vrelone \comp \vreltwo \leq \vreltwo$.

The following are standard results on Howe's extension. Proofs are 
straightforward but tedious (they closely resemble their relational 
counterparts), and thus are omitted.

\begin{lemma}\label{lemma:properties-howe-extension}
Let $\vrelone$ be a reflexive and transitive open 
$\lambda$-term $\quantale$-relation. 
Then the following hold:
\begin{varenumerate}
  \item $\howe{\vrelone}$ is reflexive.
  \item $\vrelone \leq \howe{\vrelone}$.
  \item $\vrelone \comp \howe{\vrelone} \leq \howe{\vrelone}$.
  \item $\howe{\vrelone}$ is compatible.
\end{varenumerate}
\end{lemma}
We refer to property 
$1$ as pseudo-transitivity. In particular, by very 
definition of $\quantale$-relator we also have 
$\vrelator \vrelone \comp \vrelator \howe{\vrelone} \leq \vrelator \howe{\vrelone}$. 
We refer to the latter property as $\vrelator$-pseudo-transitivity.
Notice that Proposition 
\ref{prop:applicative-v-similarity-is-reflexive-and-transitive} 
implies that $\howe{(\open{\vsim})}$ is compatible and bigger than $\open{\vsim}$. 

Finally, Howe's extension enjoys another remarkable property, namely substitutivity.

\begin{definition}\label{def:value-substitutive}
An open $\lambda$-term $\quantale$-relation $\vrelone$ is value substitutive if for all 
well-typed values 
$\envone, \varone :_{\baseone} \typeone \valimp \valone, \valtwo: 
\typetwo$, $\emptyset \valimp \valthree : \typeone$, and terms
$\envone, \varone :_{\baseone} \typeone \imp \termone, \termtwo: 
\typetwo$ we have:
\begin{align*}
(\envone, \varone :_{\baseone} \typeone \valimp \vrelone
(\valone, \valtwo): \typetwo) 
&\leq (\envone \imp \vrelone(\substval{\valone}{\valthree}{\varone}, 
\substval{\valtwo}{\valthree}{\varone}): \typetwo),
\\
(\envone, \varone :_{\baseone} \typeone \imp \vrelone
(\termone, \termtwo): \typetwo) 
&\leq (\envone \imp \vrelone(\substcomp{\termone}{\varone}{\valthree}, 
\substcomp{\termtwo}{\varone}{\valthree}): \typetwo).
\end{align*}
\end{definition}

\begin{restatable}[Substitutivity]{lemma}{substitutivityLemma} 
\label{lemma:substitutivity}
  Let $\vrelone$ be a value substitutive 
  $\lambda$-term $\quantale$-preorder. 
  For all values, 
  $\envone, \varone :_{\baseone} \typeone \valimp \valthree, \valfour: 
  \typetwo$ and $\emptyset \imp \valone, \valtwo : \typeone$, 
  and terms  $\envone, \varone :_{\baseone} \typeone \imp 
  \termone, \termtwo: \typetwo$, let 
  $\underline{\qvalone} \defeq 
  \emptyctx \valimp \howe{\vrelone}(\valone, \valtwo): \typeone$. Then:
  \begin{align*}
  (\envone, \varone :_{\baseone} \typeone \valimp 
  \howe{\vrelone}(\valthree, \valfour): \typetwo) 
  \tensor 
  \baseone(\underline{\qvalone})
  &\leq 
  \envone \valimp \howe{\vrelone}(\substval{\valthree}{\valone}{\varone}, 
  \substval{\valfour}{\valtwo}{\varone}): \typetwo,
  \\
  (\envone, \varone :_{\baseone} \typeone \imp 
  \howe{\vrelone}(\termone, \termtwo): \typetwo) 
  \tensor 
  \baseone(\underline{\qvalone}) 
  &\leq 
  \envone \imp \howe{\vrelone}(\substcomp{\termone}{\varone}{\valone}, 
  \substcomp{\termtwo}{\varone}{\valtwo}): \typetwo.
  \end{align*}
\end{restatable}

\begin{proof}
See Appendix \ref{appendix:proofs-howe-method}.
\end{proof}

Notice that the open extension of any closed $\lambda$-term $\quantale$-relation 
is value-substitutive. 
We can prove the main result of the Howe's method, the the so-called \emph{Key Lemma}. 
The latter states the Howe's extension of applicative $\vrelator$-similarity 
(restricted to closed terms/values) is an applicative $\vrelator$-simulation. 
By coinduction, we can conclude that $\vsim$ and $\howe{\vsim}$ 
(restricted to closed terms/values) coincide, meaning that the former is 
compatible. 

\begin{restatable}[Key Lemma]{lemma}{keyLemma}
\label{lemma:key-lemma}
  Let $\vrelone$ be a reflexive and transitive 
  applicative $\vrelator$-simulation. 
  Then the Howe's extension of $\vrelone$ restricted to closed 
  terms/values in an applicative $\vrelator$-simulation.
\end{restatable}

\begin{proof}[Proof sketch.]
The proof is non-trivial and a detailed account is given in 
Appendix \ref{appendix:proofs-howe-method}.
Let us write $\howe{\vrelone}$ for the Howe's extension 
of $\vrelone$ restricted to closed terms/values. 
By induction on $n$ one shows that for any $n \geq 0$, 
$
\toterm{(\howe{\vrelone})}_{\typeone}(\termone, \termtwo) \leq 
\relator \toval{(\howe{\vrelone})}_{\typeone}(\approxsem{\termone}{n}, \sem{\termtwo})
$
holds for all terms $\termone, \termtwo \in \Lambda_\typeone$. 
Since $\vrelator$ is inductive, the above inequality indeed gives the thesis. 
The base case follows again by inductivity of $\vrelator$, whereas the 
inductive step requires a case analysis on the structure of 
$\termone$. The crucial case is sequencing, where 
we rely on condition \eqref{s-Strong-Lax-Bind}. 
\end{proof}

From the Key Lemma it directly follows our main result.

\begin{restatable}[Compatibility]{theorem}{applicativeVSimilarityCompatible}
\label{prop:applicative-v-similarity-is-compatible}
Applicative $\vrelator$-similarity is compatible. 
\end{restatable}

\begin{proof}
We have to prove that $\open{\vsim}$ is compatible. By Lemma 
\ref{lemma:properties-howe-extension} we know that
$\open{\vsim} \leq \howe{(\open{\vsim})}$ and that $\howe{(\open{\vsim})}$ 
is compatible. Therefore, to conclude the thesis it is sufficient to 
prove $\howe{(\open{\vsim})} \leq \open{\vsim}$. The Key Lemma implies 
that the restriction on closed 
terms/values of $\howe{(\open{\vsim})}$ is an applicative 
$\vrelator$-simulation, and thus smaller or equal than $\vsim$. 
We can thus show that for all $\envone \imp \termone, \termone': \typeone$, 
the inequality $\envone \imp \howe{(\open{\vsim})}(\termone, \termone'): \typeone 
\leq \envone \imp \open{\vsim}(\termone, \termone'): \typeone$ holds. In fact, 
since $\howe{(\open{\vsim})}$ is substitutive and thus value 
substitutive\footnote{Notice that in Definition \ref{def:value-substitutive} 
we substitute \emph{closed} values (in terms and values) 
meaning that simultaneous substitution and sequential 
substitution coincide. In particular, value substitution implies e.g.
$$
(\envone \imp \vrelone
  (\termone, \termtwo'): \typetwo)  
\leq \meet_{\bar{\valone}: \envone} 
\toterm{\vrelone}_\typetwo(\substcomp{\termone}{\bar{\varone}}{\bar{\valone}}, 
\substcomp{\termtwo}{\bar{\varone}}{\bar{\valone}}).
$$.}
we have:
\begin{align*}
\envone \imp \howe{(\open{\vsim})}(\termone, \termone): \typeone 
& \leq 
\meet_{\bar \valone: \envone} \emptyctx \imp \howe{(\open{\vsim})}
(\substcomp{\termone}{\bar \varone}{\bar \valone}, 
\substcomp{\termone'}{\bar \varone}{\bar \valone}): \typeone 
\\ 
&\leq \meet_{\bar \valone} \toterm{\vsim}_\typeone
(\substcomp{\termone}{\bar \varone}{\bar \valone}, 
\substcomp{\termone'}{\bar \varone}{\bar \valone})
\\ 
& = \envone \imp \open{\vsim}(\termone, \termone'): \typeone.
\end{align*}
A similar argument holds for values.
\end{proof}

It is worth noticing that from our results directly follow 
the following generalisation of Reed's and Pierce's 
\emph{metric preservation} 
\cite{Pierce/DistanceMakesTypesGrowStronger/2010,GaboardiEtAl/POPL/2017}.

\begin{corollary}[Metric Preservation (cf. \cite{GaboardiEtAl/POPL/2017})]
\label{cor:metric-preservation}
For any environment $\envone \defeq \varone_1 :_{\baseone_1} \typeone, \hh, 
\varone_n :_{\baseone_n} \typeone$, values $\bar{\valone}, \bar{\valtwo}: \envone$, and 
$\envone \imp \termone: \typeone$ we have:
$$
\baseone_1 \acts \toval{\vsim}_{\typeone_1}(\valone_1, \valtwo_1) 
\tensor \cc \tensor \baseone_n \acts \toval{\vsim}_{\typeone_n}(\valone_n, \valtwo_n) 
\leq \toterm{\vsim}_\typeone(\substcomp{\termone}{\vec{\varone}}{\vec{\valone}}, 
\substcomp{\termone}{\vec{\varone}}{\vec{\valtwo}}).
$$
\end{corollary}

Having proved that applicative $\vrelator$-similarity is a compatible 
generalised metric, we now move to applicative $\vrelator$-bisimilarity.

\section{Applicative $\vrelator$-bisimilarity}
\label{section:from-applicative-v-similarity-to-applicative-v-bisimilarity}

In previous section we proved that 
applicative $\vrelator$-similarity is a compatible 
generalised metric. 
However, in the context of programming language semantics it 
is often desirable to work with equivalence 
$\quantale$-relations---i.e. pseudometrics.
In this section we discuss two natural behavioural pseudometrics: 
applicative $\vrelator$-\emph{bisimilarity} and two-way applicative 
$\vrelator$-similarity. We prove that under suitable conditions on 
CBEs (which are met by all examples we have considered 
so far) both applicative $\vrelator$-\emph{bisimilarity} and two-way applicative 
$\vrelator$-similarity are compatible pseudometrics ($\quantale$-equivalences). 
Proving compatibility of the latter is straightforward. However, proving 
compatibility of applicative $\vrelator$-bisimilarity is not trivial and 
requires a variation of the so-called \emph{transitive closure trick} 
\cite{Howe/IC/1996,Lassen/PhDThesis,Pitts/ATBC/2011} based on ideas in
\cite{Simpson-Niels/Modalities/2018}.

Before entering formalities, let us remark that so far we have 
mostly worked with inequation and inequalities. 
That was fine since we have been interested in non-symmetric $\quantale$-relations. 
However, for symmetric $\quantale$-relations inequalities seem not to be 
powerful enough, and often plain equalities are needed in order to make proofs work.
For that reason in the rest of this section we assume CBFs 
to be monotone \emph{monoid (homo)morphism}. That is, we modify 
Definition \ref{def:change-of-base-functor} requiring the equalities:
$$
h(\qunit) = \ell, \qquad
h(\qvalone \tensor \qvaltwo) = h(\qvalone) \tensor h(\qvaltwo).
$$
Note that we do not require CBEs to be join-preserving 
(i.e. continuous). We also require operations $\qop$ to be 
\emph{quantale (homo)morphism}, i.e. to preserves unit, tensor, and joins. 
It is easy to see that the new requirements are met by all 
examples considered so far. We start with two-way applicative 
$\vrelator$-similarity. 

\begin{proposition}\label{prop:two-way-applicative-similarity-is-compatible}
For a $\quantale$-relator $\vrelator$ define two-way 
applicative $\vrelator$-similarity as 
$\vsim \tensor \dual{\vsim}$. Then two-way applicative 
$\vrelator$-similarity is a compatible $\quantale$-equivalence. 
\end{proposition}

\begin{proof}[Proof sketch.]
Clearly $\vsim \tensor \dual{\vsim}$ is symmetric. Moreover, since 
CBEs are monoid (homo)morphism it is also compatible. 
\end{proof}

We now move to the more interesting case of applicative 
$\vrelator$-bisimilarity. In light of Example \ref{ex:v-relators} 
we give the following definition.

\begin{definition}\label{def:applicative-v-bismilarity}
Recall Proposition \ref{prop:algebra-of-v-relators}. 
Define applicative $\vrelator$-\emph{bisimilarity} $\vbisim$ as 
applicative $(\vrelator \wedge \dual{\vrelator})$-similarity.
\end{definition}

Proposition \ref{prop:applicative-v-similarity-is-reflexive-and-transitive} 
implies that $\vbisim$ is reflexive and transitive. 
Moreover, if CBEs preserve binary meet (a condition 
satisfied by all our examples), i.e. 
$\baseone(\qvalone) \wedge \baseone(\qvaltwo) = 
\baseone(\qvalone \wedge \qvaltwo)$ for any CBE $\baseone$ in $\Pi$, 
then $\vbisim$ is also symmetric, ad thus a pseudometric. 
Finally we observe that $\vbisim$ is the greatest $\lambda$-term 
$\quantale$-relation $\vrelone$ such that both $\vrelone$ and 
$\dual{\vrelone}$ are applicative $\vrelator$-simulation.

Proving compatibility of $\vbisim$ is not straightforward, 
and requires a variation of the so-called \emph{transitive closure trick} 
\cite{Pitts/ATBC/2011}.
First of all we notice that we cannot apply the Key Lemma on 
$\vbisim$ since 
$\vrelator \wedge \dual{\vrelator}$ being conversive is, in general, not inductive.
To overcome this problem, we follow \cite{Simpson-Niels/Modalities/2018} 
and characterise applicative $\vrelator$-bisimilarity differently.

\begin{restatable}{proposition}{symmetricSimilarityIsBisimilarity}
\label{lemma:symmetric-similarity-is-bisimilarity}
Let $\vrelator$ be a $\quantale$-relator. 
Define the $\lambda$-term 
$\quantale$-relation $\vbisim'$ as follows:
$$
\vbisim' \defeq \join \{\vrelone \mid \dual{\vrelone} = \vrelone,\ 
\vrelone \leq [\vrelone]\}.
$$ Then:
\begin{enumerate}
\item $\vbisim'$ is a symmetric applicative $\vrelator$-simulation, 
  and therefore the largest such $\lambda$-term $\quantale$-relation.
\item $\vbisim'$ coincide with applicative 
  $(\vrelator \wedge \dual{\vrelator})$-similarity $\vbisim$.
\end{enumerate}
\end{restatable}

\begin{proof}
See Appendix \ref{appendix:proofs-applicative-v-bisimilarity}.
\end{proof}

Lemma \ref{lemma:symmetric-similarity-is-bisimilarity} allows to 
apply the Key Lemma on $\vbisim$, thus showing that $\howe{\vbisim}$ is 
compatible. However, the Howe's extension is an intrinsically asymmetrical 
construction (cf. pseudo-transitivity) and there is little hope 
to prove symmetry of $\howe{\vbisim}$ (which would imply compatibility of 
$\vbisim$). Nevertheless, we observe that for a suitable class of 
CBEs the transitive closure $\transitive{(\howe{\vbisim})}$ of $\howe{\vbisim}$ 
is a symmetric, compatible, $\vrelator$-simulation (and thus smaller 
than $\vbisim$).

\begin{definition}
We say that a CBE $\baseone$ is 
\emph{finitely continuous}, if $\baseone \neq \infty$ 
implies 
$\baseone(\join A) = \join \{\baseone(\qvalone) \mid \qvalone \in A\},$
for any set $A \subseteq \quantale$. 
\end{definition}

\begin{example}
\label{ex:finitely-continuous-change-of-base-functors}
All concrete CBEs considered 
in previous examples are finitely continuous. Moreover, it is easy to prove the 
all CBEs defined from the CBEs $n, \infty$ of Example \ref{ex:change-of-base-functor} 
using operations in Lemma \ref{lemma:algebra-change-of-base-functors} 
are finitely continuous\footnote{Recall that since $\qvalone$ is integral we have the 
inequality $\qvalone \tensor \bot = \bot$ for any $\qvalone \in \quantale$.}
provided that $\qop(\qvalone_1, \hh, \bot, \hh, \qvalone_n) = \bot$
(which is the case for most of the concrete operations we considered). 
\end{example}

The following is the central result of our argument (see Appendix 
\ref{appendix:proofs-applicative-v-bisimilarity} for a proof).

\begin{restatable}{lemma}{transitiveClosureHoweExtensionCompatible}
\label{lemma:transitive-closure-of-howe-extension-is-compatible}
Assume CBEs in $\Pi$ to be finitely continuous.
Define the transitive closure $\transitive{\vrelone}$ of a 
$\quantale$-relation $\vrelone$ as 
$
\transitive{\vrelone}   \defeq \join_n \vrelone^{(n)},
$
where $\vrelone^{(0)}  \defeq \idvrel$, and 
$\vrelone^{(n+1)} \defeq \vrelone^{(n)} \comp \vrelone$.
\begin{varenumerate}
\item Let $\vrelone$ be a reflexive and transitive $\lambda$-term 
  $\quantale$-relation. Then $\transitive{(\howe{\vrelone})}$ is 
  compatible. 
\item Let $\vrelone$ be an reflexive, symmetric, and transitive 
  open $\lambda$-term $\quantale$-relation. Then 
  $\transitive{(\howe{\vrelone})}$ is symmetric.
\end{varenumerate}
\end{restatable}

Finally, we can prove that applicative $\vrelator$-bisimilarity is 
compatible.

\begin{theorem}\label{thm:applicative-v-bisimilarity-is-compatible}
If any CBE in $\Pi$ is finitely continuous, then applicative 
$\vrelator$-bisimilarity is compatible.
\end{theorem}

\begin{proof}
From Lemma 
\ref{lemma:transitive-closure-of-howe-extension-is-compatible} 
we know that $\transitive{(\howe{\vbisim})}$ is compatible. 
Therefore it is sufficient to prove $\transitive{(\howe{(\vbisim)})} = \vbisim$. 
One inequality follows from Lemma \ref{lemma:properties-howe-extension} 
as follows: $\vbisim \leq \howe{\vbisim} \leq \transitive{(\vbisim)}$. 
For the other inequality we rely on the coinduction proof principle associated 
with $\vbisim$. As a consequence, it is sufficient to prove that 
$\transitive{(\howe{(\vbisim)})}$ is a symmetric applicative 
$\vrelator$-simulation. Symmetry is given by Lemma 
\ref{lemma:transitive-closure-of-howe-extension-is-compatible}.
From Key Lemma we know that $\howe{\vbisim}$ is 
an applicative $\vrelator$-simulation. 
Since the identity $\lambda$-term $\quantale$-relation is a applicative 
$\vrelator$-simulation and that the composition ofapplicative 
$\vrelator$-simulations is itself an applicative 
$\vrelator$-simulation (see the proof of Proposition \ref{prop:applicative-v-similarity-is-reflexive-and-transitive})
we see that $\transitive{(\howe{\vbisim})}$ 
is itself an applicative $\vrelator$-simulation.
\end{proof}

Finally, we notice that all concrete CBEs considered 
in this work are finitely continuous. We can then
rely on Theorem \ref{thm:applicative-v-bisimilarity-is-compatible}
to come up with concrete notions of compatible applicative $\vrelator$-bisimilarity. 
Notably, we obtain compatible pseudometrics for \Fuzz\footnote{
  Formally, we should extend our definitions adding a basic type 
  for real numbers and primitives for arithmetical operations,
  but that is straightforward.} and $P$-\Fuzz.

\section{Further Developments}

In Section \ref{section:howe-method} we proved that applicative 
$\vrelator$-similarity is a compatible $\quantale$-peorder (i.e. 
a compatible generalised metric), whereas in Section 
\ref{section:from-applicative-v-similarity-to-applicative-v-bisimilarity} 
we proved that applicative $\vrelator$-bisimilarity (and two-way similarity) 
is a compatible $\quantale$-equivalence (i.e. a compatible pseudometric)
In this last section we shortly sketch a couple of further considerations 
on the results obtained in this work. 

\paragraph{Contextual distances} 
An issue that has not been touched concerns the quantitative counterpart of 
contextual preorder and contextual equivalence. 
Recently \cite{CrubilleDalLago/LICS/2015,CrubilleDalLago/ESOP/2017} 
define a contextual distance $\vsim^{ctx}$ for 
probabilistic $\lambda$-calculi as:
$$
\vsim^{ctx}(\termone, \termtwo) \defeq 
\sup_{\ctxone} | \sum \sem{\ctxone[\termone]} - \sum \sem{\ctxone[\termtwo]}|,
$$
for contexts and terms of appropriate types. Taking into account 
sensitivity, and thus moving to $P$-Fuzz,
such distance could be refined as 
$$
\vsim^{ctx}(\termone, \termtwo) \defeq 
\sup_{\ctxone} \frac{| \sum \sem{\ctxone[\termone]} - \sum \sem{\ctxone[\termtwo]}|}
{n_{\ctxone}},
$$
where $n_{\ctxone}$ is the sensitivity of $\ctxone$. Here some design 
choices are mandatory in order to deal with division by zero and infinity. 
Two immediate observations are that we would like 
$$
\frac{| \sum \sem{\ctxone[\termone]} - \sum \sem{\ctxone[\termtwo]}|}
{n_{\ctxone}}
$$ to be $0$ if $n_{\ctxone} = 0$ and that 
$$
\frac{| \sum \sem{\ctxone[\termone]} - \sum \sem{\ctxone[\termtwo]}|}
{n_{\ctxone}} = 0
$$ if $n_{\ctxone} = \infty$. That means that we can 
restrict contexts to range over those with sensitivity different from 
$0$ and $\infty$. In particular, excluding the latter means that we 
are considering finitely continuous CBEs. This 
observation (together with the fact that division is the right adjoint 
of multiplication) suggests a possible generalisation of the contextual 
distance to arbitrary quantales. 

Informally, fixed a 
$\lambda$-term $\quantale$-relation (i.e. a ground observation)
$\vrelone_o$ we can define the contextual distance 
$\vrelone_o^{ctx}$ between two (appropriate) terms 
$\termone, \termone'$ as:
$$
\vrelone_o^{ctx}(\termone, \termone') \defeq 
\meet_{\ctxone} \baseone^*(\vrelone_o(\ctxone[\termone], \ctxone[\termone'])),
$$
where $\ctxone$ ranges over contexts\footnote{
    Give a formal definition of $\quantale$-\Fuzz/ requires some (tedious) work. 
    In fact, contexts should be terms with a hole 
    $[-]$ to be filled in with another \emph{term} of appropriate type. 
    However, due to the fine-grained nature of $\quantale$-\Fuzz, we 
    defined substitution of values only. Therefore, what we should do 
    is to define a grammar and a notion of substitution for contexts. 
    Moreover, we should also design 
    a type system for contexts keeping track of sensitivities 
    (see e.g. \cite{Harper-Cary/Syntactical-logical-relations/2007}
    for the relational case). This is a tedious exercise but can be done 
    without difficulties. Here we simply notice that it is possible to 
    `simulate' contexts as follows. Let $\emptyctx \valimp *: \unittype$ 
    be the unit value. Suppose we want to come up with a (closed) context 
    $\ctxone[-]$ of type $\typetwo$ and sensitivity $\baseone$ taking as 
    input terms of type $\typeone$. For that we consider the term 
    (for readability we annotate the lambda):
    $$
    \abs{\vartwo: \bang_\baseone \unittype \mmap \typeone}
    {\casebang{\vartwo}{\ctxone[\vartwo *]}}
    $$
    where $\vartwo$ is a fresh variable.
    To substitute a term $\termone$ of type $\typeone$ in $\ctxone$ 
    we first thunk it to $\abs{}{\termone} \in \terms_{\unittype \mmap \typeone}$ 
    and then consider:  
    $$
    (\abs{\vartwo: \bang_\baseone \unittype \mmap \typeone}
    {\casebang{\vartwo}{\ctxone[\vartwo *]}})(\bang \abs{}{\termone})
    $$
    It is immediate to see that 
    $\sem{(\abs{\vartwo}
    {\casebang{\vartwo}{\ctxone[\vartwo *]}})(\bang \abs{}{\termone})} 
    $ captures $\sem{\ctxone[\termone]}$ (although the expression has 
    not been defined). Moreover, an easy calculation shows that for 
    any compatible $\lambda$-term $\quantale$-relation $\vrelone$, 
    and for all terms $\termone, \termone'$ of type $\typeone$ 
    we have:
    \begin{multline*}
    \baseone \acts \vrelone_\typeone(\termone, \termone') 
    \\ \leq
    \vrelone_\typetwo((\abs{\vartwo}
    {\casebang{\vartwo}{\ctxone[\vartwo *]}})(\bang \abs{}{\termone}), 
    (\abs{\vartwo}
    {\casebang{\vartwo}{\ctxone[\vartwo *]}})(\bang \abs{}{\termone'})).
    \end{multline*}} 
with sensitivity $\baseone$, 
and the latter is finitely continuous and different from $\infty$. 
We should also exclude the constantly 
$\qunit$ change of base functor. The map $\baseone^*$ is defined as 
the right adjoint of $\baseone$ which exists since $\baseone$ preserves 
arbitrary joints (see Proposition 7.34 in \cite{DaveyPriestley/Book/1990}).
 
Another possibility is to define $\vrelone^{ctx}$ as the largest compatible 
and adequate $\quantale$-relation, where adequacy is defined via the 
$\quantale$-relation $\vrelone_o$. However, proving that such 
$\quantale$-relation exists in general seems to be far from trivial. 
These difficulties seem to suggest that contrary to what happens 
when dealing with ordinary relations, a notion of contextual 
$\quantale$-preorder/equivalence appears to be less natural than 
the notion of applicative $\vrelator$-(bi)similarity. 

\paragraph{Combining Effects} 
Our last observation concerns the applicability of the framework 
developed. In fact, all examples considered in this paper deal with calculi with 
just one kind of effects (e.g. probabilistic nondeterminism). However, 
we can apply the theory developed to combined effects as well.
We illustrate this possibility by sketching how to add global states 
to $P$-Fuzz. Recall that the global state monad $\globalstate$ is defined by 
$\globalstate X \defeq (S \times X)^S$ where $S = \{0,1\}^{\mathcal{L}}$ 
for a set of (public) location names $\mathcal{L}$. 
Such monad comes together with operation symbols for reading and writing 
locations: $\signature = \{\get, \settzero,\settone \mid \ell \in \mathcal{L}\}$. 
The intended semantics of $\get(\termone, \termtwo)$ is 
to read the content of $\ell$ and to continue as $\termone$ if the 
content is $0$, otherwise continue as $\termtwo$. Dually, 
$\settzero(\termone)$ (resp. $\settone(\termone)$) stores the 
bit $0$ (resp. $1$) in the location $\ell$ and then continues 
as $\termone$ (see Example \ref{ex:monads}).

Our combination of global stores and probabilistic computations is based on the 
monad $\globalstate_p X = (\distribution_\bot(S \times X))^S$. The unit $\unit$ 
of the monad is defined by
$\unit(x)(b) = \dirac{\lan b, x\ran}$, whereas the strong Kleisli extension 
$h^{\sharp}$ of $h: Z \times X \to (\distribution_\bot(S \times Y))^S$ 
is defined as follows: 
first we uncurry $h$ (and apply some canonical isomorphisms) to obtain the function 
$$h_u: Z \times (S \times X) \to \distribution_\bot(S \times Y).$$ 
We then define $h^{\sharp}$ by 
$$
h^{\sharp}(z,m)(b) = \strongkleisli{h_u}(z,m(b)),
$$  
where $\strongkleisli{h_u}: Z \times \distribution_\bot(S \times X) 
\to \distribution_\bot(S \times Y)$
is the strong Klesli extension
of $h_u$ with respect to $\distribution_\bot$.
Easy calculations show that 
the triple $\lan \globalstate_p, \unit, -^{\sharp} \ran$ is indeed a 
strong Kleisli triple.

We now define a $[0,1]$-relator $\vrelator$ for $\globalstate_p$. 
Given $\vrelone: X \torel Y$, define 
$$
\vrelator \vrelone(m,n) = 
\sup\nolimits_{b \in S} \wassersteinbot(\idvrel_S + \vrelone)(m(b),n(b)).
$$ 
Notice that $(\idvrel_S + \vrelone)(\lan b,x\ran, \lan b', x'\ran) = 1$ if $b \neq b'$ 
and $\vrelone(x,x')$ otherwise. It is relatively easy to prove that $\vrelator$ 
satisfies conditions in Section \ref{section:v-relators-and-v-relation-lifting}. 
As an illustrative example we prove the following result.

\begin{lemma}
The $[0,1]$-relator $\vrelator$ satisfies condition \eqref{Strong-Lax-Bind}:
$$  
\vcenter{
     \xymatrix{
     \laxcommutegeq
      Z \times X   
      \ar[r]^-{h} 
      \ar[d]_{\vrelthree + \vrelone}|-*=0@{|}   
      &  
      \globalstate_p X        
      \ar[d]^{\vrelator \vreltwo}|-*=0@{|}
      \\
      Z' \times X'   
      \ar[r]_-{h'}   
      &  
      \globalstate_p Y' 
      } }
    \implies 
    \vcenter{
    \xymatrix{
    \laxcommutegeq
    Z \times \globalstate_p X
    \ar[r]^-{h^{\sharp}} 
    \ar[d]_{\vrelthree + \vrelator \vrelone}|-*=0@{|}   
    &  
    \globalstate_p Y          
    \ar[d]^{\vrelator \vreltwo .}|-*=0@{|}
    \\
    Z' \times \globalstate_p X'
    \ar[r]_-{h'^\sharp}   
    &  
    \globalstate_p Y'
    } } 
$$
\end{lemma}

\begin{proof}
Let us call $(1)$ and $(2)$ the right-hand side and left-hand side 
of the above implication, respectively. 
Moreover, we write $\vrelone_S, \vreltwo_S$ for 
$\idvrel_S + \vrelone, \idvrel_S + \vreltwo$, 
respectively. 
Then:
\begin{align*}
(1) & \implies
\vcenter{
     \xymatrix{
     \laxcommutegeq
      Z \times (S \times X)   
      \ar[r]^-{f_u} 
      \ar[d]_{\vrelthree + \vrelone_S}|-*=0@{|}   
      &  
      \distribution_\bot(S \times Y)        
      \ar[d]^{\wassersteinbot \vreltwo_S}|-*=0@{|}
      \\
      W \times (S \times U)   
      \ar[r]_-{g_u}   
      &  
      \distribution_\bot(S \times V)
} }
\\
&\implies
\vcenter{
     \xymatrix{
     \laxcommutegeq
      Z \times \distribution_\bot(S \times X)   
      \ar[r]^-{\strongkleisli{f_u}} 
      \ar[d]_{\vrelthree + \wassersteinbot \vrelone_S}|-*=0@{|}   
      &  
      \distribution_\bot(S \times Y)        
      \ar[d]^{\wassersteinbot \vreltwo_S}|-*=0@{|}
      \\
      W \times \distribution_\bot(S \times U)   
      \ar[r]_-{\strongkleisli{g_u}}   
      &  
      \distribution_\bot(S \times V)
} }
\\
& \implies (2).
\end{align*}
\end{proof}

By Theorem \ref{prop:applicative-v-similarity-is-compatible} 
we thus obtain a notion of applicative 
$\vrelator$-similarity which is a compatible generalised metric. 
Since CBEs in $P$-Fuzz are finitely continuous we 
can also apply results from Section 
\ref{section:from-applicative-v-similarity-to-applicative-v-bisimilarity}  
to obtain a compatible pseudometric.

\section{Related Work}
\label{section:related-works}

Several works have been done in the past years on quantitative 
(metric) reasoning in the context of programming language semantics. 
In particular, several authors have used (cartesian) categories of 
\emph{ultrametric spaces} as a foundation for denotational semantics of 
both concurrent 
\cite{Arnold/Metric-interpretations/1980,DeBakker/Semantics-concurrency/1982} 
and sequential programming languages \cite{Escardo/Metric-model-PCF/1999}.
A different approach is investigated in \cite{GaboardiEtAl/POPL/2017}
where a denotational semantics combining ordinary metric spaces and domains is 
given to \emph{pure} (i.e. without effects) \Fuzz. 
The main theorem of \cite{GaboardiEtAl/POPL/2017} is 
a denotational version of the so-called \emph{metric preservation} 
\cite{Pierce/DistanceMakesTypesGrowStronger/2010} (whose original proof requires 
the introduction of a suitable \emph{step-indexed metric logical relation}). 
Our Corollary \ref{cor:metric-preservation} is the operational 
counterpart of such result generalised to arbitrary algebraic effects.

A different, although deeply related, line of research has been recently proposed 
in \cite{CrubilleDalLago/LICS/2015,CrubilleDalLago/ESOP/2017} where coinductive, 
operationally-based distances have been studied for probabilistic 
$\lambda$-calculi. In particular, 
in \cite{CrubilleDalLago/LICS/2015} a notion of applicative distance 
based on the Wasserstein lifting is proposed for a probabilistic 
\emph{affine} $\lambda$-calculus. Restricting to affine programs only 
makes the calculus strongly normalising and remove copying capabilities 
of programs by construction. In this way programs cannot amplify distances 
between their inputs and therefore are forced to behave as non-expansive 
functions. 
This limitation is overcame in \cite{CrubilleDalLago/ESOP/2017}, where a 
coinductive notion of distance is proposed for a full linear 
$\lambda$-calculus, and distance trivialisation phenomena are studied in depth. 
The price to pay for such generality 
is that the distance proposed is not applicative, but
a trace distance 
somehow resembling environmental bisimilarity \cite{Sangiorgi/Environmental/2011}.

\section{Conclusion}

In this work we have introduced an abstract framework 
for studying quantale-valued behavioural relations 
for higher-order effectful languages. Such framework has been 
instantiated to define the quantitative refinements of 
Abramsky's applicative similarity and bisimilarity for $\quantale$-\Fuzz, a
universal $\lambda$-calculus with a linear type system 
tracking program sensitivity enriched with algebraic 
effects.
Our main theorems
state that under suitable conditions the quantitative notions 
of applicative similarity and bisimilarity obtained are a compatible 
generealised metric and pseudometric, respectively.
These results can be instantiated to obtain compatible 
pseudometrics for several concrete calculi.

A future research direction is to study how the abstract 
framework developed can be used to investigate quantitative refinements 
of behavioural relations different from applicative (bi)similarity. 
In particular, investigating contextual 
distances (see \cite{Gavazzo/Arxiv/2018} for some preliminary 
observations), denotationally-based distances (along the lines of 
\cite{GaboardiEtAl/POPL/2017}), and 
distances based on suitable logical relations (such as the one in 
\cite{Pierce/DistanceMakesTypesGrowStronger/2010}) are interesting topics 
for further research.

\begin{acks}                            
 
The author would like to thank Ugo Dal Lago, Rapha\"elle Crubill\'e, and
Paul Levy for the many useful comments and suggestions. 
Special thanks also goes to Alex Simpson and Niels Voorneveld for many 
insightful discussions about the topic of this work.
\end{acks}

\bibliographystyle{plain}
\bibliography{main}

\newpage
\appendix
\section{Appendix: Technical Development}

This appendix provides proofs of propositions and lemmas 
stated in the main body of this paper.

\subsection{Proofs of Section \ref{section:v-fuzz}}
\label{appendix:proofs-v-fuzz}


\evaluationSemanticsOmegaChain*

\begin{proof}
  By induction on $n$.
  We show the case for sequential composition. We have to prove
  $
  \sem{\seq{\termone}{\termtwo}}_{n+1} \cpoleq \sem{\seq{\termone}{\termtwo}}_{n+2}
  $
  (for readability we omit subscripts).
  By definition of $\sem{-}_n$ we have:
  \begin{align*}
    \sem{\seq{\termone}{\termtwo}}_{n+1} 
    &= \kleisli{\sem{\substcomp{\termtwo}{\varone}{-}}_n}(\sem{\termone}_n), \\ 
    \sem{\seq{\termone}{\termtwo}}_{n+2} 
    &= \kleisli{\sem{\substcomp{\termtwo}{\varone}{-}}_{n+1}}(\sem{\termone}_{n+1}).
  \end{align*}
  By induction hypothesis, for any closed value $\valone$ of the appropriate type 
  we have the inequality
  $
  \sem{\substcomp{\termtwo}{\varone}{\valone}}_n \cpoleq 
  \sem{\substcomp{\termtwo}{\varone}{\valone}}_{n+1},
  $
  from which follows 
  $
  \sem{\substcomp{\termtwo}{\varone}{-}}_n \cpoleq 
  \sem{\substcomp{\termtwo}{\varone}{-}}_{n+1}.
  $
  By \ocppo\ enrichment the latter implies
  $
  \kleisli{\sem{\substcomp{\termtwo}{\varone}{-}}_n}
  \cpoleq
  \kleisli{\sem{\substcomp{\termtwo}{\varone}{-}}_{n+1}}.
  $
  Finally, by induction hypothesis we have $\sem{\termone}_n \cpoleq \sem{\termone}_{n+1}$, 
  so that we can conclude the thesis as follows\footnote{Note that by 
  \ocppo-enrichment $\kleisli{f}$ is monotone, for any $f: X \to \monad Y$. Let 
  $t,u: Z \to \monad X$ with $t \cpoleq u$, i.e. $u = \lub \{t,u\}$. Then:
  $$    
  \kleisli{f} \comp u 
  = \kleisli{f} \comp \lub\{t,u\} \\
  = \lub \{\kleisli{f} \comp t, \kleisli{f} \comp u\}
  $$
  holds, i.e. $\kleisli{f} \comp t \cpoleq \kleisli{f} \comp u$. This specialises 
  to usual pointwise monotonicity, by taking $t,u: 1 \to \monad X$.}:
  $$
  \kleisli{\sem{\substcomp{\termtwo}{\varone}{-}}_n}(\sem{\termone}_n) 
  \cpoleq 
  \kleisli{\sem{\substcomp{\termtwo}{\varone}{-}}_{n+1}}(\sem{\termone}_{n})
  \cpoleq 
  \kleisli{\sem{\substcomp{\termtwo}{\varone}{-}}_{n+1}}(\sem{\termone}_{n+1}). 
  $$
\end{proof}

\subsection{Proofs of Section \ref{section:v-relators-and-v-relation-lifting}}
\label{appendix:proofs-v-relators-and-v-relation-lifting}


\algebraVrelators*

\begin{proof}
The proof consists of a number of straightforward calculations. 
As an example, we show that $\meet_{i \in I} \vrelator_i$ 
in point $2$ satisfies condition \eqref{vrel-2}. Concretely, 
we have to prove 
$$
\meet_{i \in I} \vrelator_i \vreltwo \comp \meet_{i \in I} \vrelator_i \vrelone 
\leq \meet_{i \in I} \vrelator_i(\vreltwo \comp \vrelone). 
$$
For that it is sufficient to prove that for any $j \in I$ we have:
$$
\meet_{i \in I} \vrelator_i \vreltwo \comp \meet_{i \in I} \vrelator_i \vrelone 
\leq \vrelator_j(\vreltwo \comp \vrelone). 
$$
Observe that we have 
$\meet_{i \in I} \vrelator_i \vreltwo \leq \vrelator_j \vreltwo$ 
and
$\meet_{i \in I} \vrelator_i \vrelone \leq \vrelator_j \vrelone$, 
so that by monotonicity of composition 
(recall that $\vrel$ is a quantaloid) we infer 
$\meet_{i \in I} \vrelator_i \vreltwo \comp \meet_{i \in I} \vrelator_i \vrelone 
\leq \vrelator_j\vreltwo \comp \vrelator_j\vrelone$. The thesis now follows 
from \eqref{vrel-2}.
\end{proof}

\dualityWassersteinLifting*
\begin{proof}
The proof is a direct consequence of the following 
duality theorem for countable transportation problems 
\cite{Kortanek/InfiniteTransportationProblems/1995}.

\begin{fact}\label{prop:duality-countable-transportation-problem}
  Let $i,j, \hh$ range over natural numbers. Let $m_i, n_j, c_{ij}$ 
  be non-negative real number, for all $i,j$. Define 
  \begin{align*}
  M   &\defeq \inf \{\sum_{i,j} c_{ij} x_{ij} \mid x_{ij} \geq 0, 
      \sum_j x_{ij} = m_i, \sum_i x_{ij} = n_j,  \} \\
  M^* &\defeq \sup \{\sum_i m_i a_i + \sum_j n_j b_j \mid 
      a_i + b_j \leq c_{ij}, a_i,b_j 
      \text{ bounded}\}.
  \end{align*}
  where $a_i,b_j$ bounded means that 
  there exist $\bar a,\bar b \in \mathbb{R}$ such 
  that $a_i \leq \bar a$, and $b_j \leq \bar b$, for all $i,j$.
  Then the following hold:
  \begin{varenumerate}
    \item $M = M^*$.
    \item The linear problem $P$ induced by $M$ has optimal solution.
    \item The linear problem $P^*$ induced by $M^*$ has optimal solution.
  \end{varenumerate}
\end{fact}

Now, we first of all notice that
  $\vrelator \vrelone(\mu, \nu)$ is nothing but
  \begin{align*}
  \inf & \{\sum_{x,y} \vrelone(x,y) \cdot \omega(x,y)  \\
    & \mid \omega(x,y) \geq 0, 
    \sum_y \omega(x,y) = \mu(x), \sum_x \omega(x,y) = \nu(y)\}.
  \end{align*}
  In fact, $\omega(x,y) \geq 0$, $\sum_y \omega(x,y) = \mu(x)$ and 
  $\sum_x \omega(x,y) = \nu(y)$ imply 
  $\omega \in \distribution(\setone \times \settwo)$. Moreover, 
  since $\vrelone$ is a $\intervalQuantale$-relation, 
  $\vrelone(x,y) \in [0,1]$ (recall that Fact 
  \ref{prop:duality-countable-transportation-problem} 
  requires $c_{ij}$ to be a non-negative real number).
  We conclude the thesis by Fact
  \ref{prop:duality-countable-transportation-problem}. 
  In particular, it follows that
  there exists $\omega \in \couplings(\distone, \disttwo)$ such that:
  $$
  \wasserstein \vrelone(\mu,\nu) = \sum_{x,y} \vrelone(x,y)\cdot\omega(x,y).
  $$
  Since $\vrelone(x,y), \omega(x,y) \in [0,1]$ we have 
  $\vrelone(x,y) \cdot \omega(x,y) \leq \omega(x,y)$, for all 
  $x,y$. It follows 
  $$
  0 \leq
  \sum_{x,y} \vrelone(x,y) \cdot \omega(x,y) 
  \leq \sum_{x,y} \omega(x,y) 
  = 1
  $$
  so that $\wasserstein \vrelone$ is indeed a $[0,1]$-relation.
\end{proof}

\begin{proposition} 
Wasserstein lifting $\wasserstein$ satisfies conditions in Definition 
\ref{def:strong-v-relator}.
\end{proposition}

\begin{proof}
We start by showing that $\wasserstein$ 
satisfies condition \eqref{Lax-Unit}.
Let $\dirac x$ denotes the Dirac distribution 
on $x$. We have to show that for any $z \in \setone,w \in \settwo$, 
$\vrelone(z, w) \geq \wasserstein\vrelone(\dirac z, \dirac w)$ holds. 
By duality (Proposition \ref{prop:duality-wasserstein-lifting}) we have:
$$
\wasserstein \vrelone(\dirac z, \dirac w) = 
\max \{\sum_x a_x \cdot \dirac z(x) + \sum_y b_y \cdot \dirac w(y) 
  \mid a_x + b_y \leq \vrelone(x,y)\},
$$
where $a_x, b_y$ are bounded.
Clearly 
$\wasserstein \vrelone(\dirac z, \dirac w) = a_x + b_y$, 
for suitable $x \in X$ and $y \in Y$. 
Since $a_x + b_y \leq \vrelone(x,y)$ we are done.

We now observe that condition \eqref{s-Strong-Lax-Bind} can actually be 
split in two different conditions:
\begin{align}
 \vrelator(\baseone \acts \vrelone) &= \baseone \acts \vrelator \vrelone,
  \tag{L-dist} \label{L-dist} 
  \\
  \vrelthree \tensor (\baseone \acts \vrelone) 
  \leq \dual{g} \comp \vrelator \vreltwo \comp f 
  &\implies 
  \vrelthree \tensor (\baseone \acts  \vrelator\vrelone) 
  \leq \dual{(\strongkleisli{g})} \comp \vrelator \vreltwo \comp \strongkleisli{f},
  \tag{Strong lax bind} \label{Strong-Lax-Bind}
\end{align}
where $\baseone \leq \baseid$.
In particular, we can write condition \eqref{Strong-Lax-Bind} 
as follows:
$$
    \vcenter{
       \xymatrix{
       \ar @{} [dr] |\leq
        Z \times X 
        \ar[r]^-{f} 
        \ar[d]_{\vrelthree \tensor \vrelone}|-*=0@{|}   
        &  
        \monad Y       
        \ar[d]^{\vrelator \vreltwo}|-*=0@{|}
        \\
        Z' \times X'  
        \ar[r]_-{g}   
        &  
        \monad Y' 
        } }
     \implies 
      \vcenter{
      \xymatrix{
      \ar @{} [dr] |\leq
      Z \times \monad X
      \ar[r]^-{\strongkleisli{f}} 
      \ar[d]_{\vrelthree \tensor \vrelator\vrelone}|-*=0@{|}   
      &  
      \monad Y         
      \ar[d]^{\vrelator \vreltwo.}|-*=0@{|}
      \\
      Z' \times  \monad X'
      \ar[r]_-{\strongkleisli{g}}   
      &  
      \monad Y' 
      } } 
      $$
(notice that the latter, together with condition \eqref{Lax-Unit}, is 
equivalent to stating non-expansiveness of unit, multiplication, and strength of 
$\Monad$).

Proving that $\wasserstein$ satisfies condition \eqref{L-dist} is straightforward. 
We prove it satisfies condition \eqref{Strong-Lax-Bind}. 
Concretely, we have to prove the following implication:
    $$
    \vcenter{
       \xymatrix{
       \ar @{} [dr] |\geq
        U \times \setone   
        \ar[r]^-{f} 
        \ar[d]_{\vrelthree + \vrelone}|-*=0@{|}   
        &  
        \distribution \setthree         
        \ar[d]^{\wasserstein \vreltwo}|-*=0@{|}
        \\
        V \times \settwo   
        \ar[r]_-{g}   
        &  
        \distribution \setfour 
        } }
     \implies 
      \vcenter{
      \xymatrix{
      \ar @{} [dr] |\geq
      U \times \distribution \setone
      \ar[r]^-{\strongkleisli{f}} 
      \ar[d]_{\vrelthree + \wasserstein \vrelone}|-*=0@{|}   
      &  
      \distribution \setthree         
      \ar[d]^{\wasserstein \vreltwo.}|-*=0@{|}
      \\
      V \times \distribution \settwo
      \ar[r]_-{\strongkleisli{g}}   
      &  
      \distribution \setfour 
      } } 
      $$
      We show that for any $u \in U, v \in V, 
      \distone \in \distribution \setone, \disttwo \in \distribution \settwo$
      we have:
      $$
      \wasserstein \vreltwo(\strongkleisli{f}(u,\distone),
      \strongkleisli{g}(v,\disttwo))
      \leq 
      \vrelthree(u,v) + \wasserstein\vrelone(\distone, \disttwo).
      $$
      (note that in the right hand side of the above equations we can 
      assume without loss of generality to have ordinary addition in 
      place of a truncated sum). By very definition of 
      strong Kleisli extension we have:
      \begin{align*}
      \strongkleisli{f}(u,\distone)(z) 
      &= 
      \sum_x \distone(x) \cdot f(u,x)(z), \\
      \strongkleisli{g}(v, \disttwo)(w) 
      &= 
      \sum_y \disttwo(y) \cdot g(v,y)(w).
      \end{align*}
      Let $M \defeq \wasserstein \vreltwo(\strongkleisli{f}(u,\distone),
      \strongkleisli{g}(v,\disttwo))$. 
      By duality we have:
      \begin{align*}
      M = \max \{ 
      &\sum_z a_z \cdot \sum_x \distone(x) \cdot f(u,x)(z) \\
      &+ \sum_w b_w \cdot \sum_y \disttwo(y) \cdot g(v,y)(w) \\
      &\mid a_z + b_w \leq \vreltwo(z,w)\},
      \end{align*} 
      where $a_z$ and $b_w$ are bounded.
      By Proposition \ref{prop:duality-wasserstein-lifting} 
      there exists an $\omega \in \couplings(\mu,\nu)$ such that 
      $
      \wasserstein \vrelone(\distone,\disttwo) 
      = \sum_{x,y} \omega(x,y) \cdot \vrelone(x,y).
      $
      We have to prove:
      $$
      M \leq \vrelthree(u,v) + \sum_{x,y} \omega(x,y)\cdot \vrelone(x,y).
      $$
      From 
      $\omega \in \couplings(\distone, \disttwo)$ we obtain 
      $\distone(x) = \sum_y \omega(x,y)$,
      $\disttwo(y) = \sum_x \omega(x,y)$.
      We apply the above equalities to $M$, obtaining (for readability we omit 
      the constraint $a_z + b_w \leq \vreltwo(z,w)$):
      \begin{align*}
       M
        &= 
        \max \{\sum_z a_z \cdot \sum_x \distone(x) \cdot f(u,x)(z) \\
        &\quad + 
        \sum_w b_w \cdot \sum_y \disttwo(y) \cdot g(v,y)(w) 
        \} \\
        &= 
        \max \{\sum_z a_z \cdot \sum_{x,y} \omega(x,y) \cdot f(u,x)(z) \\ 
        &\quad + 
        \sum_w b_w \cdot \sum_{x,y} \omega(x,y) \cdot g(v,y)(w) 
        \} \\
        &= 
        \max  \{ \sum_{x,y} \omega(x,y)  ( \sum_z a_z \cdot f(u,x)(z) \\
        &\quad + 
        \sum_w b_w \cdot  g(v,y)(w) )
         \} \\
        &= 
        \sum_{x,y} \omega(x,y)  \cdot
        \max \{ \sum_z a_z \cdot f(u,x)(z) \\ 
        &\quad + 
        \sum_w b_w \cdot  g(v,y)(w) 
         \} \\
        &=
        \sum_{x,y} \omega(x,y)  \cdot
        \wasserstein \vreltwo(f(u,x), g(v,y)).
      \end{align*}
      We are now in position to use our hypothesis, namely 
      the inequality:
    $$
    \wasserstein \vreltwo(f(u,x), g(v,y)) 
    \leq 
    \vrelthree(u,v) + \vrelone(f(u,x),g(v,y))
    $$
    (note that the hypothesis we have is actually stronger, 
    since it gives an inequality for truncated addition).
    We conclude:
    \begin{align*}
    M 
    &\leq
    \sum_{x,y} \omega(x,y) \cdot (\vrelthree(u,v) + \vrelone(x,y)) \\
    &= 
    \sum_{x,y} \omega(x,y) \cdot \vrelthree(u,v) + 
    \sum_{x,y} \omega(x,y) \cdot \vrelone(x,y) \\
    &= 
    \vrelthree(u,v) + 
    \sum_{x,y} \omega(x,y) \cdot \vrelone(x,y)
    \end{align*}
    (where in the last equality we used the fact that 
    $\omega(x,y) \in \couplings(\distone, \disttwo)$ 
    implies $\sum_{x,y} \omega(x,y) = 1$). We are done.
\end{proof}

\begin{proposition} 
Wasserstein lifting $\wassersteinbot$ satisfy conditions in Definition 
\ref{def:strong-v-relator}.
\end{proposition}

\begin{proof}
Showing that $\wassersteinbot$ satisfies conditions 
\eqref{Lax-Unit} and \eqref{L-dist} is straightforward 
(but notice that for the latter we need the hypothesis 
$\baseone \leq \baseid$). We prove it 
satisfies condition \eqref{Strong-Lax-Bind} as well. 
First of 
all define for $f: U \times X \to \distribution(Y_\mbot)$ 
the map $f^\mbot: U \times X_\mbot \to \distribution(Y_\mbot)$ 
by: 
\begin{align*}
f^\mbot(u,\mbot_X) &\defeq \dirac{\mbot_Y}, \\
f^\mbot(u,x)      &\defeq f(u,x).
\end{align*}
We see that the Kleisli extension $\strongkleisli{f}$ with respect to 
the subdistribution monad $\distribution_{\leq 1}$ of 
$f: U \times X \to \distribution(Y_\mbot)$ is equal to 
${f^\mbot}^\sharp$, where $-^\sharp$ denotes 
the (strong) Kleisli extension with respect to the (full) distribution monad.
Moreover, we have the following implication:
$$
    \vcenter{
       \xymatrix{
       \ar @{} [dr] |\geq
        U \times \setone   
        \ar[r]^-{f} 
        \ar[d]_{\vrelthree + \vrelone}|-*=0@{|}   
        &  
        \distribution \setthree_\mbot        
        \ar[d]^{\wassersteinbot \vreltwo}|-*=0@{|}
        \\
        V \times \settwo   
        \ar[r]_-{g}   
        &  
        \distribution \setfour_\mbot 
        } }
     \implies 
      \vcenter{
      \xymatrix{
      \ar @{} [dr] |\geq
      U \times \setone_\mbot
      \ar[r]^-{f^\mbot} 
      \ar[d]_{\vrelthree + \vrelone_\mbot}|-*=0@{|}   
      &  
      \distribution \setthree_\mbot         
      \ar[d]^{\wassersteinbot \vreltwo.}|-*=0@{|}
      \\
      V \times  \settwo_\mbot
      \ar[r]_-{g^\mbot}   
      &  
      \distribution \setfour_\mbot 
      } } 
      $$
      Proving 
      $\wassersteinbot(f^\mbot(u, \mathpzc{x}), g^\mbot(v, \mathpzc{y})) 
      \leq \vrelthree(u,v) + \vrelone_\mbot(\mathpzc{x}, \mathpzc{y})$ 
      is trivial except if $\mathpzc{x} = \mbot$, meaning that 
      $f^\mbot(u, \mathpzc{x}) = \dirac{\mbot_Z}$. In that case we 
      observe that for any distribution $\nu \in \distribution(Y_\mbot)$ 
      and $[0,1]$-relation $\vrelone: X \torel Y$ we have
      $\wassersteinbot(\dirac{\mbot_X}, \nu) = 0$. Consider an expression of the form 
      $$
      \sum_{(\mathpzc{x},\mathpzc{y}) \in X_\mbot \times Y_\mbot} 
      \omega(\mathpzc{x},\mathpzc{y}) \cdot \vrelone_\mbot(\mathpzc{x},\mathpzc{y}),
      $$
      where $\omega \in \couplings(\dirac{\mbot_X}, \nu)$. We can expand such 
      expression as:
      \begin{align*}
      \sum_{(x,y) \in X \times Y} \omega(x,y) \cdot \vrelone_\mbot(x,y) + 
      \sum_{x \in X} \omega(x, \mbot_Y) \cdot \vrelone_\mbot(x, \mbot_Y) + \\ 
      \sum_{y \in Y} \omega(\mbot_X, y) \cdot \vrelone_\mbot(\mbot_X,y) + 
      \omega(\mbot_X,\mbot_Y) \cdot \vrelone_\mbot(\mbot_X, \mbot_Y).
      \end{align*}
      By very definition of $\vrelone_\mbot$ the latter reduces to: 
      $$
      \sum_{(x,y) \in X \times Y} \omega(x,y) \cdot \vrelone(x,y) + 
      \sum_{x \in X} \omega(x, \mbot_Y).
      $$
      Since $\omega \in \couplings(\dirac{\mbot_X}, \nu)$ we have 
      $\sum_{\mathpzc{y} \in Y_\mbot} \omega(\mathpzc{x}, \mathpzc{y}) 
      = \dirac{\mbot_X}(\mathpzc{x})$, meaning that for any $x \in X$ 
      and $\mathpzc{y} \in Y_\mbot$, $\omega(x, \mathpzc{y}) = 0$.
      We can conclude $\wassersteinbot(\dirac{\mbot_X}, \nu) = 0$.

      Finally, since $\wassersteinbot$ satisfies condition \eqref{Strong-Lax-Bind} 
      we can infer the desired thesis as follows:
      \begin{align*}
      \vrelthree + \vrelone \geq \dual{g} \comp \wassersteinbot \vreltwo \comp f 
      &\implies 
      \vrelthree + \vrelone_\mbot \geq \dual{(g^\mbot)} 
      \comp \wassersteinbot \vreltwo \comp f^\mbot  \\
      &\iff
      \vrelthree + \vrelone_\mbot \geq \dual{(g^\mbot)} 
      \comp \wasserstein \vreltwo_\mbot \comp f^\mbot \\
      &\implies 
      \vrelthree + \wasserstein \vrelone_\mbot \geq \dual{({g^\mbot}^\sharp)} 
      \comp \wasserstein \vreltwo_\mbot \comp {f^\mbot}^\sharp \\
      &\iff
      \vrelthree + \wassersteinbot \vrelone \geq \dual{(\strongkleisli{g})} 
      \comp \wassersteinbot \vreltwo \comp \strongkleisli{f}.
      \end{align*} 
\end{proof}

\subsection{Proofs of Section \ref{section:behavioural-v-relations}}
\label{appendix:proofs-behavioural-v-relations}

\begin{lemma}
For all subdistributions 
$\mu \in \mathcal{D}(X_\mbot)$ and  $\nu \in \mathcal{D}(Y_\mbot)$, 
and $[0,1]$-relation (with respect to the unit interval 
quantale) $\vrelone: X \torel Y$, we have:
$$
\sum \mu  - \sum \nu
\leq \wassersteinbot \vrelone(\mu,\nu),
$$
where $\sum \mu$ denotes the `probability of convergence' of 
$\mu$ defined by $\sum \mu \defeq \sum_{x \in X} \mu(x)$ 
(and similarity for $\nu$), and $-$ denotes truncated subtraction.  
\end{lemma}

\begin{proof}
We have:
\begin{align*}
\wassersteinbot \vrelone(\mu,\nu) = 
\max \{ &
        \sum_x a_x \cdot \mu(x) + a_{\mbot_X} \cdot \mu(\mbot_X) 
        \\ 
        &+ \sum_y b_y \cdot \nu(y) + b_{\mbot_Y} \cdot \nu(\mbot_Y)\}, 
\end{align*}
where $a_x, a_{\mbot_X}, b_y, b_{\mbot_Y}$ are bounded and satisfy 
the following constraints (already simplified according to the 
definition of $\vrelone_\mbot$):
\begin{align*}
a_x + b_y &\leq \vrelone(x,y),
&
a_{\mbot_X} + b_y &\leq 0,
\\ 
a_x + b_{\mbot_Y} &\leq 1, 
&
a_{\mbot_X} + b_{\mbot_Y} &\leq 0.
\end{align*}
Choosing $a_x \defeq 1$, $b_y \defeq -1$, $a_{\mbot_X} \defeq b_{\mbot_Y} \defeq 0$ 
we obtain the desired inequality.
\end{proof}

\applicativeVSimilarityReflexivityTransitivity*

\begin{proof}
The proof is by coinduction. Let us show that $\vsim$ is transitive, 
i.e. that $\vsim \comp \vsim \leq \vsim$. We prove that the $\lambda$-term 
$\quantale$-relation
$(\toterm{\vsim} \comp \toterm{\vsim}, \toval{\vsim} \comp \toval{\vsim})$ 
is an applicative $\vrelator$-simulation. 
We split the proof into five cases:
\begin{enumerate}[wide = 0pt, leftmargin = *]
  \item We show 
    that for all terms $\termone, \termtwo \in \Lambda_{\typeone}$ we have:
    $$
    \join_{\termthree \in \Lambda_{\typeone}} 
    \toterm{\vsim}_\typeone(\termone, \termthree) \tensor 
    \toterm{\vsim}_{\typeone}(\termthree, \termtwo) 
    \leq 
    \vrelator(\toval{\vsim}_\typeone \comp \toval{\vsim}_\typeone)
    (\sem{\termone}, \sem{\termtwo}).
    $$
    By \eqref{vrel-2} it is sufficient to prove:
    $$
    \join_{\termthree \in \Lambda_{\typeone}} 
    \toterm{\vsim}_\typeone(\termone, \termthree) \tensor 
    \toterm{\vsim}_{\typeone}(\termthree, \termtwo) \leq 
    \join_{\mathpzc{V} \in \monad \values_\typeone} 
    \vrelator\toval{\vsim}_\typeone(\sem{\termone}, \mathpzc{V}) \tensor 
    \vrelator\toval{\vsim}_\typeone(\mathpzc{V}, \sem{\termtwo}).
    $$
    For any $\termthree \in \Lambda_\typeone$ instantiate $\mathpzc{V}$
    as $\sem{\termthree}$. Since 
    $\toterm{\vsim}_\typeone(\termone, \termthree) \leq 
    \vrelator\toval{\vsim}_\typeone(\sem{\termone}, \sem{\termthree})$ and 
    $\toterm{\vsim}_{\typeone}(\termthree, \termtwo) \leq 
    \vrelator\toval{\vsim}_\typeone(\sem{\termthree}, \sem{\termtwo})$, 
    we are done by very definition of $\vsim$. 
  \item We prove that
    $$
    (\toval{\vsim}_{\typeone \mmap \typetwo} \comp 
    \toval{\vsim}_{\typeone \mmap \typetwo})(\valone, \valtwo) 
    \leq \meet_{\valthree \in \values_{\typeone}}
    (\toterm{\vsim}_{\typetwo} \comp 
    \toterm{\vsim}_{\typetwo})(\valone\valthree, \valtwo\valthree)
    $$
    holds for all values $\valone, \valtwo \in \values_{\typeone \mmap \typetwo}$. 
    For that it is sufficient to prove that for any 
    $\valthree \in \values_\typeone$ and for any 
    $\valfour \in \values_{\typeone \mmap \typetwo}$ there exists 
    a term $\termone \in \Lambda_{\typetwo}$ such that:
    $$
    \toval{\vsim}_{\typeone \mmap \typetwo}(\valone, \valfour)
    \tensor 
    \toval{\vsim}_{\typeone \mmap \typetwo}(\valfour, \valtwo) 
    \leq 
    \toterm{\vsim}_{\typetwo}(\valone\valthree, \termone) \tensor 
    \toterm{\vsim}_{\typetwo}(\termone, \valtwo\valthree).
    $$
    By very definition of $\toval{\vsim}_{\typeone \mmap \typetwo}$ 
    we have:
    \begin{align*}
    \toval{\vsim}_{\typeone \mmap \typetwo}(\valone, \valfour)
    &\tensor 
    \toval{\vsim}_{\typeone \mmap \typetwo}(\valfour, \valtwo) 
    \\ &\leq 
    \meet_{\valthree' \in \values_\typeone}
    \toterm{\vsim}_{\typetwo}(\valone \valthree', \valfour\valthree')
    \tensor 
    \meet_{\valthree' \in \values_\typeone} 
    \toterm{\vsim}_{\typetwo}(\valfour\valthree', \valtwo\valthree') \\
    &\leq 
    \toterm{\vsim}_{\typetwo}(\valone \valthree, \valfour\valthree)
    \tensor 
    \toterm{\vsim}_{\typetwo}(\valfour\valthree, \valtwo\valthree),
    \end{align*}
    so that it is sufficient to instantiate $\termone$ as $\valfour\valthree$.
  \item We prove that
    \begin{align*}
      (\toval{\vsim}_{\sumtype{i\in I}{\typeone_i}} 
      \comp \toval{\vsim}_{\sumtype{i\in I}{\typeone_i}})
      (\inject{\hat \imath}{\valone}, \inject{\hat \jmath}{\valthree}) 
      & \leq \qbot, \\
      (\toval{\vsim}_{\sumtype{i\in I}{\typeone_i}} 
      \comp \toval{\vsim}_{\sumtype{i\in I}{\typeone_i}})
      (\inject{\hat \imath}{\valone}, \inject{\hat \imath}{\valtwo}) 
      & \leq (\toval{\vsim}_{\typeone_{\hat \imath}} \comp 
      \toval{\vsim}_{\typeone_{\hat \imath}})(\valone, \valtwo),
    \end{align*}
    hold for all $\valone, \valtwo \in \values_{\typeone_{\hat \imath}}$ and  
    $\valthree \in \values_{\typeone_{\hat \jmath}}$, with 
    $\hat \imath \neq \hat \jmath$. We have 
    \begin{multline*}
    (\toval{\vsim}_{\sumtype{i\in I}{\typeone_i}} 
    \comp \toval{\vsim}_{\sumtype{i\in I}{\typeone_i}})
    (\inject{\hat \imath}{\valone}, \inject{\hat \jmath}{\valthree})
    \\ = 
    \join_{\inject{\hat \ell}{\valfour} \in \values_{\sumtype{i\in I}{\typeone_i}}} 
    \toval{\vsim}_{\sumtype{i\in I}{\typeone_i}}(\inject{\hat \imath}{\valone}, 
    \inject{\hat \ell}{\valfour}) 
    \tensor 
    \toval{\vsim}_{\sumtype{i\in I}{\typeone_i}}(\inject{\hat \ell}{\valfour}, 
    \inject{\hat \jmath}{\valtwo}).
    \end{multline*}
    Since $\hat \imath \neq \hat \jmath$ at least one among $\hat \imath \neq \hat \ell$ 
    and $\hat \ell \neq \hat \jmath$ holds, for any 
    $\inject{\hat \ell}{\valfour} \in \values_{\sumtype{i\in I}{\typeone_i}}$. 
    As a consequence, by very definition of $\vsim$, the right hand side of the 
    above inequality is equal to something of the form $\qbot \tensor \qvalone$, 
    which is itself equal to $\qbot$.
    To prove the second inequality, we have to show that for any 
    $\inject{\hat \imath}{\valthree} \in \values_{\sumtype{i \in I}{\typeone_i}}$ 
    there exists $\valfour \in \values_{\typeone_{\hat \imath}}$ such that
    \begin{multline*}
    \toval{\vsim}_{\sumtype{i\in I}{\typeone_i}}
    (\inject{\hat \imath}{\valone}, \inject{\hat \imath}{\valthree})
    \tensor
    \toval{\vsim}_{\sumtype{i\in I}{\typeone_i}}
    (\inject{\hat \imath}{\valthree}, \inject{\hat \imath}{\valtwo})
    \\ \leq 
    \toval{\vsim}_{\typeone_{\hat \imath}}(\valone, \valfour) 
    \tensor 
    \toval{\vsim}_{\typeone_{\hat \imath}}(\valfour, \valtwo).
    \end{multline*}
    Notice that for a value 
    $\inject{\hat \jmath}{\valthree} \in \values_{\sumtype{i \in I}{\typeone_i}}$
    with $\hat \jmath \neq \hat \imath$ we would have, by very definition 
    of $\vsim$,
    $\toval{\vsim}_{\sumtype{i\in I}{\typeone_i}}
    (\inject{\hat \imath}{\valone}, \inject{\hat \jmath}{\valthree}) = \qbot$, 
    and thus we would be trivially done. Proving the above inequality is 
    straightforward: simply instantiate $\valfour$ as $\valthree$ and 
    observe that by definition of $\vsim$ we have
    \begin{align*}
    \toval{\vsim}_{\sumtype{i\in I}{\typeone_i}}
    (\inject{\hat \imath}{\valone}, \inject{\hat \imath}{\valthree}) 
    &\leq 
    \toval{\vsim}_{\typeone_{\hat \imath}}(\valone, \valthree),
    \\
    \toval{\vsim}_{\sumtype{i\in I}{\typeone_i}}
    (\inject{\hat \imath}{\valthree}, \inject{\hat \imath}{\valtwo})
    &\leq 
    \toval{\vsim}_{\typeone_{\hat \imath}}(\valthree, \valtwo).
    \end{align*}
  \item The case for $\recType{\typevarone}{\typeone}$ follows the same 
    pattern of the above one.
  \item We prove:
    $$
    (\toval{\vsim}_{\bang_\baseone \typeone} \comp 
    \toval{\vsim}_{\bang_\baseone \typeone})(\bang \valone , \bang \valtwo) 
    \leq 
    \baseone \acts (\toval{\vsim}_\typeone \comp \toval{\vsim}_\typeone)
    (\valone, \valtwo).
    $$
    For that we notice that for every 
    $\bang \valthree \in \bang_{\baseone}{\typeone}$ we have: 
    \begin{align*}
    \toval{\vsim}_{\bang_\baseone \typeone}(\bang \valone, \bang \valthree)
    \tensor
    \toval{\vsim}_{\bang_\baseone \typeone}(\bang \valthree , \bang \valtwo) 
    &\leq
    (\baseone \acts \toval{\vsim}_\typeone)(\valone, \valthree) 
    \tensor
    (\baseone \acts \toval{\vsim}_\typeone)(\valthree, \valtwo) 
    \\
    &\leq ((\baseone \acts \toval{\vsim}_\typeone) 
        \comp (\baseone \acts \toval{\vsim}_\typeone))(\valone, \valtwo) \\
    &\leq \baseone \acts (\toval{\vsim}_\typeone \comp \toval{\vsim}_\typeone)
    (\valone, \valtwo).
    \end{align*}
\end{enumerate}
\end{proof} 

\kernelApplicativeSim*

\begin{proof}
The proof is by coinduction. We start proving that $\varphi \acts \vsim$ 
is an applicative $\vrelatortwo_\vrelator$-simulation. 
Since 
$\toterm{\vsim}_\typeone(\termone, \termtwo) \leq \vrelator\toval{\vsim}_\typeone
(\sem{\termone},\sem{\termtwo})$ holds for all terms 
$\termone, \termtwo \in \Lambda_\typeone$, we can apply Lemma \ref{lemma:kernel-lemma} 
and infer the inequality
$\varphi \acts \toterm{\vsim}_\typeone(\termone, \termtwo) 
\leq \vrelatortwo_\vrelator(\varphi \acts \toval{\vsim}_\typeone)
(\sem{\termone},\sem{\termtwo})$. Let us now move to the value clauses.
\begin{varenumerate}
\item We prove that for all values 
  $\valone, \valtwo \in \values_{\typeone \mmap \typetwo}$ we have:
  $$
  \varphi \acts \toval{\vsim}_{\typeone \mmap \typetwo}(\valone, \valtwo) \leq 
  \meet_{\valthree \in \values_\typeone}\varphi \acts \toterm{\vsim}_\typetwo
  (\valone\valthree, \valtwo\valthree).
  $$
  Suppose $\varphi \acts \toval{\vsim}_{\typeone \mmap \typetwo}(\valone, \valtwo) 
  = \true$, so that 
  $\toval{\vsim}_{\typeone \mmap \typetwo}(\valone, \valtwo) = \qunit$. We show that 
  $\varphi \acts \toterm{\vsim}_\typetwo(\valone\valthree, \valtwo\valthree) = \true$ 
  holds for any $\valthree \in \values_\typeone$.
  By very definition of applicative $\vrelator$-similarity, 
  $\toval{\vsim}_{\typeone \mmap \typetwo}(\valone, \valtwo) = \qunit$ implies
  $\meet_{\valthree \in \values_\typeone}\toterm{\vsim}_\typetwo
  (\valone\valthree, \valtwo\valthree) = \qunit$.
  Since $\quantale$ is integral (i.e. $\qunit = \qtop$), we must have 
  $\toterm{\vsim}_\typetwo(\valone\valthree, \valtwo\valthree) = \qunit$ 
  (and thus 
  $\varphi \acts \toterm{\vsim}_\typetwo(\valone\valthree, \valtwo\valthree)=\true$)
  for any $\valthree \in \values_\typeone$.
\item Clauses for sum and recursive types are straightforward.
\item We show that for all values 
$\bang \valone,\bang \valtwo \in \values_{\bang_\baseone \typeone}$,
 $
  \varphi \acts \toval{\vsim}_{\bang_\baseone \typeone}
  (\bang \valone, \bang\valtwo) = \true 
 $ 
 implies 
 $
  (\varphi \comp \baseone \comp \psi) \acts (\varphi \acts \toval{\vsim}_\typeone)
  (\valone, \valtwo) = \true.
 $
  By algebra of CBFs we have:
  \begin{align*}
  (\varphi \comp \baseone \comp \psi) \acts (\varphi \acts \toval{\vsim}_{\typeone}) 
  &= (\varphi \comp \baseone \comp \psi \comp \varphi) \acts \toval{\vsim}_{\typeone} 
  \\ 
  &=(\varphi \comp \baseone) \acts \toval{\vsim}_{\typeone} 
  \\
  &= \varphi \acts (\baseone \acts \toval{\vsim}_{\typeone}).
  \end{align*}
  Since 
  $\varphi \acts \toval{\vsim}_{\bang_\baseone \typeone}
  (\bang \valone, \bang\valtwo) = \true$, and thus 
  $\toval{\vsim}_{\bang_\baseone \typeone}
  (\bang \valone, \bang\valtwo) = \qunit$, by very definition of $\vsim$
  we infer $\baseone \acts \toval{\vsim}_{\typeone}(\valone, \valtwo) = \qunit$. 
  We conclude 
  $(\varphi \acts (\baseone \acts \toval{\vsim}_{\typeone}))(\valone, \valtwo) = \true$.
\end{varenumerate}

We now prove by coinduction $(\psi\ \acts \preceq) \leq \vsim$, from which follows 
$((\varphi \comp \psi)\ \acts \preceq) \subseteq (\varphi \acts \vsim)$ and 
thus $\preceq\ \subseteq (\varphi \acts \vsim)$. 
The clause for terms directly follows from Lemma \ref{lemma:kernel-lemma}. 
The clauses for values follow the same structure of the previous part of the 
proof. We show the case for values of type $\bang_\baseone \typeone$. 
Suppose 
$\psi\ \acts \toval{\preceq_{\bang_\typeone \typeone}}(\bang \valone, \bang \valtwo) 
= \qunit$ to hold (otherwise we are trivially done), meaning that
$\bang \valone \toval{\preceq_{\bang_\typeone \typeone}} \bang \valtwo$
holds as well. As a consequence, we have 
$((\varphi \comp \baseone \comp \psi)\ \acts \toval{\preceq}_\typeone)(\valone,\valtwo) 
= \true$, and thus 
$\baseone \acts (\psi\ \acts \toval{\preceq}_\typeone)
(\valone,\valtwo) = \qunit$.
\end{proof}

\subsection{Howe's Method}
\label{appendix:proofs-howe-method}

\howeOptimalValue*
\begin{proof}[Proof sketch.]
We simultaneously prove statements $1$ and $2$ by induction on 
$(\valone,\termone)$. We show a couple of cases as illustrative examples:
\begin{enumerate}[wide = 0pt, leftmargin = *]
\item Suppose
  $$
  A \defeq \{\qvalone \mid \envone \howeimpval 
  \qvalone \leq \howe{\vrelone}(\varone, \valtwo): \typeone\}
  $$ 
  to be non-empty.
  If the judgment $\envone \howeimpval \qvalone \leq 
  \howe{\vrelone}(\varone, \valtwo): \typeone$ is provable, then it 
  must be the conclusion of an instance of rule $\ruleHVar$ from the 
  premise:
  $$
  \qvalone \leq (\envtwo, \varone :_\baseone \typeone \valimp 
    \vrelone(\varone, \valtwo): \typeone),
  $$
  so that $\envone = \envtwo, \varone :_\baseone \typeone$. As a 
  consequence, we see that the set $A$ is just 
  $\{\qvalone \mid \qvalone \leq (\envtwo, \varone :_\baseone \typeone 
  \valimp \vrelone(\varone, \valtwo): \typeone)\}$.
  In particular, we have 
  $\envtwo, \varone :_\baseone \typeone 
  \valimp \vrelone(\varone, \valtwo): \typeone = \join A \in A.$
 \item Suppose 
  $$
  A \defeq \{\qvalone \mid \envone \howeimp \qvalone \leq \howe{\vrelone}
  (\seq{\termone}{\termtwo}, \termthree): \typetwo\}
  $$
  to be non-empty. That means there exists $\qvalone \in \quantale$ such 
  that $\envone \howeimp \qvalone \leq \howe{\vrelone}(\seq{\termone}
  {\termtwo}, \termthree): \typetwo$ is derivable. The latter judgment 
  must be the conclusion of an instance of rule $\ruleHSeq$ from 
  premisses:
  \begin{align*}
  &\Sigma \howeimp \qvaltwo \leq \howe{\vrelone}(\termone, \termone'): 
  \typeone, 
  \\
  &\envtwo, \varone :_\baseone \typeone \howeimp \qvalthree \leq 
  \howe{\vrelone}(\termtwo, \termtwo'): \typetwo,
  \\
  &\qvalfour \leq (\baseone \wedge \baseid) \comp 
  \Sigma \tensor \envtwo \imp \vrelone
  (\seq{\termone'}{\termtwo'}, \termthree): \typetwo,
  \end{align*}
  so that $\envone = (\baseone \wedge \baseid) \comp \Sigma \tensor 
  \envtwo$ and $\qvalone = (\baseone \wedge \baseid)(\qvaltwo) 
  \tensor \qvalthree \tensor \qvalfour$. In particular, the sets 
  \begin{align*}
  B &= \{\qvaltwo \mid \Sigma \howeimp \qvaltwo \leq \howe{\vrelone}
  (\termone, \termone'): \typeone\}, 
  \\
  C &= \{\qvalthree \mid \envtwo, \varone :_\baseone \typeone 
  \howeimp \qvalthree \leq \howe{\vrelone}(\termtwo, \termtwo'): 
  \typetwo\},
  \end{align*}
  are non-empty. By induction hypothesis we have $\join B \in B$ and 
  $\join C \in C$. Let $\underline{\qvalfour} = (\baseone \wedge \baseid) \comp \Sigma 
  \tensor \envtwo \imp \vrelone(\seq{\termone'}{\termtwo'}, 
  \termthree): \typetwo$. We can now apply rule $\ruleHSeq$ obtaining 
  $(\baseone \wedge \baseid) \big(\join B\big) \tensor 
  \big(\join C\big) \tensor \underline{\qvalfour} \in A$. To see that the latter is 
  actually $\join A$ it is sufficient to show that for any 
  $\qvalone \in A$ we have $\qvalone \leq (\baseone \wedge \baseid) 
  \big(\join B\big) \tensor \big(\join C\big) \tensor \underline{\qvalfour}$. 
  But any $\qvalone \in A$ (with $\qvalone \neq \qbot$) 
  is of the form $(\baseone \wedge \baseid)(\qvaltwo) 
  \tensor \qvalthree \tensor \qvalfour$ for $\qvaltwo \in B$, 
  $\qvalthree \in C$, and $\qvalfour \leq \underline{\qvalfour}$. 
  We are done since both 
  $(\baseone \wedge \baseid)$ and $\tensor$ are monotone.
\end{enumerate}
\end{proof}

It is now easy to show that the above definition of Howe's extension 
coincide with the one of Definition \ref{def:howe-extension}. In particular, 
for an open $\lambda$-term $\quantale$-relation $\vrelone$, 
$\howe{\vrelone}$ is the least compatible open $\lambda$-term $\quantale$-relation
satisfying the inequality $\vrelone \comp \vreltwo \leq \vreltwo$.

The following are standard results on Howe's extension. Proofs are 
straightforward but tedious (they closely resemble their relational 
counterparts), and thus are omitted.

\begin{lemma}\label{lemma:properties-howe-extension}
Let $\vrelone$ be a reflexive and transitive open 
$\lambda$-term $\quantale$-relation. 
Then the following hold:
\begin{varenumerate}
  \item $\howe{\vrelone}$ is reflexive.
  \item $\vrelone \leq \howe{\vrelone}$.
  \item $\vrelone \comp \howe{\vrelone} \leq \howe{\vrelone}$.
  \item $\howe{\vrelone}$ is compatible.
\end{varenumerate}
\end{lemma}
We refer to property 
$1$ as pseudo-transitivity. In particular, by very 
definition of $\quantale$-relator we also have 
$\vrelator \vrelone \comp \vrelator \howe{\vrelone} \leq \vrelator \howe{\vrelone}$. 
We refer to the latter property as $\vrelator$-pseudo-transitivity.

\substitutivityLemma*

\begin{proof}
We simultaneously prove the following statements.
\begin{enumerate}[wide = 0pt, leftmargin = *, label=(\roman*)]
\item For any $\qvalone \in \quantale$ if 
    $\envone, \varone :_\baseone \typeone 
    \howeimp \qvalone \leq \howe{\vrelone}(\termone, \termtwo): \typetwo$ 
    is derivable, then $\qvalone \tensor \baseone(\underline{\qvalone}) 
    \leq \envone \imp \howe{\vrelone}(\substcomp{\termone}{\varone}{\valone}, 
    \substcomp{\termtwo}{\varone}{\valtwo}): \typetwo$ holds.
\item For any $\qvalone \in \quantale$ if 
    $\envone, \varone :_\baseone \typeone 
    \howeimpval \qvalone \leq \howe{\vrelone}(\valthree, \valfour): \typetwo$ 
    is derivable, then $\qvalone \tensor \baseone(\underline{\qvalone}) 
    \leq \envone \imp \howe{\vrelone}(\substval{\valthree}{\valone}{\varone}, 
    \substval{\valfour}{\valtwo}{\varone}): \typetwo$ holds.
\end{enumerate} 
The proof is by induction on the derivation of 
the judgments: 
\begin{align*}
\mathcal{J} &\defeq \envone, \varone :_\baseone \typeone 
\howeimp \qvalone \leq \howe{\vrelone}(\termone, \termtwo): \typetwo,
\\
\mathcal{J'} &\defeq \envone, \varone :_\baseone \typeone 
\howeimpval \qvalone \leq \howe{\vrelone}(\valthree, \valfour): \typetwo.
\end{align*}
\begin{enumerate}[wide = 0pt, leftmargin = *]
  \item Suppose $\mathcal{J'}$ has been inferred via 
    an instance of rule $\ruleHVar$. 
    We have two subcases to consider.
    \begin{itemize}[wide = 0pt, leftmargin = *]
      \item[1.1] $\mathcal{J'}$ has been inferred via 
    an instance of rule $\ruleHVar$ from premisses:
        \[
        \infer[\ruleHVar,]
        {
        \envone, \varone :_\baseone \typeone \howeimpval 
        \qvalone \leq \howe{\vrelone}(\varone, \valthree): \typeone
        }
        {
        \qvalone \leq \envone, \varone :_\baseone \typeone 
        \valimp \vrelone(\varone, \valthree): \typeone
        }
        \]
        so that $\baseone \leq \baseid$ and $\mathcal{J'}$ is 
        $\envone, \varone :_\baseone \typeone \howeimpval 
        \qvalone \leq \howe{\vrelone}(\varone, \valthree): \typeone$. 
        We have to prove 
        $\qvalone \tensor \baseone 
        \acts (\emptyset \valimp \howe{\vrelone}(\valone, \valtwo)) \leq 
        \envone \valimp \howe{\vrelone}
        (\valone,\substval{\valthree}{\valtwo}{\varone}) 
        : \typeone$. Since $\vrelone$ is value substitutive, from 
        $\envone, \varone :_\baseone \typeone \howeimpval 
        \qvalone \leq \howe{\vrelone}(\varone, \valthree): \typeone$ 
        we infer $\qvalone \leq \envone \valimp 
        \vrelone(\valtwo, \substval{\valthree}{\valtwo}{\varone}): \typeone$. 
        Moreover, since $\howe{\vrelone}$ is an open $\lambda$-term $\quantale$-relation 
        (and thus closed under weakening), we have 
        $\emptyctx \valimp \howe{\vrelone}(\valone, \valtwo): \typeone 
        \leq 
        \envone \valimp \howe{\vrelone}(\valone, \valtwo): \typeone$. We can now 
        conclude the thesis as follows:
        \begin{align*}
        \qvalone \tensor \baseone(\underline{\qvalone}) 
        & \leq (\envone \valimp \howe{\vrelone}(\valone, \valtwo): \typeone) 
          \tensor \baseone \acts (\envone \valimp \vrelone
          (\valtwo, \substval{\valthree}{\valtwo}{\varone}): \typeone) 
      \\ 
        & \leq (\envone \valimp \howe{\vrelone}(\valone, \valtwo): \typeone) 
          \tensor (\envone \valimp \vrelone
          (\valtwo, \substval{\valthree}{\valtwo}{\varone}): \typeone) 
      \\ 
        & \qquad \text{[ since $\baseone \leq \baseid$ ]}
      \\ 
        & \leq \envone \valimp \howe{\vrelone}
          (\valone, \substval{\valthree}{\valtwo}{\varone}): \typeone
      \\ 
        & \qquad \text{[ by pseudo-transitivity ].}
      \end{align*}
    \item[1.2] $\mathcal{J'}$ has been inferred via 
    an instance of rule $\ruleHVar$ from premisses:
        \[
        \infer[\ruleHVar]
        {
        \envone, \vartwo :_\basetwo \typetwo, \varone :_\baseone \typeone 
        \howeimpval 
        \qvalone \leq \howe{\vrelone}(\vartwo, \valthree): \typetwo
        }
        {
        \qvalone \leq \envone, \vartwo :_\basetwo \typetwo, 
        \varone :_\baseone \typeone 
        \valimp \vrelone(\vartwo, \valthree): \typetwo
        }
        \]
        so that $\mathcal{J'}$ is 
        $\envone, \vartwo :_\basetwo \typetwo, \varone :_\baseone \typeone 
        \howeimpval 
        \qvalone \leq \howe{\vrelone}(\vartwo, \valthree): \typetwo.
        $
        We have to prove 
        $\qvalone \tensor \baseone 
        \acts (\emptyset \valimp \howe{\vrelone}(\valone, \valtwo)) \leq 
        \envone, \vartwo :_\basetwo \typetwo \valimp \howe{\vrelone}
        (\vartwo,\substval{\valthree}{\valtwo}{\varone}) 
        : \typetwo$. 
        As $\quantale$ is integral and $\vrelone$ is 
        value-substitutive, we have: 
        $$
        \qvalone \tensor \baseone 
        \acts (\emptyset \valimp \howe{\vrelone}(\valone, \valtwo)) 
        \leq \qvalone 
        \leq 
        \envone, \vartwo :_\basetwo \typetwo \valimp \vrelone
        (\vartwo,\substval{\valthree}{\valtwo}{\varone}).
        $$ 
        Since $\vrelone \leq \howe{\vrelone}$ we are done.
    \end{itemize}
  \item Suppose $\mathcal{J}$ has been inferred via 
    an instance of rule $\ruleHSeq$ from premisses:
    \begin{align}
      &\envone, \varone :_\baseone \typeone \howeimp \qvalone 
      \leq \howe{\vrelone}(\termone, \termone'): \typeone',
      \label{subst-lemma:seq-premise-1}
      \\
      &\envtwo, \varone :_\basetwo \typeone, \vartwo :_\basethree \typeone' 
      \howeimp \qvaltwo
      \leq \howe{\vrelone}(\termtwo, \termtwo'): \typetwo,
      \label{subst-lemma:seq-premise-2}
      \\
      & \qvalthree \leq (\basethree \wedge \baseid) \comp 
      (\envone, \varone :_\baseone \typeone) 
      \tensor (\envtwo, \varone :_\basetwo \typeone) 
      \imp 
      \howe{\vrelone}(\seqy{\termone'}{\termtwo'}, \termthree): \typetwo.
      \label{subst-lemma:seq-premise-3}
    \end{align}
    so that $J$ is:
    \begin{multline*}
    (\basethree \wedge \baseid) \comp \envone \tensor \envtwo, 
    \varone :_{(\basethree \wedge \baseid) \comp \baseone \tensor \basetwo} \typeone 
    \howeimp (\basethree \wedge \baseid)(\qvalone) 
    \tensor \qvaltwo \tensor \qvalthree 
    \\ \leq 
      \howe{\vrelone}(\seqy{\termone}{\termtwo}, \termthree): \typetwo.
    \end{multline*}
    We have to prove: 
    \begin{multline*}
    (\basethree \wedge \baseid)(\baseone(\underline{\qvalone})) 
    \tensor \basetwo(\underline{\qvalone}) 
    \tensor (\basethree \wedge \baseid)(\qvalone) 
    \tensor \qvaltwo \tensor \qvalthree 
     \leq 
    (\basethree \wedge \baseid) \comp \envone \tensor \envtwo 
    \\ \imp 
    \howe{\vrelone}(\seqy
    {\substcomp{\termone}{\varone}{\valone}}
    {\substcomp{\termtwo}{\varone}{\valone}}, \substcomp{\termthree}{\varone}{\valtwo}): 
    \typetwo.
    \end{multline*}
    We apply the induction hypothesis on \eqref{subst-lemma:seq-premise-1} and 
    \eqref{subst-lemma:seq-premise-2} obtaining:
    \begin{align}
    &\baseone(\underline{\qvalone}) \tensor \qvalone 
    \leq \envone \imp \howe{\vrelone}
    (\substcomp{\termone}{\varone}{\valone}, \substcomp{\termone'}{\varone}{\valtwo}):
    \typeone',
    \label{subst-lemma:seq-fact-1}
    \\
    &\basetwo(\underline{\qvalone}) \tensor \qvaltwo \leq 
    \envtwo, \vartwo :_\basethree \typeone' \imp \howe{\vrelone}
    (\substcomp{\termtwo}{\varone}{\valone}, \substcomp{\termtwo}{\varone}{\valtwo}): 
    \typetwo.
    \label{subst-lemma:seq-fact-2}
    \end{align}
    From \eqref{subst-lemma:seq-fact-1} and \eqref{subst-lemma:seq-fact-2} by 
    compatibility of $\howe{\vrelone}$ (and lax equations of change of base functors) 
    we infer:
    \begin{multline}
    (\basethree \wedge \baseid)(\baseone(\underline{\qvalone})) 
    \tensor (\basethree \wedge \baseid)(\qvalone) \tensor 
    \basetwo(\underline{\qvalone}) \tensor \qvaltwo 
    \leq
    (\basethree \wedge \baseid) \comp \envone \tensor \envtwo 
    \\\imp 
    \howe{\vrelone}(\seqy{\substcomp{\termone}{\varone}{\valone}}
    {\substcomp{\termtwo}{\varone}{\qvalone}}, 
    \\ \seqy{\substcomp{\termone'}{\varone}{\valtwo}}
    {\substcomp{\termtwo'}{\varone}{\valtwo}}): \typetwo.
    \label{subst-lemma:seq-fact-3}
    \end{multline}
    Finally, since $\vrelone$ is value-substitutive, from 
    \eqref{subst-lemma:seq-premise-3} we obtain:
    $$
    \qvalthree \leq (\basethree \wedge \baseid) \comp 
      \envone \tensor \envtwo
      \imp 
      \howe{\vrelone}(\seqy{\substcomp{\termone'}{\varone}{\valtwo}}
      {\substcomp{\termtwo'}{\varone}{\valtwo}}, \termthree): \typetwo,
    $$
    and thus conclude the thesis from the latter and 
    \eqref{subst-lemma:seq-fact-3} by pseudo-transitivity.
\item Suppose $\mathcal{J}$ has been inferred via 
    an instance of rule $\ruleHOp$ from premisses 
    (as usual we write $\vec{x_i}$ for items $x_1, \hh, x_n$):
    \begin{align}
      &\forall i.\ \envone_i, \varone :_{\baseone_i} \typeone \howeimp \qvalone_i
      \leq \howe{\vrelone}(\termone_i, \termone_i'): \typetwo,
      \label{subst-lemma:op-premise-1}
      \\
      &\qvaltwo \leq 
      \qop(\vec{\envone_i}), 
      \varone :_{\qop(\vec{\baseone_i})} \typeone 
      \imp \vrelone(\op(\vec{\termone_i'}), \termtwo): \typetwo,
      \label{subst-lemma:op-premise-2}
    \end{align}
    so that $J$ is 
    $$
    \qop(\vec{\envone_i}), 
    \varone :_{\qop(\vec{\baseone_i})} \typeone 
    \howeimp \qop(\vec{\qvalone_i}) \tensor \qvaltwo 
    \leq \howe{\vrelone}(\op(\vec{\termone_i}), \termtwo): \typetwo.
    $$
    We have to prove 
    \begin{multline*}
    \qop(\overrightarrow{\baseone_i(\underline{\qvalone})}) \tensor 
    \qop(\vec{\qvalone_i}) \tensor \qvaltwo
    \\ \leq 
    \qop(\vec{\envone_i}) \imp \howe{\vrelone}
    (\overrightarrow{{\substcomp{\termone_i}{\varone}{\valone}}},
    \substcomp{\termtwo}{\varone}{\valtwo}): 
    \typetwo.
    \end{multline*}
    We apply the induction hypothesis on \eqref{subst-lemma:op-premise-1}  
    obtaining:
    \begin{align}
    &\forall i.\ 
    \baseone(\underline{\qvalone}) \tensor \qvalone_i 
    \leq \envone_i \imp \howe{\vrelone}
    (\substcomp{\termone_i}{\varone}{\valone}, \substcomp{\termone_i'}{\varone}{\valtwo}):
    \typetwo.
    \label{subst-lemma:op-fact-1}
    \end{align}
    Monotonicity of $\qop$ on \eqref{subst-lemma:op-fact-1}
    followed by compatibility gives:
    \begin{multline}
    \qop(\overrightarrow{\baseone_i(\underline{\qvalone}) \tensor \qvalone_i}) 
    \\ \leq \qop(\vec{\envone_i}) \imp \howe{\vrelone}
     (\op(\overrightarrow{\substcomp{\termone_i}{\varone}{\valone}}), 
     \op(\overrightarrow{\substcomp{\termone_i'}{\varone}{\valtwo}})).
     \label{subst-lemma:op-fact-2}
     \end{multline}
    Finally, as $\vrelone$ is value-substitutive, from 
    \eqref{subst-lemma:op-premise-2} we obtain:
    $$
    \qvaltwo \leq 
      \qop(\vec{\envone_i}), 
      \imp \vrelone(\op(\overrightarrow{\substcomp{\termone_i'}{\varone}{\valtwo}}) 
      \substcomp{\termtwo}{\varone}{\valtwo}): \typetwo.
    $$
    The latter together with 
    \eqref{subst-lemma:seq-fact-2} implies 
    $$
    \qop(\overrightarrow{\baseone_i(\underline{\qvalone}) \tensor \qvalone_i}) 
    \tensor \qvaltwo \leq \qop(\vec{\envone_i}) \imp \howe{\vrelone}
     (\op(\overrightarrow{\substcomp{\termone_i}{\varone}{\valone}}), 
     \substcomp{\termtwo}{\varone}{\valtwo})
    $$
    by pseudo-transitivity. We conclude the thesis as Definition 
    \ref{def:signature-quantale} entails:
    $$
    \qop(\overrightarrow{\baseone_i(\underline{\qvalone})}) \tensor \qop(\vec{\qvalone_i})
    \leq
    \qop(\overrightarrow{\baseone_i(\underline{\qvalone}) \tensor \qvalone_i}).
    $$
\end{enumerate}
The remaining cases follow the same pattern.
\end{proof}

\keyLemma*
\renewcommand{\kleisli}[1]{#1^\dagger}
\begin{proof}
Let us write $\howe{\vrelone}$ for the Howe's extension 
of $\vrelone$ restricted to closed terms/values. 
It is easy to see that $\howe{\vrelone}$ 
satisfies the simulation clauses for values.
For instance, we prove the inequation 
$\howe{\vrelone}_{\bang_\baseone \typeone}(\bang \valone, \bang \valtwo) 
\leq \baseone \acts \howe{\vrelone}_{\typeone}(\valone, \valtwo)$, 
where for readability we omit values superscript in $\vrelone$ and 
$\howe{\vrelone}$. It is sufficient to show that for any 
$\qvalone \in \quantale$ such that 
$\mathcal{J} \defeq \emptyctx \howeimp \qvalone \leq 
\howe{\vrelone}(\bang \valone, \bang \valtwo): \bang_\baseone \typeone$ is derivable, 
the inequation $\qvalone \leq \baseone \acts \howe{\vrelone}_\typeone(\valone, \valtwo)$ 
holds. The judgment $\mathcal{J}$ must have been inferred via an instance of 
rule $\ruleHBang$, so that without loss of generality we can assume 
$\qvalone = \baseone(\qvaltwo) \tensor 
\vrelone_{\bang_\baseone \typeone}(\bang \valthree, \bang \valtwo)$, 
with $\emptyctx \howeimp \qvaltwo \leq \howe{\vrelone}(\valone, \valthree): \typeone$
derivable, for some value $\valthree$. We conclude the thesis as follows:
\begin{align*}
\qvalone  
&\leq \baseone \acts \howe{\vrelone}_\typeone(\valone, \valthree) 
  \tensor \vrelone_{\bang_\baseone \typeone}(\bang \valthree, \bang \valtwo) 
\\
&\leq \baseone \acts \howe{\vrelone}_\typeone(\valone, \valthree) 
  \tensor \baseone \acts \vrelone_\typeone(\valthree, \valtwo) 
\\ & \qquad [\vrelone \text{ is an applicative }\vrelator\text{-simulation}]
\\
&\leq \baseone \acts (\howe{\vrelone}_\typeone(\valone, \valthree) 
  \tensor \vrelone_\typeone(\valthree, \valtwo)) 
\\
& \leq \baseone \acts (\vrelone_\typeone \comp \howe{\vrelone}_\typeone)(\valone, \valtwo) 
\\
&\leq \baseone \acts \howe{\vrelone}_\typeone(\valone, \valtwo)
\\ & \qquad [\text{pseudo-transitivity}]
\end{align*}

The crucial part of the proof
is to show that $\howe{\vrelone}$ satisfies the clause for terms.
We prove that for any $n \geq 0$,
$$
\toterm{(\howe{\vrelone})}_{\typeone}(\termone, \termtwo) \leq 
\relator \toval{(\howe{\vrelone})}_{\typeone}(\approxsem{\termone}{n}, \sem{\termtwo})
$$
holds for all terms $\termone, \termtwo \in \Lambda_\typeone$. 
Since $\vrelator$ is inductive the above inequality gives the thesis 
as follows:
\begin{align*}
\toterm{(\howe{\vrelone_\typeone})}(\termone, \termtwo) 
&\leq 
\meet_n \vrelator\toval{(\howe{\vrelone_\typeone})}(\sem{\termone}_n, \sem{\termtwo})  
\\
& \leq 
\vrelator\toval{(\howe{\vrelone_\typeone})}(\lub_n\sem{\termone}_n, \sem{\termtwo})
\\
&= 
\vrelator\toval{(\howe{\vrelone_\typeone})}(\sem{\termone}, \sem{\termtwo}).
\end{align*}
The proof is by induction on $n$ with a case analysis on 
the term structure in the inductive case. For readability 
we simply write $\vrelone$ in place of $\toterm{\vrelone}$ and 
$\toval{\vrelone}$. Moreover, to avoid confusion it is useful to 
explicitly distinguishing between (ordinary) Kleisli extension and 
strong Kleisli extension. 
Given a monoidal category $\lan \catone, I, \tensor \ran$, we denote by 
$\strongkleisli{f}: Z \tensor \monad X \to \monad Y$ the \emph{strong} 
Kleisli extension of $f: Z \tensor X \to \monad Y$ and by 
$g^\dagger: \monad X \to \monad Y$ the Kleisli extension of 
$g: X \to \monad Y$. The latter can be defined in terms of the former as 
$g^\dagger \defeq \strongkleisli{(g \comp \lambda_X)} \comp \lambda_{\monad X}^{-1}$, 
where $\lambda_X: I \tensor X \xrightarrow{\cong} X$ is the natural isomorphism given by 
the monoidal structure of $\catone$. Note that, in particular, 
$g^\dagger \comp \lambda_{\monad X} = \strongkleisli{(g \comp \lambda_X)}$.

\begin{enumerate}[wide = 0pt, leftmargin = *]
  \item
    We have to prove:
    $$
    \howe{\vrelone}_{\typeone}(\termone, \termtwo) \leq 
    \vrelator \howe{\vrelone}_{\typeone}
    (\approxsem{\termone}{0}, \sem{\termtwo}).
    $$ 
    Since $\vrelator$ is inductive and $\approxsem{\termone}{0} = 
    \mbot_{\values_\typeone}$, it is sufficient to prove
    $\howe{\vrelone}_{\typeone}(\termone, \termtwo) \leq \qunit$. 
    Because the quantale is integral the latter trivially holds. 
  \item
    We have to prove:
    $$
    \howe{\vrelone}_{\typeone}
    (\return{\valone}, \valtwo) 
    \leq 
    \vrelator \howe{\vrelone}_{\typeone}
    (\approxsem{\return{\valone}}{n+1}, \sem{\valtwo}).
    $$
    Since
    $\approxsem{\return{\valone}}{n+1} = \unit(\valone)
    $, 
    it is sufficient to prove that for any $\qvalone$ such that 
    the judgment
    $
    \emptyctx \howeimp \qvalone \leq 
    \howe{\vrelone}(\return{\valone}, \valtwo):
    \typeone
    $
    is derivable, 
    $
    \qvalone \leq 
    \vrelator \howe{\vrelone}_{\typeone}(\unit(\valone), \sem{\valtwo})
    $
    holds. Suppose 
    $
    \emptyctx \howeimp \qvalone \leq 
    \howe{\vrelone}(\return{\valone}, \valtwo):
    \typeone
    $ to be derivable. 
    The latter must have been inferred via an instance of rule 
    $\ruleHReturn$ from premisses:
    \begin{align}
    &\emptyctx \howeimp \qvaltwo \leq 
    \howe{\vrelone}(\valone, \valone'): \typeone,
    \label{key-lemma:return-premise-1}\\
    &\qvalthree \leq \vrelone_{\typeone}
    (\return{\valone'}, \valtwo).
    \label{key-lemma:return-premise-1}
    \end{align}
    In particular, we have $\qvaltwo \leq \howe{\vrelone}_\typeone(\valone, \valone')$ 
    and thus, by condition \eqref{Lax-Unit}, 
    $\qvaltwo \leq \vrelator\howe{\vrelone}_\typeone(\unit(\valone), \unit(\valone'))$.
    From, 
    \eqref{key-lemma:return-premise-1} we infer, by very definition of 
    applicative $\vrelator$-simulation, 
    $\qvalthree \leq \vrelator \vrelone_\typeone(\unit(\valone'), \sem{\valtwo})$, 
    and thus
    $
    \qvaltwo \tensor \qvalthree \leq 
    \vrelator\howe{\vrelone}_\typeone(\unit(\valone), \unit(\valone')) 
    \tensor 
    \vrelator \vrelone_\typeone(\unit(\valone'), \sem{\valtwo}).
    $
    We conclude the thesis by $\vrelator$-pseudo-transitivity.
  \item
    We have to prove:
    $$
    \howe{\vrelone}_{\typetwo}
    ((\abs{\varone}{\termone})\valone, \termtwo) 
    \leq 
    \vrelator \howe{\vrelone}_{\typetwo}
    (\approxsem{(\abs{\varone}{\termone})\valone}{n+1}, \sem{\termtwo}).
    $$
    As
    $\approxsem{(\abs{\varone}{\termone})\valone}{n+1} = 
    \approxsem{\substcomp{\termone}{\varone}{\valone}}{n}
    $, 
    it is sufficient to show that for any $\qvalone$ such that 
    $
    \emptyctx \howeimp \qvalone \leq 
    \howe{\vrelone}((\abs{\varone}{\termone})\valone, \termtwo):
    \typetwo
    $
    holds, we have 
    $
    \qvalone \leq 
    \vrelator \howe{\vrelone}_{\typetwo}
    (\approxsem{\substcomp{\termone}{\varone}{\valone}}{n}, \sem{\termtwo})
    $
    . Assume 
    $
    \emptyctx \howeimp \qvalone \leq 
    \howe{\vrelone}((\abs{\varone}{\termone})\valone, \termtwo):
    \typetwo
    $. 
    The latter must have been inferred via an instance of rule 
    $\ruleHApp$ from premisses:
    \begin{align}
    &\emptyctx \howeimp \qvaltwo \leq 
    \howe{\vrelone}(\valone, \valtwo): \typeone,
    \label{key-lemma:app-premise-1}\\
    &\emptyctx \howeimp \qvalthree 
    \leq \howe{\vrelone}(\abs{\varone}{\termone}, 
    \valthree): \typeone \mmap \typetwo,
    \label{key-lemma:app-premise-2}\\
    &\qvalfour \leq \vrelone_{\typetwo}
    (\valthree\valtwo, \termtwo).
    \label{key-lemma:app-premise-3}
    \end{align}
    Let us examine premise \eqref{key-lemma:app-premise-2}. First 
    of all, since $\valthree$ is a closed value of type 
    $\typeone \mmap \typetwo$ it must be 
    of the form $\abs{\varone}{\termthree}$. Moreover, 
    \eqref{key-lemma:app-premise-2} must have been inferred via an
    instance rule rule $\ruleHAbs$ from premisses:
    \begin{align}
    &\varone :_{\baseid} \typeone \howeimp \qvalthree_1 \leq 
    \howe{\vrelone}(\termone, \termfour): \typetwo,
    \label{key-lemma:app-premise-2.1}\\
    &\qvalthree_2 \leq 
    \vrelone_{\typeone \mmap \typetwo}
    (\abs{\varone}{\termfour}, \abs{\varone}{\termthree}).
    \label{key-lemma:app-premise-2.2}
    \end{align}
    In particular, we have the equality
    $\qvalthree_1 \tensor \qvalthree_2 = \qvalthree$.
    From \eqref{key-lemma:app-premise-2.1} we deduce 
    $
    \qvalthree_1 \leq 
    \varone :_{\baseid} \typeone \imp 
    \howe{\vrelone}(\termone, \termfour): \typetwo
    $, whereas from 
    \eqref{key-lemma:app-premise-1} we infer
    $\qvaltwo \leq \howe{\vrelone}_{\typeone}(\valone, \valtwo)$. 
    We are now in position to apply the Substitution Lemma,
    obtaining
    $
    \qvalthree_1 \tensor \qvaltwo \leq 
    \howe{\vrelone}_{\typetwo}
    (\substcomp{\termone}{\varone}{\valone}, \substcomp{\termfour}{\varone}{\valtwo}).
    $
    By very definition of 
    applicative $\vrelator$-simulation, \eqref{key-lemma:app-premise-2.2}
    implies the inequality
    $
    \qvalthree_2 \leq \vrelone_{\typetwo}
    (\substcomp{\termfour}{\varone}{\valtwo}, 
    \substcomp{\termthree}{\varone}{\valtwo}).
    $
    Applying pseudo-transitivity followed 
    by the induction hypothesis we obtain:
    \begin{multline*}
    \qvalthree_1 \tensor \qvalthree_2 \tensor \qvaltwo \leq 
    \howe{\vrelone}_{\typetwo}(\substcomp{\termone}{\varone}{\valone}, 
    \substcomp{\termthree}{\varone}{\valtwo})
    \\ \leq 
    \vrelator \howe{\vrelone}_{\typetwo}
    (\approxsem{\substcomp{\termone}{\varone}{\valone}}{n}, 
    \sem{\substcomp{\termthree}{\varone}{\valtwo}}).
    \end{multline*}
    Finally, from \eqref{key-lemma:app-premise-3}, by definition 
    of applicative $\vrelator$-simulation we infer
    $
    \qvalfour 
    \leq 
    \vrelator \vrelone_{\typetwo}
    (\sem{\substcomp{\termthree}{\varone}{\valtwo}}, \sem{\termtwo})
    $
    (recall that $\valthree = \abs{\varone}{\termthree}$, so that 
    $\sem{\valthree \valtwo} = \sem{\substcomp{\termthree}{\varone}{\valtwo}}$).
    We can now conclude the thesis by 
    $\vrelator$-pseudo-transitivity.
  \item
    Cases for pattern matching against folds and sums are standard (they follow the 
    same pattern of point 5 but are simpler).
  \item 
    We have to prove: 
    $$
    \howe{\vrelone}_{\typetwo}(\casebang{\bang \valone}{\termone}, \termtwo) 
    \leq \vrelator \howe{\vrelone}_{\typetwo}
    (\approxsem{\casebang{\bang \valone}{\termone}}{n+1}, \sem{\termtwo}).
    $$
    As 
    $\approxsem{\casebang{\bang \valone}{\termone}}{n+1} = 
    \approxsem{\substcomp{\termone}{\varone}{\valone}}{n}$,
    we show that for any $\qvalone$ such that 
    $\emptyctx \howeimp \qvalone \leq \howe{\vrelone}
    (\casebang{\bang \valone}{\termone}, \termtwo): \typetwo$ 
    is derivable, the inequality
    $
    \qvalone \leq \vrelator \howe{\vrelone}_{\typetwo}
    (\approxsem{\substcomp{\termone}{\varone}{\valone}}{n}, 
    \sem{\termtwo})
    $ holds. 
    Suppose $\emptyctx \howeimp \qvalone \leq \howe{\vrelone}
    (\casebang{\bang \valone}{\termone}, \termtwo): \typetwo$. The 
    latter must have been inferred via an instance of rule 
    $\ruleHPmBang$ from premisses:
    \begin{align}
    &\emptyset \howeimp \qvaltwo \leq 
    \howe{\vrelone}(\bang \valone, \valthree): \bang_{\baseone}\typeone, 
    \label{key-lemma:pmbang-premise-1}\\
    &\varone :_{\basetwo \comp \baseone} \typeone \howeimp 
    \qvalthree \leq \howe{\vrelone}(\termone, \termone'):\typetwo,
    \label{key-lemma:pmbang-premise-2}\\
    &\qvalfour \leq \vrelone_{\typetwo}
    (\pmbang{\valthree}{\termone'}, \termtwo).
    \label{key-lemma:pmbang-premise-3}
    \end{align}
    In particular, we have 
    $\qvalone = \basetwo(\qvaltwo) \tensor \qvalthree \tensor \qvalfour$.
    Let us examine premise \eqref{key-lemma:pmbang-premise-1}. 
    First of all, since $\valthree$ is a closed value of type 
    $\bang_{\baseone} \typeone$ it must be 
    of the form $\bang \valone'$. Moreover, 
    \eqref{key-lemma:pmbang-premise-1} must have been inferred via 
    an instance of rule $\ruleHBang$ from premisses:
    \begin{align}
    &\emptyset \howeimp \qvaltwo_1 \leq 
    \howe{\vrelone}(\valone, \valtwo): \typeone,
    \label{key-lemma:pmbang-premise-1.1}\\
    &\qvaltwo_2 \leq 
    \vrelone_{\bang_{\baseone}\typeone}
    (\bang \valtwo, \bang \valone').
    \label{key-lemma:pmbang-premise-1.2}
    \end{align}
    In particular, $\qvaltwo = \baseone(\qvaltwo_1) \tensor \qvaltwo_2$.
    From \eqref{key-lemma:pmbang-premise-1.2}, by definition of 
    applicative $\vrelator$-simulation we infer 
    $\qvaltwo_2 \leq \baseone \acts \vrelone_{\typeone}(\valtwo, \valone')$. 
    Since \eqref{key-lemma:pmbang-premise-1.1} implies 
    $\qvaltwo_1 \leq \howe{\vrelone}_{\typeone}(\valone, \valtwo)$, we have:
    \begin{align*}
    \qvaltwo 
    &= 
    \baseone(\qvaltwo_1) \tensor \qvaltwo_2 \\
    &\leq 
    \baseone \acts \howe{\vrelone}_{\typeone}(\valone, \valtwo) 
    \tensor \baseone \acts \vrelone_{\typeone}(\valtwo, \valone') \\
    &\leq
    \baseone \acts (\howe{\vrelone}_{\typeone}(\valone, \valtwo) 
    \tensor \vrelone_{\typeone}(\valtwo, \valone')) \\
    &\leq
    \baseone \acts \howe{\vrelone}_{\typeone}(\valone, \valone'),
    \end{align*}
    where the last inequality follows by pseudo-transitivity.
    From \eqref{key-lemma:pmbang-premise-2} we infer the inequality
    $\qvalthree \leq \varone:_{\basetwo \comp \baseone} \typeone
    \imp \howe{\vrelone}(\termone, \termone'): \typetwo$. 
    We are now in position to apply the Substitution Lemma
    obtaining:
    $$
    (\basetwo \comp \baseone) \acts 
    \howe{\vrelone}_{\typeone}(\valone,\valone') \tensor 
    \qvalthree \leq \howe{\vrelone}_{\typetwo}
    (\termone[\varone := \valone],\termone'[\varone := \valone']).
    $$
    The latter, together with the inequality
    $\qvaltwo \leq \baseone \acts \howe{\vrelone}_{\typeone}(\valone, \valone')$, 
    implies 
    $\basetwo(\qvaltwo) \tensor \qvalthree 
    \leq 
    \howe{\vrelone}_{\typetwo}
    (\termone[\varone := \valone],\termone'[\varone := \valone']).$ 
    Applying the induction hypothesis we conclude:
    $$
    \basetwo(\qvaltwo) \tensor \qvalthree 
    \leq
    \vrelator\howe{\vrelone}_{\typetwo}
    (\sem{\termone[\varone := \valone]}_n,\sem{\termone'[\varone := \valone']})
    .
    $$
    Finally, from \eqref{key-lemma:pmbang-premise-3} by definition of 
    applicative $\vrelator$-simulation we infer 
    $\qvalfour \leq \vrelator \vrelone_{\typetwo}
    (\sem{\termone'[\varone := \valone']}, \sem{\termtwo})$ 
    (recall that $\valthree = \bang \valone'$) and thus 
    conclude the thesis by $\vrelator$-pseudo-transitivity.
  \item
    We have to prove:
    $$
    \howe{\vrelone}_{\typetwo}(\seq{\termone}{\termtwo}, \termthree) \leq 
    \vrelator \howe{\vrelone}_{\typetwo}
    (\approxsem{\seq{\termone}{\termtwo}}{n+1}, \sem{\termthree}).
    $$
    As 
    $\approxsem{\seq{\termone}{\termtwo}}{n+1} = 
    \kleisli{\fun{\termtwo}_n}\approxsem{\termone}{n}$,
    it is sufficient to prove that for any $\qvalone$ such that 
    $\emptyctx \howeimp \qvalone \leq \howe{\vrelone}
    (\seq{\termone}{\termtwo}, \termthree): \typetwo$ 
    is derivable, we have
    $
    \qvalone \leq \vrelator \howe{\vrelone}_{\typetwo}
    (\kleisli{\fun{\termtwo}_n}\approxsem{\termone}{n}, 
    \sem{\termthree})
    .$
    Suppose $\emptyctx \howeimp \qvalone \leq \howe{\vrelone}
    (\seq{\termone}{\termtwo}, \termthree): \typetwo$. The 
    latter must have been inferred via an instance of rule 
    $\ruleHSeq$ from premisses:
    \begin{align}
    &\emptyctx \howeimp \qvaltwo \leq 
    \howe{\vrelone}(\termone, \termone'): \typeone, 
    \label{key-lemma:seq-premise-1}\\
    &\varone :_{\baseone} \typeone \howeimp 
    \qvalthree \leq \howe{\vrelone}(\termtwo, \termtwo'):\typetwo,
    \label{key-lemma:seq-premise-2}\\
    &\qvalfour \leq \vrelone_{\typetwo}
    (\seq{\termone'}{\termtwo'}, \termthree).
    \label{key-lemma:seq-premise-3}
    \end{align}
    In particular, we have 
    $\qvalone = (\baseone \wedge \baseid)(\qvaltwo) 
    \tensor \qvalthree \tensor \qvalfour$.
    We now claim to have:
    \begin{multline}        
    (\varone :_\baseone \typeone \imp \howe{\vrelone}
    (\termtwo,\termtwo'): \typetwo)
    \tensor
    (\baseone \wedge \baseid) \acts \howe{\vrelone}_{\typeone}
    (\termone, \termone') 
    \\ \leq 
     \vrelator\howe{\vrelone}_{\typetwo}(\sem{\seq{\termone}{\termtwo}}_{n+1},
    \sem{\seq{\termone'}{\termtwo'}}).
    \label{key-lemma:seq-claim}
    \end{multline}
    By very definition of Howe's extension, the latter obviously entails 
    $
    (\baseone \wedge \baseid)(\qvaltwo) 
    \tensor \qvalthree
    \leq 
    \vrelator\howe{\vrelone}_{\typetwo}(\sem{\seq{\termone}{\termtwo}}_{n+1},
    \sem{\seq{\termone'}{\termtwo'}})$. Moreover, by 
     definition of applicative $\vrelator$-simulation,
    \eqref{key-lemma:seq-premise-3} implies
    $\qvalfour \leq 
    \vrelator\vrelone_{\typetwo}
    (\sem{\seq{\termone'}{\termtwo'}}, \sem{\termthree})$, 
    which allows to conclude the thesis by $\vrelator$-pseudo-transitivity. 
    Let us now turn to the proof of \eqref{key-lemma:seq-premise-3}. 
    First of all we apply the induction hypothesis on 
    $\howe{\vrelone}_{\typeone}(\termone, \termone')$. By monotonicity 
    of $\baseone \wedge \baseid$ we have thus reduced the proof of 
    \eqref{key-lemma:seq-premise-3} to proving the inequality:
    \begin{multline}
    (\varone :_\baseone \typeone \imp \howe{\vrelone}
    (\termtwo,\termtwo'): \typetwo) 
    \tensor
    (\baseone \wedge \baseid) \acts \vrelator\howe{\vrelone}_{\typeone}
    (\sem{\termone}_n, \sem{\termone'})
    \\ \leq 
    \vrelator\howe{\vrelone}_{\typetwo}
    (\kleisli{\fun{\termtwo}_n}\approxsem{\termone}{n},
    \kleisli{\fun{\termtwo'}}\sem{\termone'}).
    \label{key-lemma:seq-claim-1}
    \end{multline}
    Consider the diagram:
    \begin{align}
    \vcenter{
    \xymatrixcolsep{4pc}\xymatrix{
    \laxcommute
    \catunit \times \monad \values_{\typeone} 
    \ar[d]_{\vrelthree \tensor (\baseone \wedge \baseid) \acts 
    \vrelator \howe{\vrelone}_{\typeone}}|-*=0@{|}   
    \ar[r]^-{\kleisli{\fun{\termtwo}_n} \comp 
    \unitor_{\monad\values_{\typeone}}}  
    &  
    \monad \values_{\typetwo}  
    \ar[d]^{\relator \howe{\vrelone}_{\typetwo}}|-*=0@{|} 
    \\
    \catunit \times  \monad \values_{\typeone}
    \ar[r]_-{\kleisli{\fun{\termtwo'}} \comp 
    \unitor_{\monad \values_{\typeone}}}   
    &  
    \monad \values_{\typetwo}}} 
    \label{key-lemma:square-1}
    \end{align}
    where $\catunit = \{*\}$ and 
    $\vrelthree(*,*) = (\varone :_\baseone \typeone \imp \howe{\vrelone}
    (\termtwo,\termtwo'): \typetwo)$. 
    It is easy to see that \eqref{key-lemma:seq-claim-1} 
    follows from \eqref{key-lemma:square-1}, 
    since e.g.:
    $$
    (\kleisli{\fun{\termtwo}}_n \comp 
    \unitor_{\monad \values_{\typeone}})
    (*, \sem{\termone}_n) 
    = 
    \kleisli{\fun{\termone}}_{n} \sem{\termone}_n.
    $$
    To prove \eqref{key-lemma:square-1} we first observe that 
    by very definition of strong monad 
    we have
    $\kleisli{\fun{\termtwo}}_{n} \comp 
    \unitor_{\monad\values_{\typeone}} = 
    \strongkleisli{(\fun{\termtwo}_{n} \comp
        \unitor_{\values_{\typeone}})}$.
    We can now apply condition \eqref{s-Strong-Lax-Bind}. 
    As a consequence, 
    to prove \eqref{key-lemma:square-1} it is sufficient to 
    prove
    that for all closed values $\valone, \valtwo$ of type $\typeone$,
    we have:
    \begin{multline*}
    (\varone :_\baseone \typeone \imp \howe{\vrelone}
    (\termtwo,\termtwo'): \typetwo)
    \tensor 
    (\baseone \wedge \baseid) \acts \howe{\vrelone}_{\typeone}
    (\valone, \valtwo)
    \\ \leq 
    \vrelator \howe{\vrelone}_{\typetwo}
    (\approxsem{\substcomp{\termtwo}{\varone}{\valone}}{n}, 
    \sem{\substcomp{\termtwo'}{\varone}{\valtwo}}).
    \end{multline*}
    By Substitution Lemma and induction hypothesis we 
    have: 
    \begin{multline*}
    (\varone :_\baseone \typeone \imp \howe{\vrelone}
    (\termtwo,\termtwo'): \typetwo)
    \tensor
    \baseone \acts \howe{\vrelone}_{\typeone}
    (\valone, \valtwo)
    \\ \leq 
    \vrelator \howe{\vrelone}_{\typetwo}
    (\approxsem{\substcomp{\termtwo}{\varone}{\valone}}{n}, 
    \sem{\substcomp{\termtwo'}{\varone}{\valtwo}}).
    \end{multline*}
    We conclude the thesis since $\baseone \wedge \baseid \leq \baseone$. 
  \item
    We have to prove: 
    $$
    \howe{\vrelone}_{\typeone}(\op(\termone_1, \hh, \termone_m), \termtwo) 
    \leq \vrelator \howe{\vrelone}_{\typeone}
    (\approxsem{\op(\termone_1, \hh, \termone_m)}{n+1}, \sem{\termtwo}),
    $$
    where $\op$ is an $m$-ary operation symbol in $\signature$. 
    As usual, we use the notation $\vec{x_i}$ for items $x_1, \hh, x_m$.\\
    We show that for any $\qvalone$ such that 
    $\emptyctx \howeimp \qvalone \leq \howe{\vrelone}
    (\op(\vec{\termone_i}), \termtwo): \typeone$ 
    is derivable, 
    $
    \qvalone \leq \vrelator \howe{\vrelone}_{\typetwo}
    (\approxsem{\op(\vec{\termone_i})}{n}, 
    \sem{\termtwo})
    $ holds. 
    Suppose to have $\emptyctx \howeimp \qvalone \leq \howe{\vrelone}
    (\op(\vec{\termone_i}), \termtwo): \typetwo$. The 
    latter must have been inferred via an instance of rule 
    $\ruleHOp$ from premisses:
    \begin{align}
    &\forall i\leq m.\ \emptyctx \howeimp \qvalone_i \leq 
    \howe{\vrelone}(\termone_i, \termtwo_i): \typeone,
    \label{key-lemma:op-premise-1}\\
    &\qvaltwo \leq \vrelone_{\typetwo}
    (\op(\termtwo_1, \hh, \termtwo_m), \termtwo).
    \label{key-lemma:op-premise-2}
    \end{align}
    In particular, we have 
    $\qvalone = \qop(\qvalone_1, \hh, \qvalone_m) \tensor \qvaltwo$. 
    We apply the induction hypothesis on \eqref{key-lemma:op-premise-1} 
    obtaining, for each $i \leq m$, the inequality
    $\qvalone_i \leq \vrelator \howe{\vrelone}
    (\sem{\termone_i}_n, \sem{\termtwo_i})$. By monotonicity 
    of $\qop$ we thus infer:
    \begin{align*}
    \qop(\vec{\qvalone_i}) 
    &\leq 
    \qop
    (\vrelator \howe{\vrelone}(\sem{\termone_1}_n, \sem{\termtwo_1}), 
    \hh,
    \vrelator \howe{\vrelone}(\sem{\termone_m}_n, \sem{\termtwo_m})) \\
    &\leq
    \vrelator \howe{\vrelone_\typeone}
    (\mop_{\values_\typeone}(\sem{\termone_1}_n, \hh, \sem{\termone_m}_n), 
    \mop_{\values_\typeone}(\sem{\termtwo_1}, \hh, \sem{\termtwo_m})) \\
    &= 
    \vrelator \howe{\vrelone_\typeone}
    (\sem{\op(\termone_1, \hh, \termone_m)}_{n+1}, 
    \sem{\op(\termtwo_1, \hh, \termtwo_m)}),
    \end{align*}
    where the second inequality follows since $\vrelator$ is $\signature$-compatible. 
    We conclude the thesis from \eqref{key-lemma:op-premise-2} 
    by $\vrelator$-pseudo-transitivity and definition of applicative 
    $\vrelator$-simulation.
\end{enumerate}
\end{proof}

\subsection{Applicative $\vrelator$-bisimilarity}
\label{appendix:proofs-applicative-v-bisimilarity}
In this last section we expand on some technical details necessary to prove that 
applicative $\vrelator$-bisimilarity is compatible.

\symmetricSimilarityIsBisimilarity*
\begin{proof}
Obviously $\vbisim$ is an applicative 
$\vrelator$-simulation. Moreover, 
$\vbisim$ is symmetric and thus we have $\vbisim \leq \vbisim'$. To see 
that $\vbisim' \leq \vbisim$ it is sufficient to prove that $\vbisim'$ is 
an applicative $(\vrelator \wedge \dual{\vrelator})$-simulation. 
Clauses on values are trivially satisfied. We now show that for any 
symmetric applicative 
$\vrelator$-simulation $\vrelone$, we have the inequality
$\toterm{\vrelone}_\typeone(\termone, \termone') \leq 
\vrelator\toval{\vrelone}_\typeone(\sem{\termone}, \sem{\termone'}) \wedge 
\vrelator \dual{(\toval{\vrelone}_\typeone)}(\sem{\termone'}, \sem{\termone})$ 
for all terms $\termone, \termone' \in \Lambda_\typeone$. 
For that it is 
sufficient to prove $\toterm{\vrelone}_\typeone(\termone, \termone') \leq 
\vrelator\dual{(\toval{\vrelone}_\typeone)}(\sem{\termone'}, \sem{\termone})$, 
which obviously holds since $\vrelone$ is symmetric.
\end{proof}

\transitiveClosureHoweExtensionCompatible*

\begin{proof}
We start with point $1$.
First of all observe that by Lemma \ref{lemma:properties-howe-extension} 
$\howe{\vrelone}$ is compatible. To prove compatibility of
$\transitive{(\howe{\vrelone})}$ we have to check that it satisfies all 
clauses in Figure \ref{fig:compatibility-clauses}. We show the case 
for sequential composition as an illustrative example (the other cases 
are proved in a similar, but easier, way). 
We have to prove:
\begin{multline*}
(\baseone \wedge \baseid) \acts 
(\envone \imp \transitive{(\howe{\vrelone})}(\termone, \termone'): \typeone) 
\tensor 
(\envtwo, \varone :_\baseone \typeone \imp 
\transitive{(\howe{\vrelone})}(\termtwo, \termtwo'): \typetwo)
\\
\leq 
(\baseone \wedge \baseid) \comp \envone \imp \transitive{(\howe{\vrelone})}
(\seq{\termone}{\termtwo},\seq{\termone'}{\termtwo'}): \typetwo.
\end{multline*}
Let $\qvalthree \defeq ((\baseone \wedge \baseid) \comp \envone \imp 
\transitive{(\howe{\vrelone})}
(\seq{\termone}{\termtwo},\seq{\termone'}{\termtwo'}): \typetwo)$. 
By definition of transitive closure we have to prove:
\begin{multline*}
(\baseone \wedge \baseid) \acts 
\join_n (\envone \imp (\howe{\vrelone})^{(n)}(\termone, \termone'): \typeone) 
\\ \tensor 
\join_m (\envtwo, \varone :_\baseone \typeone \imp 
(\howe{\vrelone})^{(m)}(\termtwo, \termtwo'): \typetwo)
\leq 
\qvalthree.
\end{multline*}
By finite continuity either $\baseone \wedge \baseid = \infty$ or it 
is continuous with respect to joints. In the former case we are trivially 
done. So suppose the latter case, so that thesis becomes:
\begin{multline*}
\join_n (\baseone \wedge \baseid) \acts 
 (\envone \imp (\howe{\vrelone})^{(n)}(\termone, \termone'): \typeone) 
\\ \tensor 
\join_m (\envtwo, \varone :_\baseone \typeone \imp 
(\howe{\vrelone})^{(m)}(\termtwo, \termtwo'): \typetwo)
\leq 
\qvalthree.
\end{multline*}
In particular, we also have $\baseone \neq \infty$. We prove that for 
any $n,m \geq 0$ the following holds: for all $\termone, \termone', \termtwo, \termtwo'$ 
(of appropriate type),
\begin{multline*}
(\baseone \wedge \baseid) \acts 
 (\envone \imp (\howe{\vrelone})^{(n)}(\termone, \termone'): \typeone) 
\tensor 
(\envtwo, \varone :_\baseone \typeone \imp 
(\howe{\vrelone})^{(m)}(\termtwo, \termtwo'): \typetwo)
\\
\leq 
((\baseone \wedge \baseid) \comp \envone \imp \transitive{(\howe{\vrelone})}
(\seq{\termone}{\termtwo},\seq{\termone'}{\termtwo'}): \typetwo) 
\end{multline*}
holds.
First of all we observe that since $\howe{\vrelone}$ is reflexive, we can 
assume $n = m$. In fact, if e.g. $n = m + l$, then we can `complete' 
$(\howe{\vrelone})^{(m)}$ as follows: 
$$
(\howe{\vrelone})^{(m)} = 
(\howe{\vrelone})^{(m)} \underbrace{\comp \idvrel \cc \comp \idvrel}_{l\text{-times}} 
\leq (\howe{\vrelone})^{(m)} 
\underbrace{\comp \howe{\vrelone} \cc \comp \howe{\vrelone}}_{l\text{-times}}
= (\howe{\vrelone})^{(n)}.
$$
We now do induction on $n$. The base case is trivial. Let us turn on the 
inductive step. We have to prove:
\begin{multline*}
(\baseone \wedge \baseid) \acts 
\Big( 
\join_{\termone''}(\envone \imp \howe{\vrelone}(\termone, \termone''): \typeone)
\tensor (\envone \imp (\howe{\vrelone})^{(n)}(\termone'', \termone'): \typeone) 
\Big )
\\
\tensor 
\join_{\termtwo''}
(\envtwo, \varone :_\baseone \typeone \imp 
\howe{\vrelone}(\termtwo, \termtwo''): \typetwo)
\\ \tensor
(\envtwo, \varone :_\baseone \typeone \imp 
(\howe{\vrelone})^{(n)}(\termtwo'', \termtwo'): \typetwo)
\leq \qvalthree.
\end{multline*}
Since $\baseone \wedge \baseid$ is continuous it is 
sufficient to prove that for all terms $\termone'', \termtwo''$ 
we have: 
\begin{multline*}
(\baseone \wedge \baseid) \acts 
(\envone \imp \howe{\vrelone}(\termone, \termone''): \typeone)
\tensor 
(\baseone \wedge \baseid) \acts 
(\envone \imp (\howe{\vrelone})^{(n)}(\termone'', \termone'): \typeone) 
\\
\tensor 
(\envtwo, \varone :_\baseone \typeone \imp 
\howe{\vrelone}(\termtwo, \termtwo''): \typetwo)
\tensor
(\envtwo, \varone :_\baseone \typeone \imp 
(\howe{\vrelone})^{(n)}(\termtwo'', \termtwo'): \typetwo)
\leq \qvalthree,
\end{multline*}
i.e. 
\begin{multline*}
(\baseone \wedge \baseid) \acts 
(\envone \imp \howe{\vrelone}(\termone, \termone''): \typeone)
\tensor 
(\envtwo, \varone :_\baseone \typeone \imp 
\howe{\vrelone}(\termtwo, \termtwo''): \typetwo) 
\\
\tensor
(\baseone \wedge \baseid) \acts 
(\envone \imp (\howe{\vrelone})^{(n)}(\termone'', \termone'): \typeone) 
\\ \tensor
(\envtwo, \varone :_\baseone \typeone \imp 
(\howe{\vrelone})^{(n)}(\termtwo'', \termtwo'): \typetwo)
\leq \qvalthree.
\end{multline*}
We can now apply compatibility of $\howe{\vrelone}$ plus the induction 
hypothesis, thus reducing the thesis to:
\begin{multline*}
\Big ((\baseone \wedge \baseid) \comp \envone \tensor \envtwo \imp
\howe{\vrelone}(\seq{\termone}{\termtwo},\seq{\termone''}{\termtwo''}): \typeone)
\Big )
\\
\tensor 
\Big ( (\baseone \wedge \baseid) \comp \envone \tensor \envtwo \imp
\transitive{(\howe{\vrelone})}
(\seq{\termone''}{\termtwo''},\seq{\termone'}{\termtwo}): \typeone)
\Big)
\leq \qvalthree.
\end{multline*}
We can now conclude the thesis by very definition of 
$\transitive{(\howe{\vrelone})}$.

To prove point $2$ we have to show 
$\transitive{(\howe{\vrelone})} 
\leq \dual{(\transitive{(\howe{\vrelone})})}$. For that it is 
sufficient to show 
$\howe{\vrelone} 
\leq \dual{(\transitive{(\howe{\vrelone})})}.$
That amounts to prove that for all terms 
$\envone \imp \termone, \termone' : \typeone$ 
and values $\envone \valimp \valone, \valone'$, and for any 
$\qvalone \in \quantale$ such that 
$\envone \howeimp \qvalone \leq \howe{\vrelone}(\termone, \termone'): \typeone$ 
is derivable we have 
$\qvalone \leq \envone \imp \transitive{(\howe{\vrelone})}(\termone, \termone'):\typeone$
(and similarity for $\envone \valimp \valone, \valone': \typeone$). 
The proof is by induction on the derivation of 
$\envone \howeimp \qvalone \leq \howe{\vrelone}(\termone, \termone'): \typeone$ 
using point $1$.
\end{proof}

\end{document}